\documentclass[11pt]{article}

\usepackage{xcolor}
\usepackage{pagecolor}
\usepackage{lipsum}  
\usepackage{mdframed}
\usepackage[export]{adjustbox}

\usepackage{amsthm}
\usepackage{amsmath}
\usepackage{amssymb}
\usepackage{bbm}
\usepackage{commath}

\usepackage{hyperref}
\usepackage{todonotes}

\newtheorem{theorem}{Theorem}[section]
\newtheorem{corollary}{Corollary}[section]
\newtheorem{lemma}{Lemma}[section]
\newtheorem{definition}{Definition}
\newtheorem{proposition}{Proposition}[section]
\newtheorem{assumption}{Assumption}[section]
\newtheorem{remark}{Remark}[section]

\usepackage{hyperref}
\usepackage{cleveref}
\Crefname{assumption}{Assumption}{Assumptions}

\newcommand{\MH}{\mathrm{MH}}

\usepackage{booktabs}
\usepackage[scale=2]{ccicons}

\usepackage{fancyhdr}
\usepackage{algpseudocode}
\usepackage{algorithm}
\usepackage{tikz}
\usetikzlibrary{arrows.meta}
\usepackage{setspace}

\usepackage[a4paper, left=3cm, right=3cm, top=2.5cm, bottom=3cm]{geometry}

\fancypagestyle{firstpage}{
  \lhead{}
  \rhead{}
}

\title{\vspace{-1cm} \textbf{Weak Poincar{\'{e}} inequalities for two-component
Deterministic-scan Metropolis-within-Gibbs samplers}\\
        }

\author{\href{mailto:Mengxi.Gao@warwick.ac.uk}{Mengxi Gao }\\University of Warwick
\and \href{mailto:gareth.o.roberts@warwick.ac.uk}{Gareth O. Roberts}\\University of Warwick
\and \href{mailto:andi.wang@warwick.ac.uk}{Andi Q. Wang} \\University of Warwick}

\date{\vspace{-5ex}} 
\providecommand{\keywords}[1]{\textbf{\textit{Keywords: }} #1}

\begin{document}
\pagenumbering{roman}
\maketitle

\thispagestyle{firstpage}
\pagestyle{firstpage}

\setlength{\parskip}{0.8em}
\setlength{\parindent}{0pt}

\begin{abstract}
Using the framework of weak Poincar{\'{e}} inequalities, we  analyze the convergence properties of deterministic-scan Metropolis-within-Gibbs samplers, an important class of Markov chain Monte Carlo algorithms. 
Our analysis applies to nonreversible Markov chains and yields explicit (subgeometric) convergence bounds through novel comparison techniques based on Dirichlet forms.
We show that the joint chain inherits the convergence behavior of the marginal chain and conversely.
In addition, we establish several fundamental results for weak Poincar{\'{e}} inequalities for discrete-time Markov chains, such as a tensorization property for independent chains.
We apply our theoretical results through applications to algorithms for Bayesian inference for a hierarchical regression model and a diffusion model under discretely-observed data.
\end{abstract}


\pagenumbering{arabic}
\setcounter{page}{1}

\section{Introduction}
\label{sec:introduction}
\subsection{Motivation}
Gibbs samplers are a class of Markov Chain Monte Carlo (MCMC) algorithms that are routinely used to sample from complex distributions \cite{gelfand1990sampling}. These samplers iteratively sample from the conditional distributions of each component given the others, with two important classes being \textit{random-scan} and \textit{deterministic-scan} samplers. Random-scan Gibbs samplers update one component at each iteration, selected at random, while  deterministic-scan Gibbs samplers update the components sequentially in a fixed order.

In many practical applications, the conditional distributions required by the Gibbs sampler are not easily sampled from. 
To address this limitation, Metropolis-within-Gibbs (MwG) samplers have been widely developed, see, for example \cite{besag1993spatial}. In such settings, one employs a collection of Markov kernel operators targeting the conditional distributions, rather than exact sampling. In principle, any suitable MCMC method can be used within this framework, and by convention the resulting algorithm is still referred to as MwG, even when the inner update is not a Metropolis--Hastings (MH) step.

The quantitative theories available for standard Gibbs samplers have been studied intensely in the literature; see, for example, \cite{barone1990improving, liu1994covariance, roberts1997updating, gaitonde2024comparison, ascolani2024entropy, chlebicka2025solidarity, goyal2025mixing}.
Given the popularity and practical success of MwG samplers, progress has also been made in developing analytical tools to study their convergence properties; for \textit{random-scan} MwG samplers, see \cite{ qin2023spectral, ascolani2024entropy,secchi2025spectral}, with some related insights found in \cite{roberts1997geometric, fort2003geometric, ascolani2024scalability}. The relationship between geometric ergodicity of random-scan and deterministic-scan MwG samplers has also been investigated in \cite{johnson2013component, qin2022convergence}.

Moving to \textit{deterministic-scan} MwG samplers, there are two sources of analytic difficulty. Firstly, the deterministic scan induces \textit{nonreversibility}, rendering spectral analysis substantially more complex. Secondly, Metropolis-within-Gibbs-type transition kernels are mixtures of densities with respect to potentially different dimensional dominating measures. In particular, the mixing of kernels with incompatible supports complicates the verification of basic Markov chain properties, which hold under much simpler conditions for pure Gibbs or full-dimensional MH schemes. Indeed, \cite{roberts2006harris} shows that even in a two-dimensional MwG algorithm, fundamental properties like Harris recurrence may fail.

However, some progress has been made based on \textit{comparison} techniques. For instance, \cite{jones2014convergence,qin2022convergence} provide sufficient conditions for the geometric convergence of deterministic-scan MwG samplers. For two-component deterministic-scan MwG samplers, the latter-updated component induces a \textit{marginal} Markov chain, obtained by integrating out the other component. A related line of work studies this marginal chain, see, for example, \cite{qin2022convergence}. 
In addition, \textit{hybrid slice samplers}, which can be seen as a marginal version of a particular deterministic-scan MwG sampler, have also been analyzed in \cite{power2024weak}. More recently, \cite{qin2025spectral} proposed a spectral gap decomposition framework for a class of MwG algorithms. Despite these advances, existing approaches \cite{jones2014convergence, qin2022convergence} largely depend on geometric ergodicity or explicit spectral gap conditions, which may fail for practically-relevant MwG samplers \cite{andrieu2015quantitative,durmus2016subgeometric}. 

The aim of this work is to fill this gap by developing a general and flexible framework based on \textit{weak Poincar{\'{e}} inequalities} (WPIs) \cite{rockner2001weak, andrieu2023weak}. By utilising WPIs within a comparison framework via Dirichlet forms -- see \cite{diaconis1993comparison,  andrieu2022comparison, power2024weak, Caprio2025} -- we will analyze the convergence behavior of nonreversible two-component deterministic-scan MwG samplers. This approach allows us to derive quantitative convergence bounds under mild assumptions, with a particular focus on the \textit{subgeometric} setting where spectral gaps do not exist. 

\subsection{Main result and structure}
First, we study the convergence behavior of two-component deterministic-scan MwG samplers under the assumption that the convergence properties of the standard deterministic-scan Gibbs sampler and those of the conditional MCMC samplers are known. More specifically, the convergence bound of the deterministic-scan MwG sampler can be characterized as a particular composition of the convergence bounds of the standard Gibbs sampler and its corresponding conditional update kernels; see \Cref{thm:weakPtotildeP}. This analysis requires us to extend the WPI comparison framework into the nonreversible setting, extending previous works which have largely required reversibility \cite{andrieu2022comparison,power2024weak, Caprio2025}.

Next, we demonstrate that the convergence bounds of the deterministic-scan Gibbs sampler and its associated marginal chain are closely related, see \Cref{thm:PtoP_X} and \Cref{thm:P_XtoP}. This provides a way to obtain improved convergence bounds for such MwG samplers, see \Cref{thm:weakPtotildeP2}.

As a byproduct of these investigations, we also establish several fundamental theoretical results regarding weak Poincar{\'{e}} inequalies for Markov chains, including a characterization of the relationship between the convergence of a Markov chain and that of its adjoint in the subgeometric setting, \Cref{Thm:K*:PtoP*}, as well as a \textit{tensorization} property for weak Poincar{\'{e}} inequalities, \Cref{thm: tensor product}, which enables the straightforward analysis of product chains.

This paper is organized as follows. In \Cref{Sec: Preliminary}, we introduce the necessary theoretical background on Gibbs samplers and weak Poincar{\'{e}} inequalities. In \Cref{Sec: Comparison for Gibbs}, we develop a series of new comparison results between two-component standard Gibbs and MwG samplers. Viewing independent chains on tensor product spaces as a special case of MwG samplers, we establish more general convergence results by deriving a tensorization property of weak Poincar{\'{e}} inequalities in \Cref{Sec: tensor product}. In \Cref{Sec: toy example}, we apply our results to Normal--inverse-gamma distributions, demonstrating the comparison results and convergence bounds for specific MwG samplers. A similar analysis is carried out for a Bayesian hierarchical model in \Cref{Sec: bayes hierarchical model}. Finally, in \Cref{Sec: diffusion}, we present convergence results for sampling from discretely-observed diffusion processes. 

\subsection{Related work}
Until fairly recently, most progress on subgeometric rates of convergence for Markov chains, particularly with a view to MCMC application, was obtained using Lyapunov drift conditions, for example see \cite{tuominen1994subgeometric,jarner2002polynomial,jarner2007convergence,fort2003polynomial,roberts2023polynomial}, and also the recent \cite{brevsar2025central} for \textit{lower} bounds. Although these techniques are very general, quantitative bounds obtained in this way tend to be extremely conservative, especially in high-dimensional problems.

The work \cite{qin2025spectral} introduced new spectral gap decomposition techniques for several hybrid Gibbs algorithms. Notably, one of their key examples is a hybrid data augmentation algorithm involving two intractable conditional distributions, which coincides with the setting considered in this work. Their focus is primarily on spectral decomposition, whereas our analysis is designed to handle convergence in cases which may or may not have positive spectral gaps. Although they do consider settings where spectral gaps are not uniformly bounded and discuss a decomposition approach based on WPIs, these techniques are not applied to their specific examples, which would require overcoming several additional technical difficulties. For example, the Dirichlet forms associated with operator-adjoint compositions (denoted in their paper as $\mathcal{E}(P^*E^*EP,f)$ and $\mathcal{E}(E^*P^*PE,f)$) are not identical, even when the associated right spectral gaps are the same, which limits the applicability of their Proposition~22 to the case of two intractable conditional distributions. In contrast, our work directly addresses this issue by introducing \Cref{Thm:K*:PtoP*}, providing a more robust framework for analyzing convergence in the absence of strong spectral properties. Moreover, we present the complete framework to derive an explicit convergence bound for two-component deterministic-scan MwG samplers.

The contributions \cite{liu1994covariance,qin2022convergence,ascolani2025mixing} have shown that the convergence properties of the marginal chain are closely related to those of deterministic-scan Gibbs kernels in the two-component setting. It is well-established that, in geometrically ergodic cases, the convergence rate of the marginal chain is essentially identical to that of the original deterministic-scan Gibbs chain, as demonstrated through operator norm arguments in \cite{liu1994covariance, qin2022convergence}. Moreover, \cite[Lemma 3.2]{qin2022convergence} extends this relationship to MwG kernels and their corresponding marginal chains.
 \cite[Proposition 7.1]{ascolani2025mixing} further provides a more general result based on entropy contraction, showing that the distance between the $t$-step distribution of the deterministic-scan Gibbs sampler and its target can be bounded by the corresponding distances for the $t$-step and $(t+1)$-step marginal chains. In our work, we establish a similar connection within the weak Poincar{\'{e}} 
 inequality framework, focusing on $\mathrm{L}^2$-based convergence analysis and \textit{subgeometric} rates.

\section{Preliminaries}
\label{Sec: Preliminary}
\subsection{Markov chain theory}
Let $(\mathsf{E}, \mathcal{F})$ be a measurable space equipped with a probability measure $\mu$. For a measurable function $f: \mathsf{E} \rightarrow \mathbb{R}$, define $\mu(f)=\int_\mathsf{E} f(x)\mu(\dif x)$.

Let $\mathrm{L}^2(\mu)$ denote the Hilbert space of the set of real-valued measurable functions that are square integrable with respect to $\mu$, equipped with the inner product $\langle f, g\rangle_\mu=\int_{\mathrm{E}} f(x) g(x) \mathrm{d} \mu(x)$, for $f, g \in \mathrm{L}^2(\mu)$, let $\|f\|^2_\mu = \langle f, f\rangle_{\mu}$, and  write $\mathrm{L}^2_0(\mu)$ as the set of functions $f \in \mathrm{L}^2(\mu)$ satisfying $\mu(f) = 0$. 

Let $T$ be a Markov transition kernel, and define the operator $(Tf)(x)=\int_\mathsf{E} T(x,\dif  x')f(x')$, and $T^*$ as the adjoint operator of $T$.
Recall that a Markov transition kernel $T$ on a $\sigma$-finite measure space $(\mathsf{E},\mathcal{F},\mu)$ is $\mu$-reversible if $\mu(\dif x)T(x,\dif y)=\mu(\dif y)T(y,\dif x)$ as measures on $(\mathsf{E}\times \mathsf{E}, \mathcal{F} \otimes \mathcal{F})$. Moreover, a $\mu$-reversible Markov kernel $T$ is said to be positive if $\langle Tf, f \rangle_\mu \geqslant 0$ for all $f \in \mathrm{L}^2(\mu)$. We say $T$ is $\mu$-invariant if $\int _\mathsf{E}\mu(\dif x)T(x, A)=\mu(A)$, for any $ A \in \mathcal{F}$. For a $\mu$-invariant Markov transition kernel $T$, define the Dirichlet form by
\begin{align}
\label{equ:DirichletForm}
\begin{split}
    {\mathcal{E}}_\mu(T,f)&:=\langle (I-T)f,f \rangle_\mu
    =\frac{1}{2}\int_{\mathsf{E}}\int_{\mathsf{E}} (f(x)-f(y))^2T(x,\dif y)\mu(\dif x).
\end{split}
\end{align}
$T$ is said to be $ \mathrm{L}^2(\mu)$-exponentially convergent if there exists $\alpha \in [0,1)$ such that for any $f \in \mathrm{L}^2_0(\mu)$ and all $n \in \mathbb{N}$,
\begin{align}
\label{equ:exp}
    \|T^nf\|^2_\mu \leqslant \|f\|_\mu^2 \alpha^n.
\end{align}
The infimum of such $\alpha$ is the $\mathrm{L}^2$ convergence rate of the Markov chain associated with $T$.

For a measurable function $f : \mathsf{E}\rightarrow \mathbb{R}$, we define also the oscillation seminorm by
\begin{align*}
    \|f\|_\mathrm{osc}:=\text{ess}_\mu\sup f-\text{ess}_\mu\inf f,
\end{align*}
where $\text{ess}_\mu\sup f := \inf\{a \in \mathbb{R}: \mu(f^{-1}(a,\infty))=0\}$, $\text{ess}_\mu\inf f := \sup\{b \in \mathbb{R}: \mu(f^{-1}(-\infty,b)) =0\}$. For any Markov kernel $T$ on $\mathrm{L}^2(\mu)$, $f \in \mathrm{L}^2(\mu)$, we have $\|Tf\|_\mathrm{osc}^2\leqslant \|f\|_\mathrm{osc}^2$ since $\text{ess}_\mu\inf f\leqslant Tf(x)=\mathbb E[f(X_1)\mid X_0 =x]\leqslant \text{ess}_\mu\sup f$ for $\mu$-almost every $x$.

\subsection{Weak Poincar{\'{e}} inequalities}
\label{subsec:WPI}
We introduce the framework of weak Poincar{\'{e}} inequalities in the context of discrete-time Markov chains as developed in \cite{andrieu2022comparison, andrieu2023weak}.
\begin{definition}[Weak Poincar{\'{e}} inequality]
\cite{rockner2001weak, andrieu2022comparison}
\label{def:WPI}
    Given a Markov transition operator $T$ with $\mathrm{L}^2(\mu)$-adjoint $T^*$ on $\mathsf{E}$, we will say that $T^*T$ satisfies a weak Poincar{\'{e}} inequality (WPI) if for any $f \in \mathrm{L}^2_0(\mu)$,
    \begin{align}
        \|f\|_{\mu}^{2}\leqslant s\, {\mathcal{E}_\mu}(T^{*}T,f)+\beta(s) \|f\|^2_\mathrm{osc},\quad\forall s>0,
        \label{eq:WPI_beta}
    \end{align}
    where $\beta:(0,\infty)\rightarrow[0,\infty)$ is a decreasing function with $\lim_{s \rightarrow \infty} \beta(s) = 0$. 
    
    For $\beta$ as above, let
    $\mathrm{K} : [0, \infty) \rightarrow [0, \infty)$ be given by $ \mathrm{K}(u):= \mathbbm{1}_{\{u>0\}}u \cdot \beta(1 / u)$.
    Then let $\mathrm{K}^* : [0, \infty) \rightarrow [0, \infty]$ be defined as the convex conjugate,
    \begin{align*}
        \mathrm{K}^*(v) := \sup_{u\geqslant 0}\{uv - \mathrm{K}(u)\}.
    \end{align*} 

    Then the weak Poincar{\'{e}} inequality \eqref{eq:WPI_beta} can be equivalently expressed as an \textit{optimized} $K^*$-WPI; see \cite[Proposition 4]{andrieu2023weak}:
\begin{align}
    \mathrm{K}^*\left(\frac{\|f\|_\mu^2}{\|f\|^2_\mathrm{osc}}\right) \leqslant \frac{\mathcal{E}_\mu\left(T^* T, f\right)}{\|f\|^2_\mathrm{osc}}, \quad\quad  \forall f \in \mathrm{L}^2_0(\mu),\, 0<\|f\|_\mathrm{osc}<\infty.
\end{align}
\end{definition}
Without loss of generality, $\beta$ can be chosen such that $\beta(s) \leqslant 1/4$ for all $s > 0$, since $\|f\|^2_\mu\leqslant \|f\|^2_{\mathrm{osc}}/4$ for any $f \in \mathrm{L}^2_0(\mu)$ by Popoviciu's inequality. 

A weak Poincar{\'{e}} inequality implies the existence of a function decreasing to 0, which controls the (subgeometric) convergence rate of $T$ in $\mathrm{L}^2$. 

\begin{theorem}
    \cite[Theorem 8]{andrieu2022comparison}
\label{Thm:rate}
    Suppose a WPI with $\beta: (0, \infty) \rightarrow [0, \infty)$ holds for $\mu$ and $T^*T$. Then, for any $f \in \mathrm{L}^2_0(\mu)$ such that $\|f\|^2_\mathrm{osc} < \infty$ and any $n \in \mathbb{N}$,
    \begin{align*}
    \|T^nf \|^2_\mu \leqslant \|f\|^2_\mathrm{osc} F^{-1}(n),
    \end{align*}
    where $F: \left(0, \frac{1}{4}\right] \rightarrow \mathbb{R}$ is the decreasing convex and invertible function $F(x) :=\int_x^{\frac{1}{4}} \frac{\dif v}{\mathrm{K}^*(v)}$, with $\mathrm{K}^* : [0, \infty) \rightarrow [0, \infty]$ defined as \Cref{def:WPI}.
\end{theorem}

Thus for the function $\mathrm{K}^*(v)$, its behaviour in a neighbourhood of 0 critically influences the tail behavior of $F^{-1}$.

\begin{proposition} 
\cite[Remark 3]{andrieu2022comparison}
    If for some $\gamma > 0$, $\beta(s) = 0$ for all $s \geqslant \gamma^{-1}$, a WPI entails a standard Poincar{\'{e}}  inequality (SPI) for $T^*T$: for any $f \in \mathrm{L}^2_0(\mu)$,
    \begin{align}
    \label{eq:SPI}
        \gamma \cdot\|f\|_\mu^2 \leqslant \mathcal{E}_\mu(T^*T, f).
    \end{align}
     When this holds for a Markov kernel $T$ with $\mathrm{L}^2(\mu)$-adjoint $T^*$, we can deduce exponential convergence immediately: for any $f \in \mathrm{L}^2_0(\mu)$, 
     \begin{align*}
         \|T^n f\|^2_\mu \leqslant (1-\gamma)^n\|f\|_\mu^2 \leqslant \exp(-\gamma n)\|f\|_\mu^2. 
     \end{align*}
\end{proposition}

\begin{remark}
\label{rmk:strongtoweak}
    If $T^*T$ satisfies a standard Poincar{\'{e}} inequality with constant $\gamma$, one may take the corresponding $\beta$ to be $\beta(s)=\mathbbm{1}_{\{s\leqslant \gamma^{-1}\}}$. Conversely, if  $\beta(s)=\mathbbm{1}_{\{s\leqslant \gamma^{-1}\}}$ then one can deduce that a standard Poincar{\'{e}} inequality  holds.
\end{remark}
\begin{remark}
    If $T$ is reversible and the spectral gap $\gamma$ of $T$ exists, i.e. for any $f \in \mathrm{L}^2(\mu)$, $\gamma=\inf_{f: \mathrm{Var}_{\mu}(f) \neq 0} \frac{\mathcal{E}_{\mu}(T,f)}{\mathrm{Var}_{\mu}(f)}>0$, it implies a SPI with constant $\gamma$ for $T$.
    
\end{remark}

When the kernel $T$ is reversible and positive, we can derive a weak Poincar{\'{e}} inequality in terms of $T$, rather than $T^*T$, making the approach much more practical: 
\begin{lemma}[\cite{andrieu2022comparison}]
\label{lem:PtoP^2}
    Suppose that the kernel $T$ is $\mu$-reversible and positive. Then we have the bound on the Dirichlet forms for $f \in \mathrm{L}^2 (\mu)$,
    \begin{align*}
        \mathcal{E}_\mu(T^2,f)\geqslant \mathcal{E}_\mu(T,f).
    \end{align*}
\end{lemma}
Thus for positive kernels, a weak Poincar{\'{e}} inequality for $T$ implies a weak Poincar{\'{e}} inequality for $T^2$ with the same $\beta$ function; we can then apply \Cref{Thm:rate}. See \cite[Section 2.2]{andrieu2022comparison} for further details.
Conversely, if $T^*T$ satisfies a WPI, we can deduce that $T$ satisfies a WPI by the following lemma, whose proof is in \Cref{Prf:lem:P^2toP}.

\begin{lemma}   
    \label{lem:P^2toP}
    Let $T$ be $\mu$-invariant. Then  
    \begin{align*}
        \mathcal{E}_\mu(T^*T, f) \leqslant  2\mathcal{E}_\mu(T, f), \quad  f \in \mathrm{L}^2(\mu).
    \end{align*}
\end{lemma}

Furthermore, comparison of Dirichlet forms can be used to deduce convergence properties of a given Markov chain from another one. 

\begin{theorem}[Chaining Poincar{\'{e}} inequalities]
\cite[Proposition 33]{andrieu2022comparison}
\label{Thm:chaining}
    Let $T_1$ and $T_2$ be two $\mu$-invariant Markov kernels. Assume that for all $s > 0$ and $f \in \mathrm{L}^2_0 (\mu)$,
    \begin{align}
            \|f\|_\mu^2 & \leqslant s \mathcal{E}_\mu\left(T_1^*T_1, f\right)+\beta_1(s) \|f\|^2_\mathrm{osc}, \nonumber \\
            \label{eq:comparison}
            \mathcal{E}_\mu\left(T_1^*T_1, f\right) & \leqslant s \mathcal{E}_\mu\left(T_2^*T_2, f\right)+\beta_2(s) \|f\|^2_\mathrm{osc},
    \end{align}
    where $\beta_1, \beta_2:(0, \infty) \rightarrow(0, \infty)$ satisfy the conditions of Definition~\ref{def:WPI}.
    Then there is the following WPI for $T_2^* T_2$: 
    \begin{align*}
        \|f\|_\mu^2 & \leqslant s \mathcal{E}_\mu\left(T_2^*T_2, f\right)+\beta(s) \|f\|^2_\mathrm{osc},\quad s > 0, f\in \mathrm L^2_0(\mu),\\
                \beta(s)&:=\inf \left\{s_1 \beta_2\left(s_2\right)+\beta_1\left(s_1\right) \mid s_1>0, s_2>0, s_1 s_2=s\right\}.
    \end{align*}
    Additionally, writing $\mathrm{K}_i(u):=u \cdot \beta_i(1 / u)$, $\mathrm{K}(u):=u \cdot \beta(1 / u)$, it holds that 
    \begin{align*}
        \begin{split}
            \mathrm{K}(u)&=\inf \left\{\mathrm{K}_2\left(u_2\right)+u_2 \mathrm{K}_1\left(u_1\right) \mid u_1>0, u_2>0, u_1 u_2=u\right\}, \\
             \mathrm{K}^*(v)&=\mathrm{K}_2^* \circ \mathrm{K}_1^*(v).
        \end{split}
    \end{align*}
\end{theorem}

\Cref{Thm:chaining} also holds if $T_1^*T_1$ is replaced by $T_1$, and if $T_2^*T_2$ is replaced by $T_2$. Therefore, if we know $T_1$ or $T^*_1T_1$ satisfies SPI or WPI, we can deduce $T_2$ or $T^*_2T_2$ satisfies WPI, and
deduce the convergence rate of $T_2$. 

\begin{remark}
\label{def:comparison WPI}
    We refer to \eqref{eq:comparison}  as the comparison WPI for $T_1^*T_1$ and $T^*_2T_2$. Similarly, $T_1^*T_1$ and $T^*_2T_2$ satisfy the comparison $K^*$-WPI: $\mathrm{K}^*\left(\frac{\mathcal{E}_\mu(T_1^*T_1,f)}{\|f\|^2_\mathrm{osc}}\right) \leqslant \frac{\mathcal{E}_\mu\left(T_2^*T_2, f\right)}{\|f\|^2_\mathrm{osc}}$, with $\mathrm{K}^*$ defined as \Cref{def:WPI}.
\end{remark}

\begin{lemma}
\cite[Lemma 14]{andrieu2022comparison}
\label{thm:ConstantOnBeta}
    Let $\tilde{\beta}(s):=c_1 \beta\left(c_2 s\right)$ for $c_1, c_2>0$. Then $\tilde{K}^*(v):=\sup _{u \in \mathbb{R}_{+}} u[v-\tilde{\beta}(1 / u)]=c_1 c_2 \mathrm{K}^*\left(v / c_1\right)$ and the corresponding function $\tilde{F}(w)=c_2^{-1} \int_{w/c_1}^{1/{4c_1}} \frac{\dif v}{\mathrm{K}^*(v)}$. Furthermore, when $c_1 \geqslant 1, \tilde{F}(w) \leqslant c_2^{-1} F\left(w / c_1\right)$, and we can conclude $\tilde{F}^{-1}(x) \leqslant$ $c_1 F^{-1}\left(c_2 x\right)$.
\end{lemma}

\begin{lemma}
\cite[Lemma 64]{andrieu2022comparison}
\label{thm:K*Properties}
With $\mathrm{K}^*$ defined as in \Cref{def:WPI}, we have:
\begin{enumerate}
    \item $\mathrm{K}^*(0)=0$ and for $v \neq 0$, $\mathrm{K}^*(v)>0$; 
    \item $v \mapsto \mathrm{K}^*(v)$ is convex, continuous and strictly increasing on its domain;
    \item for $v \in[0, 1/4]$, it holds that $\mathrm{K}^*(v) \leqslant v$;
    \item  the function  $v \mapsto v^{-1} \mathrm{K}^*(v)$  is increasing.
\end{enumerate}
\end{lemma}

\subsection{Two-component Gibbs samplers}

Suppose $\mathcal{X}$
and $\mathcal{Y}$ are Polish spaces equipped with their Borel $\sigma$-algebras $\mathcal{F}_\mathcal{X}$ and $\mathcal{F}_\mathcal{Y}$, respectively.  Consider the product measurable space $(\mathcal{X}\times\mathcal{Y}, \mathcal{F}_\mathcal{X}\otimes \mathcal{F}_\mathcal{Y})$, and denote $\Pi(\dif x, \dif y)$ as a probability measure defined on this space. Let $\Pi_X(\dif x)$ and $\Pi_Y(\dif y)$ be the associated marginal distributions, and, $\Pi_{X|Y}(\dif x|y)$ and $\Pi_{Y|X}(\dif y|x)$ be the full conditional distributions, which in an abuse of notation we may write as $\Pi(\dif x |y), \Pi(\dif y|x)$. The deterministic-scan Gibbs (DG) sampler is given in \Cref{Alg:DG}.

\begin{algorithm}[H]
\caption{Deterministic-Scan Gibbs (DG) Sampler}
\label{Alg:DG}
\begin{algorithmic}[1]
\Require  Current position $(X_n, Y_n) = (x, y)$.
\State Draw $y'$ from $\Pi_{Y|X}(\cdot|x)$, and set $Y_{n+1}=y'$.
\State Draw $x'$ from $\Pi_{X|Y}(\cdot|y')$, and set $X_{n+1}=x'$.
\State Set $n \leftarrow n+1$.
\end{algorithmic}
\end{algorithm}
Let $(x,y)\mapsto H_{1|x}(\dif y'|y)$ be a Markov kernel $(\mathcal{X}\times\mathcal{Y}, \mathcal{F}_\mathcal{X}\otimes \mathcal{F}_\mathcal Y) \to (\mathcal{Y},\mathcal{F}_\mathcal{Y})$ 
and $(x,y)\mapsto H_{2|y}(\dif x'|x)$ be a Markov kernel $(\mathcal{X}\times\mathcal{Y}, \mathcal{F}_\mathcal{X}\otimes \mathcal{F}_Y) \to (\mathcal{X},\mathcal{F}_\mathcal{X})$. Assume $y\mapsto H_{1|x}(\dif y'|y)$ and $x\mapsto H_{2|y}(\dif x'|x)$ are reversible with respect to their corresponding conditional distributions, $\Pi_{Y|X}(\cdot |x)$ and $\Pi_{X|Y}(\cdot |y)$, for fixed $x$ and $y$ respectively. We will refer to methods that use such kernels in place of the full conditionals as MwG samplers (MwG) to distinguish them from exact Gibbs samplers;
see \Cref{Alg:DHG}.

\begin{algorithm}[H]
\caption{Deterministic-Scan MwG (MwG) Sampler}
\label{Alg:DHG}
\begin{algorithmic}[1]
\Require Current position $(X_n, Y_n) = (x, y)$.
\State Draw $y'$ from $H_{1|x}(\cdot|y)$, which is reversible w.r.t $\Pi_{Y|X}(\cdot|x)$, and set $Y_{n+1}=y'$.
\State Draw $x'$ from $H_{2|y'}(\cdot|x)$, which is reversible w.r.t $\Pi_{X|Y}(\cdot|y')$, and set $X_{n+1}=x'$.
\State Set $n \leftarrow n+1$.
\end{algorithmic}
\end{algorithm}

The transition kernel for the DG sampler is
\begin{align*}
    P((x,y),(\dif x',\dif y'))=\Pi_{X|Y}(\dif x'|y')\Pi_{Y|X}(\dif y'|x),
\end{align*}
which corresponds to doing first a $Y$-move and then an $X$-move.
It is straightforward to verify that $P$ leaves $\Pi$  invariant but it is not reversible with respect to $\Pi$. 

Let $G_1$ (respectively $G_2$) characterize the transition rule in the \textit{joint space} $\mathcal X\times \mathcal Y$ for updating $Y$ (respectively $X$) through the conditional distribution of $Y$ given $X$ (respectively $X$ given $Y$):
\begin{align*}
\begin{split}
    G_1((x,y),(\dif x',\dif y'))=\Pi_{Y|X}(\dif y'|x)\delta_x(\dif x'),\quad 
    G_2((x,y),(\dif x',\dif y'))=\Pi_{X|Y}(\dif x'|y)\delta_y(\dif y').
\end{split}
\end{align*}
The kernels $G_1$ and $G_2$ satisfy $\Pi G_1 = \Pi G_2 = \Pi$.

In order to determine the convergence properties of the  MwG sampler, we consider two further auxiliary MwG samplers, where one of the conditional distributions associated with $\Pi$ is intractable. 

The first component deterministic-scan MwG (1-MwG) samplers is where, instead of sampling exactly from $\Pi_{Y|X}$, we run one step of a kernel $H_1$, which is invariant with respect to $\Pi_{Y|X}$. 
The transition kernel $P_1$ for the 1-MwG sampler is
\begin{align*}
    P_1((x,y),(\dif x',\dif y'))=H_{1|x}(\dif y'|y)\Pi_{X|Y}(\dif x'|y').
\end{align*}
Then we embed $H_{1|x}$ into the joint space,
\begin{align*}
    H_1((x,y),(\dif x',\dif y'))=H_{1|x}(\dif y'|y)\delta_x(\dif x').
\end{align*}

Similarly, we can define the second component deterministic-scan MwG (2-MwG) samplers by the transition kernel $P_2$ and the embedded chain for $H_{2|y'}$: 
\begin{align*}
    P_2((x,y),(\dif x',\dif y'))&=H_{2|y'}(\dif x'|x)\Pi_{Y|X}(\dif y'|x),\\
    H_2((x,y),(\dif x',\dif y'))&=H_{2|y'}(\dif x'|x)\delta_y(\dif y').
\end{align*}
Consequently, the exact DG sampler is a special case of both first or second component deterministic-scan MwG samplers. Furthermore, these first or second component deterministic-scan MwG samplers are, in turn, special cases of the MwG sampler (\Cref{Alg:DHG}).

We also consider the operators on $\mathrm{L}^2(\Pi)$ associated with these kernels. Let $P$ be the operator on $\mathrm{L}^2(\Pi)$ associated to the Gibbs sampling kernel $P$. Then we can write $Pf(x):=(G_1(G_2f))(x)$, where $G_1$ and $G_2$ are operators defined on $\mathrm{L}^2(\Pi)$:
\begin{align*}
\begin{split}
    G_1f(x,y)&=\int_{\mathcal Y} f(x,y')\Pi_{Y|X}(\dif y'|x), \quad G_2f(x,y)=\int_{\mathcal X} f(x',y)\Pi_{X|Y}(\dif x'|y).
\end{split}
\end{align*}
It is clear that $P_1=H_1G_2$, $P_2=G_1H_2$, and $P_{1,2}=H_1H_2$ can be regarded as operators defined on $\mathrm{L}^2(\Pi)$ in the same way. These correspond to the 1-MwG sampler, the 2-MwG sampler, and the MwG sampler, respectively, where 
\begin{align*}
\begin{split}
    H_1f(x,y)=\int_{\mathcal Y}f(x,y')H_{1|x}(\dif y'|y), \quad 
    H_2f(x,y)=\int_{\mathcal X} f(x',y)H_{2|y}(\dif x'|x).
\end{split}
\end{align*}
These single-component Gibbs samplers $G_i$ and $H_i$, $i=1,2$, are readily observed to be reversible Markov chains corresponding to self-adjoint operators. Indeed, $G_1$ and $G_2$ are projection operators on the joint space. Similarly, if $H_{1|x}$ and $H_{2|y}$ are reversible with respect to $\Pi_{Y|X}(\dif y|x)$ and $\Pi_{X|Y}(\dif x|y)$, respectively, then $H_1$ and $H_2$ are self-adjoint operators on $\mathrm{L}^2(\Pi)$. 

Finally, we define the $X$-marginal versions of the DG chain, which is a Markov chain on $\mathcal{X}$.
We denote $P_{X}$ as the $X$-marginal chain of $P$, and denote $\bar{P}_{X}$ as the $X$-marginal chain of $P_2$:
\begin{align}
    P_{X}(x,\dif x')=\int_{\mathcal Y} \Pi_{X|Y}(\dif x'|y)\Pi_{Y|X}(\dif y|x), \quad \bar{P}_{X}(x,\dif x')=\int_{\mathcal Y} H_{2|y}(\dif x'|x)\Pi_{Y|X}(\dif y|x).
    \label{eq:marginalPX}
\end{align}

Note that $P_{X}$ is reversible with respect to $\Pi_X$. It is well-known that the convergence rate of the marginal chain $P_X$ is essentially the same as that of the original DG chain, in geometric cases \cite{liu1994covariance, qin2022convergence}. Similarly, $\bar{P}_{X}$ is reversible with respect to $\Pi_X$. However, due to lack of the Markov property, there is no corresponding $Y$-marginal version of the chain for $P_2$. This is because, after one update, the conditional distribution of $Y'$ given $(X,Y)$ depends on $X$, whereas the conditional distribution of $X'$ given $(X,Y)$ is identical to the conditional distribution of $X'$ given $X$.

\section{Comparison results for DG Sampling}
\label{Sec: Comparison for Gibbs}
In this section, we shall give comparison results for DG samplers and MwG samplers. Our ultimate goal is to supply bounds which describe the cost incurred when using MwG, relative to the rate of the underlying DG chain.

\subsection{MwG sampler with two intractable conditionals}
\label{Sec: HDG}
Our roadmap is as follows. We begin our analysis on the joint space by updating the $Y$-component while keeping $X$ fixed, and then establish a comparison between the transition kernel $G_1$ and its corresponding MwG version $H_1$ (see \Cref{lem:P1_to_tildeP1}). A similar result holds for the comparison between $G_2$ and $H_2$.

Next, we establish the comparison between the DG kernel $P$ (\Cref{Alg:DG}) and 1-MwG sampler $P_1$ by decomposing the Dirichlet form associated with $P$ (\Cref{thm:K*:PtoP1}). This provides a framework for replacing the first component with its corresponding MwG version in Gibbs samplers.

Furthermore, \Cref{Thm:K*:PtoP*} provides a framework for comparing a transition kernel with its adjoint. This allows us to derive the WPI when swapping the two components, giving the WPI for 2-MwG samplers as an intermediate step toward the WPI for MwG samplers.

In summary, our comparison path is as follows:  
\begin{enumerate}
    \item swap to make $G_2$ the first component;  
    \item replace $G_2$ with its corresponding MwG version $H_2$; 
    \item swap again to make $G_1$ the first component; 
    \item replace $G_1$ with its corresponding MwG version $H_1$.  
\end{enumerate}

These comparisons are summarized in the schematic diagram presented in \Cref{fig:PandP2}.

\begin{figure}[H] 
    \centering 
    \begin{tikzpicture}[>=stealth]
        \node (P) at (0,2) {$P=G_1G_2$};
        \node (P^*) at (4,2) {$P^*=G_2G_1$};
        \node (P_{(2)}^*) at (8,2) {$P_2^*=H_2G_1$};
        \node (P_{(2)}) at (8,0) {$P_2=G_1H_2$};
        \node (tildeP) at (4,0) {$P_{1,2}=H_1H_2$};
        \draw [->] (P) -- node[above] {\Cref{Thm:K*:PtoP*}} node[below] {Step 1}  (P^*); 
        \draw [->] (P^*) -- node[above] {\Cref{thm:K*:P*toP2*}}
        node[below] {Step 2} (P_{(2)}^*);
        \draw [->] (P_{(2)}^*) -- node[right] {\Cref{Thm:K*:PtoP*}}
        node[left] {Step 3} (P_{(2)});
        \draw [->] (P_{(2)}) -- node[below] {\Cref{thm:K*:P1tildeP2totildeP1tildeP2}} node[above] {Step 4} (tildeP);
    \end{tikzpicture}
    \caption{Relationship among two-component Gibbs samplers.} 
    \label{fig:PandP2} 
\end{figure}

\Cref{thm:weakPtotildeP} is our main result, which establishes a WPI for the MwG sampler under \Cref{Asp:P1,Asp:P2,Asp:P}. The proof is deferred to \Cref{prf:thm:weakPtotildeP}.

\begin{assumption}
\label{Asp:P1}
    We assume that the kernel $H_{1|x}$ is $\Pi(\cdot|x)$-reversible, positive and satisfies a WPI for each $x\in \mathcal X$: there is a function $\beta_1: (0, \infty) \times \mathcal{X} \rightarrow [0, \infty)$ with $\beta_1(\cdot, x)$ satisfying the conditions in \Cref{def:WPI} for each $x \in \mathcal{X}$, such that
    \begin{align*}
        \mathrm{Var}_{\Pi(\cdot|x)}(h) \leqslant s \cdot \mathcal{E}_{\Pi(\cdot|x)}(H_{1|x},h)+\beta_1(s,x) \cdot \|h\|^2_\mathrm{osc}, \quad \forall s > 0, h \in \mathrm{L}^2(\Pi(\cdot|x)).
    \end{align*}
\end{assumption}
\begin{assumption}
    \label{Asp:P2}
        We assume that the kernel $H_{2|y}$ is $\Pi(\cdot|y)$-reversible, positive and satisfies a WPI for each $y\in \mathcal Y$: there is a function $\beta_2: (0, \infty) \times \mathcal{Y} \rightarrow [0, \infty)$ with $\beta_2(\cdot, y)$ satisfying the conditions in \Cref{def:WPI} for each $y \in \mathcal{Y}$, such that
        \begin{align*}
            \mathrm{Var}_{\Pi(\cdot|y)}(g) \leqslant s \cdot \mathcal{E}_{\Pi(\cdot|y)}(H_{2|y},g)+\beta_2(s,y) \cdot \|g\|^2_\mathrm{osc},\quad \forall s > 0, g \in \mathrm{L}^2(\Pi(\cdot|y)).
        \end{align*}
    \end{assumption}

    \begin{assumption}
        \label{Asp:P}
        We assume that the kernel $P^*P$ satisfies a WPI: there is a function $\beta_0: (0, \infty) \rightarrow [0, \infty)$ satisfying the conditions in \Cref{def:WPI}, such that 
        \begin{align*}
            \|f\|_\Pi^2 \leqslant s \cdot \mathcal{E}_\Pi(P^*P, f) +\beta_0(s) \cdot \|f\|^2_\mathrm{osc},\quad\quad  \forall s > 0, f \in \mathrm L_0^2(\Pi).
        \end{align*}
    \end{assumption}

    \begin{theorem}
    \label{thm:weakPtotildeP}
        Under the \Cref{Asp:P1}, \Cref{Asp:P2} and \Cref{Asp:P}, the operator $P_{1,2}^*P_{1,2}$ satisfies the $K^*$-WPI with $K^*=2K_1^*\left(K_2^*\left(\frac{1}{2}K_0^*(v/4)\right)\right)$, where $K^*_0$, $K^*_1$ and $K^*_2$ correspond to $\beta_0$, $\beta_1(s)=\int_\mathcal{X} \beta_1(s,x)\Pi_X(\dif x)$ and $\beta_2=\int_\mathcal{Y} \beta_2(s,y)\Pi_Y(\dif y)$, respectively.
    \end{theorem}
    
\begin{remark}
\label{rmk:requirementforH}
    The requirement that the kernels $H_{1|x}$ and $H_{2|y}$ satisfy a weak Poincar{\'{e}} inequality is extremely mild, and holds under irreducibility; see \cite{andrieu2023weak}. We further assume that both $H_{1|x}$ and $H_{2|y}$ are reversible, which is satisfied, for instance, if we use Metropolis--Hastings or slice sampling. Positivity is satisfied automatically by chains such as Random Walk Metropolis with Gaussian increments, Independent Metropolis--Hastings, or can be enforced by considering the \textit{lazy chain}. \Cref{Asp:P} is also a standard assumption, which states that the convergence behavior of the exact DG kernel $P$ is tractable. In the two-component case, \cite{liu1994covariance} shows that the convergence rate of the DG kernel $P$ is determined by the maximal correlation between the components; see also \cite{diaconis1999iterated}.
\end{remark}

Since our goal is to obtain the WPIs for the MwG sampler based on \Cref{Asp:P1,Asp:P2,Asp:P} as shown in \Cref{thm:weakPtotildeP}, as a first step, we establish comparison WPIs for the corresponding embedded chain and its corresponding MwG version under these assumptions.
The following result is standard; for completeness, we include a proof in \Cref{Prf:lem:P_1reversible}.
\begin{lemma}
\label{lem:P_1reversible}
    The kernels $G_1$ and $G_2$ are $\Pi$-reversible and positive. Under the \Cref{Asp:P1}, the kernel $H_1$ is $\Pi$-reversible and positive. Under the \Cref{Asp:P2}, the kernel $H_2$ is $\Pi$-reversible and positive.
\end{lemma}

The following Lemma shows the comparison between $G_1$ and $H_1$, as well as between $G_2$ and $H_2$. The proof is similar to \cite[Theorem 4.2.3]{power2024weak} and is given in \Cref{prf:lem:P1_to_tildeP1}.
\begin{lemma}
\label{lem:P1_to_tildeP1}
    Under \Cref{Asp:P1}, we have the following WPI relating $G_1$ and $H_1$:
    \begin{align*}
        \mathcal{E}_{\Pi}(G_1,f) \leqslant s\cdot \mathcal{E}_{\Pi}(H_1,f)+\beta_1(s)\|f\|^2_{\mathrm{osc}}, \quad  \forall s>0, f \in \mathrm{L}^2(\Pi),
    \end{align*}
    where $\beta_1: (0, \infty) \rightarrow [0, \infty)$ is given by $\beta_1(s):= \int_\mathcal{X} \beta_1(s,x)\Pi_X(\dif x)$.
    
    Similarly, under \Cref{Asp:P2}, we have the following WPI for $G_2$ and $H_2$:
    \begin{align*}
        \mathcal{E}_{\Pi}(G_2,f) \leqslant s\cdot \mathcal{E}_{\Pi}(H_2,f)+\beta_2(s)\|f\|^2_{\mathrm{osc}}, \quad \forall s>0, f \in \mathrm{L}^2(\Pi),
        \end{align*}
    where $\beta_2: (0, \infty) \rightarrow [0, \infty)$ is given by $\beta_2(s):= \int_\mathcal{Y} \beta_2(s,y)\Pi_Y(\dif y)$. 
    \end{lemma}

We now state the comparison results between the DG kernel $P^*P$ and the 1-MwG kernel $P_1^*P_1$. Since we have derived the comparison between $G_1$ and $H_1$, the comparison between $P^*P$ and $P_1^*P_1$ follows from the decomposition of Dirichlet forms (\Cref{thm: decomposition}). 
We will use the $K^*$-WPI as the comparison tool in \Cref{thm:K*:PtoP1}, and show that the $K^*$ function of $P^*P$ and $P_1^*P_1$ can be derived by the $K^*$ function of $G_1$ and $H_1$.

We begin with a useful general observation.
\begin{lemma}
\label{lem:f_osc}
    Suppose that $T$ satisfies the optimised $K^*$-WPI for a given $\mathrm{K}^*$, then for any $\mu$-invariant Markov transition operator $R$ on $\mathrm{L^2(\mu)}$, 
    we have for $f \in \mathrm{L}_0^2(\mu)$,
    \begin{align}
        \mathrm{K}^*\left(\frac{\|Rf\|_\mu^2}{\|f\|^2_\mathrm{osc}}\right) \leqslant \frac{\mathcal{E}_\mu\left(T, Rf\right)}{\|f\|^2_\mathrm{osc}}.
    \end{align}
    Similarly, suppose the comparison $K^*$-WPI between $T_1$ and $T_2$ holds for a given $\mathrm{\bar{K}}^*$, then for any Markov transition operator $R$, $\mathrm{\bar{K}}^*\left(\frac{\mathcal{E}_\mu\left(T_1, Rf\right)}{\|f\|^2_\mathrm{osc}}\right) \leqslant \frac{\mathcal{E}_\mu\left(T_2, Rf\right)}{\|f\|^2_\mathrm{osc}}$.
\end{lemma}
\begin{proof}
    Applying the WPI with $f$ replaced by $Rf$, we have
    \begin{align*}
        \mathrm{K}^*\left(\frac{\|Rf\|_\mu^2}{\|Rf\|^2_\mathrm{osc}}\right)\leqslant\frac{\mathcal{E}_\mu(T,Rf)}{\|Rf\|^2_\mathrm{osc}}.
    \end{align*}
    It is easy to see that $\|Rf\|^2_\mathrm{osc} \leqslant \|f\|^2_\mathrm{osc}$, since $R$ is a Markov operator. Hence $\frac{\left\|R f\right\|_\mu^2}{\|Rf\|^2_\mathrm{osc}} \geqslant \frac{\left\|R f\right\|_\mu^2}{\|f\|^2_\mathrm{osc}}$. Additionally, since $u \mapsto u^{-1} \mathrm{K}^*(u)$ is increasing by \Cref{thm:K*Properties}, we obtain the conclusion
    \begin{align*}
    \left(\frac{\left\|R f\right\|_\mu^2}{\|Rf\|^2_\mathrm{osc}}\right)^{-1} \cdot \mathrm{K}^*\left(\frac{\left\|R f\right\|_\mu^2}{\|Rf\|^2_\mathrm{osc}}\right) \geqslant\left(\frac{\left\|R f\right\|_\mu^2}{\|f\|^2_\mathrm{osc}}\right)^{-1} \cdot \mathrm{K}^*\left(\frac{\left\|R f\right\|_\mu^2}{\|f\|^2_\mathrm{osc}}\right) \\
    \Longrightarrow \mathrm{K}^*\left(\frac{\left\|R f\right\|_\mu^2}{\|Rf\|^2_\mathrm{osc}}\right) \geqslant\left(\frac{\|f\|^2_\mathrm{osc}}{\|Rf\|^2_\mathrm{osc}}\right) \cdot \mathrm{K}^*\left(\frac{\left\|R f\right\|_\mu^2}{\|f\|^2_\mathrm{osc}}\right).
    \end{align*}

    We thus see that
    \begin{align*}
    & \frac{\mathcal{E}_\mu\left(T, R f\right)}{\|Rf\|^2_\mathrm{osc}} \geqslant \left(\frac{\|f\|^2_\mathrm{osc}}{\|Rf\|^2_\mathrm{osc}}\right) \cdot \mathrm{K}^*\left(\frac{\left\|R f\right\|_\mu^2}{\|f\|^2_\mathrm{osc}}\right) \\
    \Longrightarrow & \frac{\mathcal{E}_\mu\left(T, R f\right)}{\|f\|^2_\mathrm{osc}} \geqslant \mathrm{K}^*\left(\frac{\left\|R f\right\|_\mu^2}{\|f\|^2_\mathrm{osc}}\right) .
    \end{align*}
    The argument still holds after replacing $\left\|R f\right\|_\mu^2$ with $\mathcal{E}_\mu\left(T_1, R f\right)$ and $\mathcal{E}_\mu\left(T, R f\right)$ with $\mathcal{E}_\mu\left(T_2, R f\right)$.
\end{proof}
We will require the following  decomposition of Dirichlet forms.
\begin{lemma}
\label{thm: decomposition}
Suppose $T_1$ and $T_2$ are $\Pi$-invariant Markov kernels, each of which is self-adjoint on $\mathrm{L}^2(\Pi)$.  
For their composition $T=T_1T_2$, which is $\Pi$-invariant, we have
\begin{align*}
    \mathcal{E}_{\Pi}(T^*T,f)=\mathcal{E}_{\Pi}(T_2^2,f)+\mathcal{E}_{\Pi}(T_1^2,T_2f).
\end{align*}
\end{lemma}
\begin{proof}
    By the definition of Dirichlet form,
    \begin{align*}
        \begin{split}
            \mathcal{E}_{\Pi}(T^*T,f)&=\mathcal{E}_{\Pi}(T_2T_1T_1T_2,f)\\
            &=\|f\|^2_\Pi-\|T_1T_2f\|^2_\Pi\\
            &=\|f\|^2_\Pi-\|T_2f\|_\Pi^2+\|T_2f\|_\Pi^2-\|T_1T_2f\|^2_\Pi\\
            &=\mathcal{E}_{\Pi}(T_2^*T_2,f)+\mathcal{E}_{\Pi}(T_1^*T_1,T_2f)\\
            &=\mathcal{E}_{\Pi}(T_2^2,f)+\mathcal{E}_{\Pi}(T_1^2,T_2f).
        \end{split}
    \end{align*}
\end{proof}
The following theorem provides a comparison result for the DG sampler and 1-MwG sampler.
\begin{theorem}
    \label{thm:K*:PtoP1}
        We suppose that we have the following $K^*$-WPI comparison between $G_1$ and $H_1$: for all $s>0$, $f \in \mathrm{L}^2(\Pi)$,
        \begin{align*}
        \mathrm{K}_1^*\left(\frac{\mathcal{E}_\Pi(G_1,f)}{\|f\|^2_{\mathrm{osc}}}\right)\leqslant \frac{\mathcal{E}_\Pi(H_1,f)}{\|f\|^2_{\mathrm{osc}}} .
        \end{align*}
        Then we have the following comparison between $P^*P$ and $P_1^*P_1$:
        for all $s>0$, $f \in \mathrm{L}^2(\Pi)$,
        \begin{align*}
        \tilde{\mathrm{K}}_1^*\left(\frac{\mathcal{E}_\Pi(P^*P,f)}{\|f\|^2_{\mathrm{osc}}}\right) \leqslant \frac{\mathcal{E}_\Pi(P_1^*P_1,f)}{\|f\|^2_{\mathrm{osc}}},
        \end{align*}
        where $\tilde{\mathrm{K}}_1^*(v)=2\mathrm{K}_1^*(v/2)$.
\end{theorem}
\begin{proof}
        From \Cref{lem:PtoP^2}, it is easy to show that
        \begin{align*}
            \frac{\mathcal{E}_\Pi(H_1^2,f)}{\|f\|^2_{\mathrm{osc}}} \geqslant \frac{\mathcal{E}_\Pi(H_1,f)}{\|f\|^2_{\mathrm{osc}}} \geqslant \mathrm{K}_1^*\left(\frac{\mathcal{E}_\Pi(G_1,f)}{\|f\|^2_{\mathrm{osc}}}\right)=\mathrm{K}_1^*\left(\frac{\mathcal{E}_\Pi(G_1^2,f)}{\|f\|^2_{\mathrm{osc}}}\right),
        \end{align*}
        using the fact that $G_1$ is a projection. 
        
        By \Cref{lem:f_osc}, since $G_2$ is a Markov operator, we have
        \begin{align}
        \label{eq:prf:H_1 and G_1}
        \frac{\mathcal{E}_\Pi(H_1^2,G_2f)}{\|f\|^2_{\mathrm{osc}}} \geqslant \mathrm{K}_1^*\left(\frac{\mathcal{E}_\Pi(G_1^2,G_2f)}{\|f\|^2_{\mathrm{osc}}}\right).
        \end{align}
        Since for $v \in \left[0, 1/4\right]$, it holds that $\mathrm{K}_1^*(v) \leqslant v$, and $v \mapsto \mathrm{K}_1^*(v)$ is convex in its domain by \Cref{thm:K*Properties}. Then, we decompose $\mathcal{E}_\Pi\left(P_1^*P_1,f\right)$ by \Cref{thm: decomposition},
        \begin{align*}
            \begin{split}
                \frac{\mathcal{E}_\Pi(P_1^*P_1,f)}{\|f\|^2_{\mathrm{osc}}}&=\frac{\mathcal{E}_\Pi(H_1^2,G_2f)}{\|f\|^2_{\mathrm{osc}}}+\frac{\mathcal{E}_\Pi(G_2^2,f)}{\|f\|^2_{\mathrm{osc}}}\\
                &\geqslant \mathrm{K}_1^*\left(\frac{\mathcal{E}_\Pi(G_1^2,G_2f)}{\|f\|^2_{\mathrm{osc}}}\right)+\frac{\mathcal{E}_\Pi(G_2^2,f)}{\|f\|^2_{\mathrm{osc}}} \quad \text{by  \eqref{eq:prf:H_1 and G_1}}\\
                &\geqslant \mathrm{K}_1^*\left(\frac{\mathcal{E}_\Pi(G_1^2,G_2f)}{\|f\|^2_{\mathrm{osc}}}\right)+\mathrm{K}_1^*\left(\frac{\mathcal{E}_\Pi(G_2^2,f)}{\|f\|^2_{\mathrm{osc}}}\right) \quad \text{by $\mathrm{K}_1^*(v) \leqslant v$}\\
                &\geqslant 2\mathrm{K}_1^*\left(\frac{1}{2}\left(\frac{\mathcal{E}_\Pi(G_1^2,G_2f)}{\|f\|^2_{\mathrm{osc}}}+\frac{\mathcal{E}_\Pi(G_2^2,f)}{\|f\|^2_{\mathrm{osc}}}\right)\right) \quad \text{by convexity}\\
                &=2\mathrm{K}_1^*\left(\frac{1}{2}\frac{\mathcal{E}_\Pi(P^*P,f)}{\|f\|^2_{\mathrm{osc}}}\right)=: \tilde{\mathrm{K}}_1^*\left(\frac{\mathcal{E}_\Pi(P^*P,f)}{\|f\|^2_{\mathrm{osc}}}\right).
            \end{split}
        \end{align*}
        
\end{proof}

\begin{remark}
    By \Cref{thm:ConstantOnBeta}, we observe that the function $\tilde{\mathrm{K}}_1^*$ closely resembles $\mathrm{K}_1^*$, differing only by constant factors. Hence, in our context, comparison methods can supply positive results for 1-MwG sampler in the case
where both i) the DG sampler converges well (\Cref{Asp:P}), and ii) the kernels $H_{1|x}$ converge for every $x$ (\Cref{Asp:P1}). Theorem 12 in \cite{qin2025spectral} establishes a similar result although the proofs are formulated using different parameterizations.
\end{remark}

We can similarly derive comparisons between $P_2^*P_2$ and $P^*_{1,2}P_{1,2}$, as well as between $PP^*$ and $P_2{P_2}^*$, by replacing the first component with its corresponding MwG version. The resulting statements are presented in the following corollaries, with proofs provided in \Cref{prf:thm:K*:P1tildeP2totildeP1tildeP2}.
\begin{corollary}
\label{thm:K*:P1tildeP2totildeP1tildeP2}
    Suppose that the comparison $K^*$-WPI between $G_1$ and $H_1$ holds with the function $\mathrm{K}_1^*$. It follows that we have the comparison $K^*$-WPI between $P_2^*P_2$ and $P_{1,2}^*P_{1,2}$ with the function $\tilde{\mathrm{K}}_1^*(v)=2\mathrm{K}_1^*(v/2)$.
\end{corollary}
\begin{corollary}
\label{thm:K*:P*toP2*}
     We suppose that we have the comparison $K^*$-WPI between $G_2$ and $H_2$ with the function $\mathrm{K}_2^*$, it follows that we have the comparison $K^*$-WPI between $PP^*$ and $P_2{P_2}^*$ with the function $\tilde{\mathrm{K}}_2^*(v)=2\mathrm{K}_2^*(v/2)$.
\end{corollary}

    \Cref{thm:K*:PtoP1} provides a method for comparing $P^*P$ with $P_1^*P_1$, while the following lemma establishes a way of comparing $P^*P$ and $PP^*$. This describes the relationship of the rate between a chain and its adjoint chain in the subgeometric regime, enabling us to extend the comparison to $P^*P$ and $P_2^*P_2$.
    \begin{lemma}
    \label{Thm:K*:PtoP*}
        Suppose that kernel $TT^*$ satisfies the $K^*$-WPI with a given $\mathrm{K^*}$,
        Then the kernel $T^*T$ satisfies a $K^*$-WPI with $\tilde{\mathrm{K}}^*(v):=\mathrm{K}^*(v/2)$.
    \end{lemma}
    \begin{proof}
        Observe that $0 \leqslant \|f-T^*Tf\|_{\Pi}^2 = \|f\|_{\Pi}^2-\langle T^*Tf,f \rangle_{\Pi} - \langle T^*Tf,f \rangle_{\Pi}+\|T^*Tf\|_{\Pi}^2$. For any real function $f$, this gives $\|f\|_{\Pi}^2-\langle T^*Tf,f \rangle_{\Pi} \geqslant \langle T^*Tf,f \rangle_{\Pi}-\|T^*Tf\|_{\Pi}^2$. Thus, 
        \begin{align*}
            \mathcal{E}_{\Pi}(T^*T,f)=\|f\|_{\Pi}^2-\langle T^*Tf,f \rangle_{\Pi} \geqslant \langle T^*Tf,f \rangle_{\Pi}-\|T^*Tf\|_{\Pi}^2=\mathcal{E}_{\Pi}(TT^*,Tf).
        \end{align*}
        
        Hence, combining with \Cref{lem:f_osc}, since $T$ is a Markov operator, we have
        \begin{align}
        \label{eq:T}
            \mathrm{K}^*\left(\frac{\|Tf\|_\Pi^2}{\|f\|^2_{\mathrm{osc}}}\right)\leqslant \frac{\mathcal{E}_{\Pi}(TT^*,Tf)}{\|f\|^2_{\mathrm{osc}}} \leqslant \frac{\mathcal{E}_{\Pi}(T^*T,f)}{\|f\|^2_{\mathrm{osc}}}.
        \end{align}
        Then we have
        \begin{align*}
            \begin{split}
                \mathrm{K}^*\left(\frac{1}{2} \cdot\frac{\|f\|_{\Pi}^2}{\|f\|^2_{\mathrm{osc}}}\right)&=\mathrm{K}^*\left(\frac{1}{2}\left( \frac{\|f\|_{\Pi}^2-\|Tf\|_{\Pi}^2}{\|f\|^2_{\mathrm{osc}}}+\frac{\|Tf\|_{\Pi}^2}{\|f\|^2_{\mathrm{osc}}}\right)\right)\\
                &\leqslant \frac{1}{2}\left( \mathrm{K}^*\left(\frac{\mathcal{E}_{\Pi}(T^*T,f)}{\|f\|^2_{\mathrm{osc}}}\right)+\mathrm{K}^*\left(\frac{\|Tf\|_{\Pi}^2}{\|f\|^2_{\mathrm{osc}}}\right)\right) \quad \text{by convexity}\\
                &\leqslant \frac{1}{2}\left( \frac{\mathcal{E}_{\Pi}(T^*T,f)}{\|f\|^2_{\mathrm{osc}}}+\frac{\mathcal{E}_{\Pi}(T^*T,f)}{\|f\|^2_{\mathrm{osc}}}\right) \quad \text{by $\mathrm{K}^*(v) \leqslant v$ and \eqref{eq:T}}\\
                &=\frac{\mathcal{E}_{\Pi}(T^*T,f)}{\|f\|^2_{\mathrm{osc}}}.
            \end{split}
        \end{align*}
        
        Then, the $K^*$-WPI for $T^*T$ is obtained with $\tilde{\mathrm{K}}^*(v)=\mathrm{K}^*(v/2)$.
    \end{proof}
    
    We can also derive a comparison between $T^*T$ and $TT^*$ in terms of the $\beta$-parameterized WPI; the proof is in \Cref{Prf:Thm:PtoP*}. 
    \begin{proposition}
    \label{Thm:PtoP*}
        Suppose that $TT^*$ satisfies a WPI with a given $\beta$.
    Then $T^*T$ satisfies a WPI with 
    \begin{equation*}
        \Tilde{\beta}(s)=\begin{cases} \beta(s-1), & s>1, \\ \frac{1}{4}, & 0<s\leqslant 1 .\end{cases}
    \end{equation*}
    \end{proposition}

    \begin{remark}
         By \Cref{Thm:PtoP*}, although $\Tilde{\beta}(s)$ decays slightly slower than ${\beta}(s)$, they exhibit the same asymptotic tail behavior. Furthermore, \Cref{Thm:K*:PtoP*} indicates that the resulting $K^*$ function differs only by a constant factor in its argument, which is reflected in the scaling of the final convergence rate via \Cref{thm:ConstantOnBeta}.
    \end{remark}
    
    Therefore, to derive a WPI for $P_2^*P_2$, we proceed as follows. Starting from the WPI for $P^*P$ (\Cref{Asp:P}), we first obtain the WPI for $PP^*$ by applying \Cref{Thm:K*:PtoP*}. Next, we deduce a WPI comparison result between $PP^*$ and $P_2P_2^*$ by \Cref{thm:K*:P*toP2*}. Finally, applying \Cref{Thm:K*:PtoP*} once more yields the WPI for $P_2^*P_2$.
    
    \begin{theorem}
        \label{thm:K*:weakPtotildeP2}
        Under \Cref{Asp:P2} and \Cref{Asp:P}, the operator $P_2^*P_2$ satisfies a $K^*$-WPI with $K^*=2K_2^*\left(\frac{1}{2}K_0^*( v/4)\right)$, where $K^*_0$ and $K^*_2$ correspond to $\beta_0$ and $\beta_2=\int_\mathcal{Y} \beta_2(s,y)\Pi_Y(\dif y)$, respectively.
    \end{theorem}
    \begin{proof}
        The following WPIs can be obtained directly by \Cref{Asp:P} and \Cref{lem:P1_to_tildeP1},
        \begin{align*}
        \frac{\mathcal{E}_\Pi\left(P^*P, f\right)}{\|f\|^2_{\mathrm{osc}}} \geqslant K_0^*\left(\frac{\|f\|_\Pi^2}{\|f\|^2_{\mathrm{osc}}}\right),\quad 
        \frac{\mathcal{E}_\Pi(H_2,f)}{\|f\|^2_{\mathrm{osc}}}\geqslant K_2^*\left(\frac{\mathcal{E}_\Pi\left(G_2, f\right)}{\|f\|^2_{\mathrm{osc}}}\right),
    \end{align*}
    where $K_i^*(v)=\sup_{u\geqslant 0}\{uv - u \cdot \beta_i(1 / u)\}$, $i=0,2$, and $\beta_2(s)=\int_\mathcal{Y} \beta_2(s,y)\Pi_Y(\dif y)$.

    Then, we can obtain the following WPIs by \Cref{Thm:K*:PtoP*} and \Cref{thm:K*:P*toP2*}:
    \begin{align*}
        \frac{\mathcal{E}_\Pi\left(PP^*, f\right)}{\|f\|^2_{\mathrm{osc}}} & \geqslant K_0^*\left(\frac{1}{2}\cdot\frac{\|f\|_\Pi^2}{\|f\|^2_{\mathrm{osc}}}\right),\\
        \frac{\mathcal{E}_\Pi\left(P_2{P_2}^*, f\right)}{\|f\|^2_{\mathrm{osc}}} & \geqslant 2 \cdot K_2^*\left(\frac{1}{2}\cdot\frac{\mathcal{E}_\Pi\left(PP^*, f\right)}{\|f\|^2_{\mathrm{osc}}}\right).
    \end{align*}

    Then, by chaining rule \Cref{Thm:chaining}, we can derive the following $K^*$-WPI:
    \begin{align*}
        \frac{\mathcal{E}_\Pi\left(P_2{P_2}^*, f\right)}{\|f\|^2_{\mathrm{osc}}} & \geqslant 2\cdot K_2^*\left(\frac{1}{2}\cdot K^*_0 \left(\frac{1}{2}\cdot \frac{\|f\|_\Pi^2}{\|f\|^2_{\mathrm{osc}}}\right)\right).
    \end{align*}
    Next, we can obtain the following $K^*$-WPI by \Cref{Thm:K*:PtoP*}:
\begin{align*}
    \frac{\mathcal{E}_\Pi\left(P_2^*P_2, f\right)}{\|f\|^2_{\mathrm{osc}}} & \geqslant 2\cdot K_2^*\left(\frac{1}{2}\cdot K^*_0 \left(\frac{1}{4}\cdot \frac{\|f\|_\Pi^2}{\|f\|^2_{\mathrm{osc}}}\right)\right).
\end{align*}
    \end{proof} 

\begin{remark}
    We did not establish a relative WPI comparison between $P^*P$ and $P_2^*P_2$, so we cannot directly establish a relative comparison result for their convergence bounds. However, the WPI framework still provides a way to derive the convergence bound of the 2-MwG sampler, given that the convergence bound of the MwG sampler is known.
\end{remark}

Now we give the proof of \Cref{thm:weakPtotildeP}.
\begin{proof}[Proof of \Cref{thm:weakPtotildeP}]
\label{prf:thm:weakPtotildeP}
By \Cref{thm:K*:weakPtotildeP2} we can derive the following $K^*$-WPI for $P_2^*P_2$:
\begin{align*}
    \frac{\mathcal{E}_\Pi\left(P_2^*P_2, f\right)}{\|f\|^2_{\mathrm{osc}}} & \geqslant 2\cdot K_2^*\left(\frac{1}{2}\cdot K^*_0 \left(\frac{1}{4}\cdot \frac{\|f\|_\Pi^2}{\|f\|^2_{\mathrm{osc}}}\right)\right).
\end{align*}

The following $K^*$-WPI can be obtained directly by \Cref{lem:P1_to_tildeP1},
\begin{align*}
    \frac{\mathcal{E}_\Pi(H_1,f)}{\|f\|^2_{\mathrm{osc}}}&\geqslant K_1^*\left(\frac{\mathcal{E}_\Pi\left(G_1, f\right)}{\|f\|^2_{\mathrm{osc}}}\right),
\end{align*}
where $K_1^*(v)=\sup_{u\geqslant 0}\{uv - u \cdot \beta_1(1 / u)\}$, and $\beta_1(s)=\int_\mathcal{X} \beta_1(s,x)\Pi_X(\dif x)$.

Finally, by \Cref{thm:K*:P1tildeP2totildeP1tildeP2} and \Cref{Thm:chaining}, we can obtain the following $K^*$-WPIs:
\begin{align*}
    \frac{\mathcal{E}_\Pi\left(P_{1,2}^*P_{1,2}, f\right)}{\|f\|^2_{\mathrm{osc}}} & \geqslant 2\cdot K_1^*\left(\frac{1}{2}\cdot \frac{\mathcal{E}_\Pi\left(P_2^*P_2, f\right)}{\|f\|^2_{\mathrm{osc}}}\right),\\
    \frac{\mathcal{E}_\Pi\left(P_{1,2}^*P_{1,2}, f\right)}{\|f\|^2_{\mathrm{osc}}} & \geqslant 2\cdot K_1^*\left(K_2^*\left(\frac{1}{2}\cdot K^*_0 \left(\frac{1}{4}\cdot \frac{\|f\|_\Pi^2}{\|f\|^2_{\mathrm{osc}}}\right)\right)\right).
\end{align*}
\end{proof}

\begin{remark}
    \Cref{Thm:PtoP*}  demonstrates that, in the subgeometric case, the convergence behavior of the original kernel and its adjoint are comparable up to a constant. Consequently, reversing the order of comparison (namely, (Step 1) replace $G_1$ with its corresponding MwG version $H_1$; (Step 2) swap the order so that $G_2$ becomes the first component; (Step 3) replace $G_2$ with its corresponding MwG version $H_2$; and (Step 4) swap again to restore $H_1$ as the first component), the expression $2K_2^*\left(K_1^*\left(\frac{1}{2}K_0^*( v/4)\right)\right)$ can also be considered a valid choice for $K^*$. However, it is not immediately clear which formulation yields better performance in different scenarios.
\end{remark}

    It is straightforward to see that $K^*$ can be expressed as a composition of $K_2^*$, $K_1^*$, and $K_0^*$ up to constant factors. Thus we can obtain the rate function $F^{-1}$ via \Cref{thm:ConstantOnBeta}.

    If the exact kernel $P^*P$ satisfies a SPI as \Cref{Asp:P_SPI}, from \Cref{thm:strongPtoP^*}, we can obtain a stronger comparison as shown in \Cref{thm:strongPtotildeP}; the proofs of the two results are in \Cref{Prf:thm:strongPtoP^*}.
    \begin{assumption}
    \label{Asp:P_SPI}
        We assume that the kernel $P^*P$ satisfies a SPI: there is $\gamma>0$ with
        \begin{align*}
            \gamma \cdot\|f\|_\Pi^2 \leqslant \mathcal{E}_\Pi(P^*P, f), \quad \forall f \in \mathrm L_0^2(\Pi).
        \end{align*}
    \end{assumption}

    \begin{lemma}
    \label{thm:strongPtoP^*}
        Suppose \Cref{Asp:P_SPI} holds,
        then, the kernel $PP^*$ satisfies a SPI with the same constant.
    \end{lemma}
    
    \begin{corollary}
        \label{thm:strongPtotildeP}
        Under the \Cref{Asp:P1},\Cref{Asp:P2} and \Cref{Asp:P_SPI}, the operator $P_{1,2}^*P_{1,2}$ satisfies the $K^*$-WPI with $K^*=2\cdot K_1^*\circ K_2^*\left(\frac{\gamma}{4}\cdot v\right)$, where $K^*_1$ and $K^*_2$ correspond to $\beta_1(s)=\int_\mathcal{X} \beta_1(s,x)\Pi_X(\dif x)$ and $\beta_2=\int_\mathcal{Y} \beta_2(s,y)\Pi_Y(\dif y)$, respectively.
    \end{corollary}

\subsection{Marginal MwG sampler with one intractable conditional}

In certain marginal settings, it is possible to sample the first component exactly from its conditional distribution while approximating the second component using an MCMC scheme. The results from the previous section introduce additional constant factors, potentially leading to loose bounds on the convergence rate. In this section, we present an alternative approach to obtain the WPI for the 2-MwG sampler, aiming to provide a sharper characterization of its convergence behavior.

We consider the relationship between the kernel of DG sampler $P^*P$ and its $X$-marginal kernel $P_X^*P_X$ as defined in \eqref{eq:marginalPX}.

\begin{theorem}
\label{thm:PtoP_X}
    Suppose that the kernel $P^*P$ satisfies a weak Poincar{\'{e}} inequality with function $\beta:(0,\infty)\rightarrow(0,\infty)$;
    \begin{align*}
         \|f\|_\Pi^2 & \leqslant s \cdot \mathcal{E}_\Pi\left(P^*P, f\right)+\beta(s) \cdot \|f\|^2_\mathrm{osc}\quad \forall s>0, f \in \mathrm{L}^2_0(\Pi).
    \end{align*}
    Then the kernel $P_X^* P_X$ satisfies a WPI with the same function $\beta$:
    \begin{align}
        \|g\|_{\Pi_X}^2 & \leqslant s \cdot \mathcal{E}_{\Pi_X}\left(P_X^*P_X, g\right)+\beta(s) \cdot \|g\|^2_\mathrm{osc},\quad \forall s>0,g \in \mathrm{L}^2_0(\Pi_X).
        \label{eq:WPI_PX}
    \end{align}
\end{theorem}
\begin{proof}
     Let $g \in \mathrm{L}^2_0 (\Pi_X)$ and $f_g \in \mathrm{L}^2_0 (\Pi)$ be such that $f_g(x, y) = g(x)$. It is clear that $\|f_g\|^2_\mathrm{osc}=\|g\|^2_\mathrm{osc}$. 
    We now show that $\|f_g\|^2_{\Pi}=\|g\|^2_{\Pi_X}$: 
    \begin{align*}
    \begin{split}
        \|f_g\|^2_{\Pi}&=\frac{1}{2}\int(f_g(x,y)-f_g(x',y'))^2\Pi(\dif x,\dif y)\Pi(\dif x',\dif y')\\
        &=\frac{1}{2}\int(g(x)-g(x'))^2\Pi_X(\dif x)\Pi_X(\dif x')\\
        &=\|g\|^2_{\Pi_X}.
    \end{split}
    \end{align*}
    Then we show $\mathcal{E}_\Pi(P^*P,f_g)=\mathcal{E}_{\Pi_{X}}(P_{X}^*P_{X},g).$ Since $G_1$ and $G_2$ are projections, we have $P^*=G_2G_1$ and $P^*P=G_2G_1G_2$.
    \begin{align*}
    \begin{split}
        \mathcal{E}_\Pi(P^*P,f_g)&=\frac{1}{2}\int(f_g(x,y)-f_g(x',y'))^2P^*P((x,y),(\dif x',\dif y'))\Pi(\dif x,\dif y)\\
        &=\frac{1}{2}\int(g(x)-g(x'))^2 \Pi_{X|Y}(\dif x'|y')\Pi_{Y|X}(\dif y'|x'')\Pi_{X|Y}(\dif x''|y)\Pi(\dif x,\dif y)\\
        &=\frac{1}{2}\int(g(x)-g(x'))^2 \int_\mathcal{Y}\Pi_{X|Y}(\dif x'|y')\Pi_{Y|X}(\dif y'|x'')\int_\mathcal{Y}\Pi_{X|Y}(\dif x''|y)\Pi_{Y|X}(\dif y|x)
        \Pi_X(\dif x)\\
        &=\frac{1}{2}\int(g(x)-g(x'))^2 P_X^*P_X(x,\dif x')
        \Pi_X(\dif x)\\
        &=\mathcal{E}_{\Pi_{X}}(P_{X}^*P_{X},g).\\
    \end{split}
    \end{align*}
    
    This implies for any $g \in \mathrm{L}^2_0(\Pi_X)$, the WPI \eqref{eq:WPI_PX} holds with the same $\beta$.
\end{proof}

Then by \Cref{Thm:rate}, we can derive $F^{-1}(n)$ to bound the convergence of $P_X^*P_X$. This implies the convergence for the $X$-marginal chain $P_{X}$ is as at least fast as the joint chain $P$.

We now determine how to bound the convergence of the joint chain in terms of that of the $X$-marginal chain.
\begin{theorem}
\label{thm:P_XtoP}
    Suppose that kernel $P_X^*P_X$ satisfies a weak Poincar{\'{e}} inequality;
     \begin{align*}
        \|g\|_{\Pi_X}^2 & \leqslant s \cdot \mathcal{E}_{\Pi_X}\left(P_X^*P_X, g\right)+\beta(s) \cdot \|g\|^2_\mathrm{osc},\quad \forall s>0, g \in \mathrm{L}^2_0(\Pi_X).
    \end{align*}
    Then, for any $f \in \mathrm{L}^2_0(\Pi)$ such that $0 < \|f\|_\mathrm{osc} < \infty$, the $K^*$-WPI holds for every $Pf$:
    \begin{align*}
        \frac{\mathcal{E}_{\Pi}(P^*P,Pf)}{\|Pf\|^2_\mathrm{osc}}\geqslant K^*\left(\frac{\|Pf\|_{\Pi}^2}{\|Pf\|^2_\mathrm{osc}}\right),
    \end{align*}
    where $K^*$ is as in \Cref{def:WPI}. 
    Hence, for any $n \in \mathbb{N_+} \backslash \{1\}$,
    \begin{align*}
    \|P^nf \|^2_{\Pi} \leqslant \|f\|^2_\mathrm{osc} \cdot F^{-1}(n-1),
    \end{align*}
    where $F$ is as in \Cref{Thm:rate}.
\end{theorem}
\begin{proof}
    Let $f\in \mathrm{L}^2_0 (\Pi)$ and  $g_f(x)=(Pf)(x,y)$ for any $ y \in \mathcal{Y}$; this right-hand quantity in fact does not depend upon $y$. It is then easy to verify that $g_f\in \mathrm{L}^2_0 (\Pi_X)$, and $\|g_f\|_{\Pi_X}^2=\|Pf\|_{\Pi}^2$. For $(x, y) \in \mathcal{X}\times\mathcal{Y}$, We first demonstrate a crucial transfer that establishes the relationship between $Pf (x, y)$ and $ P_{X} g_f (x)$.
    \begin{align*}
        P_{X} g_f (x)&=\int_\mathcal{X} g_f (x') \int_\mathcal{Y} \Pi_{X|Y}(\dif x'|y)\Pi_{Y|X}(\dif y|x)\\
        &=\int_\mathcal{X} Pf(x',y') \int_\mathcal{Y} \Pi_{X|Y}(\dif x'|y)\Pi_{Y|X}(\dif y|x)\\
        &=\int_\mathcal{X} \int_{\mathcal{X},\mathcal{Y}} f(x'',y'')\Pi_{X|Y}(\dif x''|y'')\Pi_{Y|X}(\dif y''|x') \int_\mathcal{Y} \Pi_{X|Y}(\dif x'|y)\Pi_{Y|X}(\dif y|x)\\
        &=\int_{\mathcal{X},\mathcal{Y}}\left[\int_{\mathcal{X},\mathcal{Y}} f(x'',y'')\Pi_{X|Y}(\dif x''|y'')\Pi_{Y|X}(\dif y''|x')\right] \Pi_{X|Y}(\dif x'|y)\Pi_{Y|X}(\dif y|x)\\
        &=P^2f(x,y),
    \end{align*}
    where the second equality holds for all $y' \in \mathcal{Y}$, since in fact $Pf(x',y')$ only depends on the value of $x'$.
    
    Hence, based on $P_{X} g_f (x)=P^2f(x,y)$, it is easy to check $\|P_{X}g_f\|_{\Pi_X}^2=\|P^2f\|_{\Pi}^2$. On the other hand, we have $\|g_f\|_{\Pi_X}^2=\|Pf\|_{\Pi}^2$.
    Therefore, 
    \begin{align*}
    \begin{split}
        \mathcal{E}_{\Pi_X}(P_{X}^*P_{X},g_f)&=\|g_f\|_{\Pi_X}^2-\|P_{X}g_f\|_{\Pi_X}^2\\
        &=\|Pf\|_{\Pi}^2-\|P^2f\|_{\Pi}^2\\
        &=\mathcal{E}_{\Pi}(P^*P,Pf).
    \end{split}
    \end{align*}
    Hence for $g_f$, the $K^*$-WPI shows: $\frac{\mathcal{E}_{\Pi_X}(P_{X}^*P_{X},g_f)}{\|g_f\|^2_\mathrm{osc}}\geqslant K^*\left(\frac{\|g_f\|_{\Pi_X}^2}{\|g_f\|^2_\mathrm{osc}}\right)$, that is equivalent to $\frac{\mathcal{E}_{\Pi}(P^*P,Pf)}{\|Pf\|^2_\mathrm{osc}}\geqslant K^*\left(\frac{\|Pf\|_{\Pi}^2}{\|Pf\|^2_\mathrm{osc}}\right)$, where $K^*(v)$ corresponds to $\beta(s)$.

    Since $P^{n-2}$ ia a Markov operator for $n \geqslant 2$, by \Cref{lem:f_osc}, we have
    \begin{align*}
    \frac{\mathcal{E}_\Pi\left(P^* P, P^{n-1} f\right)}{\|f\|^2_\mathrm{osc}} \geqslant K^*\left(\frac{\left\|P^{n-1} f\right\|_\Pi^2}{\|f\|^2_\mathrm{osc}}\right) .
    \end{align*}

    Since by \Cref{thm:K*Properties}, $v \rightarrow 1/K^*(v)$ is decreasing and $k \rightarrow \|P^kf\|_\Pi^2$ decreasing, for $n \geqslant 2$,
    \begin{align*}
    \begin{split}
        F\left(\left\|P^n f\right\|_\Pi^2 / \|f\|^2_\mathrm{osc}\right)-F\left(\left\|P^{n-1} f\right\|_\Pi^2 / \|f\|^2_\mathrm{osc}\right) 
        & =\int_{\left\|P^n f\right\|_\Pi^2 / \|f\|^2_\mathrm{osc}}^{\left\|P^{n-1} f\right\|_\Pi^2 / \|f\|^2_\mathrm{osc}} 1 / K^*(v) \mathrm{d} v \\
        & \geqslant \frac{\left\|P^{n-1} f\right\|_\Pi^2-\left\|P^n f\right\|_\Pi^2}{K^*\left(\left\|P^{n-1} f\right\|_\Pi^2 / \|f\|^2_\mathrm{osc}\right) \|f\|^2_\mathrm{osc}} \\
        & =\frac{\mathcal{E}_\Pi\left(P^* P, P^{n-1} f\right) / \|f\|^2_\mathrm{osc}}{K^*\left(\left\|P^{n-1} f\right\|_\Pi^2 / \|f\|^2_\mathrm{osc}\right)} \geqslant 1,
    \end{split}
    \end{align*}
    where $F(x):=\int_{x}^{\frac{1}{4}} \frac{\dif v}{K^*(v)}$. As a result, for $n \geqslant 2$, it holds that
    \begin{align*}
        F\left(\left\|P^n f\right\|_\Pi^2 / \|f\|^2_\mathrm{osc}\right)-F\left(\left\|P f\right\|_\Pi^2 / \|f\|^2_\mathrm{osc}\right) \geqslant n-1,
    \end{align*}
    from which combining with $F^{-1}$ is strictly decreasing, we obtain
    \begin{align*}
    \begin{split}
        \left\|P^n f\right\|_\Pi^2  &\leqslant \|f\|^2_\mathrm{osc} \cdot F^{-1}\left(F\left(\left\|P f\right\|_\Pi^2 / \|f\|^2_\mathrm{osc}\right)+n-1\right)\\
        &\leqslant \|f\|^2_\mathrm{osc} \cdot F^{-1}\left(n-1\right).
    \end{split}
    \end{align*}
\end{proof}
This shows the convergence bound of $P$ can be obtained by the convergence bound of $P_X$ directly. Moreover, even though the convergence rate of $P$ is slower than $P_X$, they asymptotically converge at a comparable rate.  This result coincides with Proposition 7.1 in \cite{ascolani2025mixing}, which is expressed through entropy contraction.

\begin{remark}
    If the convergence rate of either the joint chain or the marginal chain is known, \Cref{thm:PtoP_X} and \Cref{thm:P_XtoP} offer a framework for deriving the convergence rate of the other. Specifically, in the case of exponential convergence, if the joint chain $P$ exhibits exponential convergence, then the marginal chain $P_X$ will also converge at the same exponential rate. Conversely, if the marginal chain $P_X$ converges exponentially, the joint chain $P$ will do so as well. Moreover, the convergence rates will be preserved between the joint and marginal chains, up to a possible multiplicative constant in \eqref{equ:exp}. In fact, the study of exponential convergence for DG samplers is well-established, and similar conclusions based on the operators' norm can be found in Theorem 3.2 in \cite{liu1994covariance}; see also Lemma 3.2 in \cite{qin2022convergence}. \Cref{thm:PtoP_X} and \Cref{thm:P_XtoP} extend this analysis to subgeometric cases, where the operators' norms may not be applicable.
\end{remark}

Similarly, we have the same argument between 2-MwG sampler kernel $P_2^*P_2$ and its $X$-marginal kernel $\bar{P}_{X}^*\bar{P}_{X}$, from the following propositions, proven in \Cref{Prf:thm:ApproxPtoP_X}.
\begin{proposition}
\label{thm:ApproxPtoP_X}
    Suppose that kernel $P_2^*P_2$ satisfies a WPI with function $\beta$, 
    Then, the kernel $\bar{P}_{X}^*\bar{P}_{X}$ satisfies a WPI with the same $\beta$.
\end{proposition}
\begin{proposition}
\label{thm:ApproxP_XtoP}
    Suppose that kernel $\bar{P}_X^*\bar{P}_X$ satisfies a weak Poincar{\'{e}}  inequality with the function $\beta$.
    Hence, for any $f \in \mathrm{L}^2_0(\Pi)$ such that $0 < \|f\|_\mathrm{osc} < \infty$, any $n \in \mathbb{N_+}\backslash \{1\}$,
    \begin{align*}
    \|P_2^nf \|^2_{\Pi} \leqslant \|f\|^2_\mathrm{osc} \cdot F^{-1}(n-1),
    \end{align*}
    where $F: (0, \frac{1}{4}] \rightarrow \mathbb{R}$ is derived by $\beta$ as \Cref{Thm:rate}. 
\end{proposition}

\Cref{thm:PtoP_X} and \Cref{thm:ApproxP_XtoP} offer a framework for obtaining the convergence bound of $P_2$. These theorems allow us to translate the comparison from the joint chains to the $X$-marginal chains. We can also obtain a comparison between $P_X$ and $\Bar{P}_X$, which is similar to the comparison between $G_2$ and $H_2$ in \Cref{lem:P1_to_tildeP1}; the proofs of the following two results are in \Cref{Prf:lem:P_xreversible}.

\begin{lemma}
    \label{lem:P_xreversible}
    The exact $X$-marginal kernel $P_X$ is $\Pi_X$-reversible and positive. Under \Cref{Asp:P2}, the 2-MwG $X$-marginal kernel $\bar{P}_X$ is also $\Pi_X$-reversible and positive.
\end{lemma}

\begin{theorem}
\label{thm:XDGtoXHDG}
    Under \Cref{Asp:P2}, we have the following WPI comparison for $P_X$ and $\bar{P}_X$: 
    \begin{align*}
        \mathcal{E}_{\Pi_X}(P_X,g) \leqslant s \cdot \mathcal{E}_{\Pi_X}(\bar{P}_X,g)+\beta(s) \cdot \|g\|^2_\mathrm{osc},\quad \forall s > 0, g \in \mathrm{L}^2(\Pi_X).
    \end{align*}
   where $\beta: (0, \infty) \rightarrow [0, \infty)$ is given by $\beta(s):= \int_\mathcal{Y} \beta_2(s,y)\Pi_Y(\dif y)$. 
\end{theorem}

\begin{figure}[H] 
    \centering 
    \begin{tikzpicture}[>=stealth]
        \node (DG) at (0,2) {$P$};
        \node (XDG) at (4,2) {$P_X$};
        \node (HDG) at (0,0) {$P_2$};
        \node (XHDG) at (4,0) {$\bar{P}_X$};
        \draw [->] (DG) -- node[above] {\Cref{thm:PtoP_X}} (XDG);
        \draw [->] (XDG) -- node[right] {\Cref{thm:XDGtoXHDG}} (XHDG);
        \draw [->] (XHDG) -- node[below] {\Cref{thm:ApproxP_XtoP}} (HDG);
        
        \draw [->, dashed] (3.6,1.9) -- node[below] {\Cref{thm:P_XtoP}} (0.3,1.9);
        \draw [->, dashed] (0.3,0.1) -- node[above] {\Cref{thm:ApproxPtoP_X}} (3.6,0.1);
    \end{tikzpicture}
    \caption{Relationship among Gibbs samplers and their X-marginal version.} 
    \label{fig:PandPX} 
\end{figure}
The relations that we have described so far can be summarized in \Cref{fig:PandPX}. The solid arrows give us a path to obtain the convergence bound of $P_2$. We state it as the following corollary.
\begin{corollary}
    \label{thm:weakPtotildeP2}
        Under the \Cref{Asp:P2} and \Cref{Asp:P} we have the following convergence property for $P_2$: for any $n \in \mathbb{N_+}\backslash \{1\}$,
    \begin{align}
    \|P_2^nf \|^2_{\Pi} \leqslant \|f\|^2_\mathrm{osc} \cdot F^{-1}(n-1),
    \end{align}
    where $F$ corresponds to $K^*(v) =K_2^*\circ \frac{1}{2}K_0^*(v)$, and $K_0$ and $K_2$ are associated with $\beta_0$ and $\beta_2(s)=\int_\mathcal{Y} \beta(s,y)\Pi_Y(\dif y)$, respectively.
\end{corollary}
\begin{proof}
    The following WPIs can be obtained from \Cref{thm:PtoP_X}, \Cref{thm:XDGtoXHDG}: for all $g \in \mathrm{L}^2_0(\Pi_X),$
    \begin{align*}
        \begin{split}
            \|g\|_{\Pi_X}^2 & \leqslant s \cdot \mathcal{E}_{\Pi_X}\left(P_{X}^*P_{X}, g\right)+\beta_0(s) \cdot \|g\|^2_\mathrm{osc},\\
            \mathcal{E}_{\Pi_X}(P_X,g) & \leqslant s \cdot \mathcal{E}_{\Pi_X}(\bar{P}_X,g)+\beta_2(s) \cdot \|g\|^2_\mathrm{osc},
        \end{split}
    \end{align*}
    where $\beta_2(s)=\int_\mathcal{Y} \beta_2(s,y)\Pi_Y(\dif y)$.
    
    By \Cref{lem:PtoP^2} and \Cref{lem:P^2toP}, for every $s_1,s_2>0$, we have 
    \begin{align*}
        \|g\|_{\Pi_X}^2 & \leqslant s_1 \cdot \mathcal{E}_{\Pi_X}\left(P_{X}^*P_{X}, g\right)+\beta_0(s_1) \cdot \|g\|^2_\mathrm{osc}\\
         & \leqslant 2s_1 \cdot \mathcal{E}_{\Pi_X}(P_X,g)+\beta_0(s_1) \cdot \|g\|^2_\mathrm{osc}\\
         & \leqslant 2s_1 \cdot \left( s_2 \cdot \mathcal{E}_{\Pi_X}(\bar{P}_X,g)+\beta_2(s_2) \cdot \|g\|^2_\mathrm{osc}\right)+\beta_0(s_1) \cdot \|g\|^2_\mathrm{osc}\\
         &\leqslant 2s_1  s_2  \cdot \mathcal{E}_{\Pi_X}(\bar{P}_X^*\bar{P}_X,g)+ \left(2s_1\beta_2(s_2) +\beta_0(s_1) \right)\cdot \|g\|^2_\mathrm{osc}.
    \end{align*}
    Let $\beta(s):=\inf\{s_1\beta_2(s_2)+\beta_0'(s_1)|s_1>0, s_2>0, s_1s_2=s,s>0\}$, where $\beta_0'(s)=\beta_0(s/2)$.
    Let $K_0^*(v)=\sup_{u\geqslant 0}\{uv - u \cdot \beta_0(1 / u)\}$ and $\bar{K}_0^*(v)=\sup_{u\geqslant 0}\{uv - u \cdot \beta_0'(1 / u)\}$. By \Cref{thm:ConstantOnBeta}, since $\beta_0'(s)=\beta_0(s/2)$, we have $\bar{K}_0^*(v)=\frac{1}{2} K_0^*(v)$. Then, let $K_2^*(v)=\sup_{u\geqslant 0}\{uv - u \cdot \beta_2(1 / u)\}$, applying \Cref{Thm:chaining} we can derive the WPI for $\bar{P}_X^*\bar{P}_X$ with 
    \begin{align*}
        K^*(v):=K^*_2 \circ \bar{K}_0^*(v)= K^*_2 \circ \frac{1}{2} K_0^*(v).
    \end{align*}
    
    Hence, it holds that $\frac{\mathcal{E}_{\Pi_X}(\bar{P}_X^*\bar{P}_X,g)}{\|g\|^2_\mathrm{osc}} \geqslant K^*\left(\frac{\|g\|_{\Pi_X}^2}{\|g\|^2_\mathrm{osc}}\right)$.
    Then by \Cref{thm:ApproxP_XtoP}, we have
    \begin{align*}
    \|P_2^nf \|^2_{\Pi} \leqslant \|f\|^2_\mathrm{osc} \cdot F^{-1}(n-1),
    \end{align*}
    where $F: (0, \frac{1}{4}] \rightarrow \mathbb{R}$ is the decreasing convex and invertible function as $F(x):=\int_{x}^{\frac{1}{4}} \frac{\dif v}{K^*(v)}$.

\end{proof}

\begin{remark}
    In this case, we obtained $K^*(v)=K_2^*(\frac{1}{2}K_0^*(v))$. Referring to \Cref{thm:K*:weakPtotildeP2}, we obtained $\tilde{K}^*(v)=2K_2^*(\frac{1}{2}K_0^*(v/4))$. Hence $\tilde{K}^*(v)=2K^*(v/4)$. Comparing the associated convergence rates of $F^{-1}$, \Cref{thm:weakPtotildeP2} yields a rate of $F^{-1}(n-1)$, while \Cref{thm:K*:weakPtotildeP2} gives $\tilde{F}^{-1}(n)\leqslant 4\cdot F^{-1}(n/2)$ by applying \Cref{thm:ConstantOnBeta}, where $F$ and $\tilde{F}$ correspond to $K^*$ and $\tilde{K}^*$, respectively. Generally, the latter converges more slowly than the former, suggesting that, under a marginal framework, \Cref{thm:weakPtotildeP2}  yields  a better convergence rate bound compared to \Cref{thm:K*:weakPtotildeP2}. 
\end{remark}

\begin{remark}
\label{rmk:lack_WPI_P2}
    Moreover, although \Cref{thm:weakPtotildeP2} provides a method to obtain the convergence rate bound of $P_2$, it does not yield a WPI for $P_2^*P_2$, nor does it establish a comparison between the DG and 2-MwG samplers. Consequently, we cannot directly apply the chaining WPI rule \Cref{Thm:chaining} with \Cref{thm:K*:PtoP1} to derive the WPI for $P_{1,2}^*P_{1,2}$. This motivates the use of \Cref{Thm:K*:PtoP*} to obtain the WPI for $P_2^*P_2$,as presented in \Cref{thm:K*:weakPtotildeP2} in \Cref{Sec: HDG}.
\end{remark}

\begin{remark}
\label{rmk:lim_two}
    Our method is currently restricted to the two-component case. As illustrated in \Cref{fig:PandP2}, the comparison framework is mainly based on the first component is replaced by its MwG version. The gap is the lack of WPI between the DG and 2-MwG as mentioned in \Cref{rmk:lack_WPI_P2}. In fact, as many previous works have pointed out, the \emph{two}-component Gibbs sampler possesses a particularly unique structure \cite{andrieu2016random, qin2024analysis}. On the other hand, in cases of greater than two components, the scenario is much more complex.  
    For example, for the three-component kernel $G_1G_2G_3$, our result can cover $H_1G_2G_3$ or $G_1G_2H_3$ but not $G_1H_2G_3$. When components are more than three, our comparison results are even more limited.
\end{remark}

\begin{remark}
\label{rmk:update order1}
    Under the assumption that the MwG kernels $H_{1|x}$ and $H_{2|y}$ are self-adjoint, we can show that $H_1H_2$ and $H_2H_1$ have essentially the same convergence behavior based on \Cref{Thm:PtoP*}, since one is the adjoint of the other. Hence, in the two-component case, the order of the updates does not significantly affect the convergence behavior. 
\end{remark}
    \begin{remark}
\label{rmk:update order2}
    A natural question asks about the comparison between random and deterministic scans. For the  standard Gibbs sampler, \cite{qin2022convergence} shows that deterministic scan always has better convergence rate than random scan at least in the geometric case, and similar comparisons should be available for sub-geometric two-component Gibbs samplers. However, to our knowledge there is currently no general theoretical result for MwG, even in the two-component case.
\end{remark}

\section{Tensor products as Gibbs Samplers}
\label{Sec: tensor product}
We consider a special case of the MwG sampler under conditional independence. In this setting, the MwG sampler reduces to a \textit{tensor product} of two Markov chains. Assuming that each of these chains has a well-defined spectral gap, the spectral gap of the resulting MwG sampler is equal to the minimum of the individual spectral gaps \cite{brown1966spectra}. This leads to a more favorable convergence rate compared to the product of the spectral gaps, which appears in the bound provided by \Cref{thm:strongPtotildeP}.

More precisely, suppose that $(\mathcal{X},\mathcal{F}_\mathcal{X})$ and $(\mathcal{Y},\mathcal{F}_\mathcal{Y})$
are two measurable spaces with probability measures $\Pi_X$ and $\Pi_Y$, respectively. Let $\Pi=\Pi_X\otimes \Pi_Y$ be the resulting product measure on the product space $\mathcal{X}\otimes\mathcal{Y}$, equipped with the product $\sigma$-algebra.
Let $H_1$ and $H_2$ be two Markov chains targeting $\Pi_Y$ and $\Pi_X$, respectively. Then, the tensor product $H_1\otimes H_2$ is a Markov chain targeting $\Pi$, given by the independent coupling: 
$$
    [H_1\otimes H_2]((x,y), (\dif x',\dif y'))= H_1(\dif y'|y)H_2(\dif x'|x).
$$

For continuous-time Markov processes, the standard Poincar{\'{e}} inequality is stable under products. More precisely, as shown in \cite[Proposition 4.3.1]{bakry2013analysis}, the product Markov triple satisfies a standard Poincar{\'{e}} inequality with a constant equal to the \textit{minimum} of the corresponding constants for the two component processes. 
For discrete-time chains, we similarly have that the spectrum of the tensor product operator is the Minkownski product of the spectra of the component two operators \cite{brown1966spectra}.
This result for discrete-time chains is generally not informative for subgeometric chains with no spectral gap. However, by utilising WPIs, we can derive an analogous result for subgeometric chains, where the resulting WPI $\beta$-function for the tensor product is simply the sum of the individual functions.

\begin{theorem}
\label{thm: tensor product}
Let $H_1$ and $H_2$ be two independent, reversible and positive Markov kernels targeting $\Pi_Y$ and $\Pi_X$, respectively. Suppose that $H_1^2$ and $H_2^2$ each satisfy a WPI: for $\beta_1,\beta_2$ as in \Cref{def:WPI},
\begin{align}\label{eq:tensor_prod_1}
    \mathrm{Var}_{\Pi_Y}(h) & \leqslant s \cdot \mathcal{E}_{\Pi_Y}(H_1^2,h)+ \beta_1(s) \cdot \|h\|_\mathrm{osc}^2,\quad \forall s>0,  h \in \mathrm{L}^2(\Pi_Y),\\
    \mathrm{Var}_{\Pi_X}(g) & \leqslant s \cdot \mathcal{E}_{\Pi_X}(H_2^2,g)+ \beta_2(s) \cdot \|g\|_\mathrm{osc}^2,\quad \forall s>0, g \in \mathrm{L}^2(\Pi_X).\label{eq:tensor_prod_2}
\end{align}

Then, the tensor product chain $H_1 \otimes H_2$ satisfies a WPI:  for all $ s>0$, $f \in \mathrm{L}^2_0(\Pi_X \otimes \Pi_Y)$,
\begin{align*}
    \|f\|_{\Pi}^2 &\leqslant s \cdot \mathcal{E}_{\Pi}\left(\left({H_1 \otimes H_2}\right)^2,f\right)+\beta(s) \cdot \|f\|_\mathrm{osc}^2,\\
    \beta(s) &:=\beta_1(s)+\beta_2(s).
\end{align*}
\end{theorem}

\begin{proof}
    We start with defining the embedded chains in the joint space:
    \begin{align*}
        \bar{H}_1((x,y),(\dif x', \dif y'))=H_1(\dif y' |y)\delta_x(\dif x'), \quad
        \bar{H}_2((x,y),(\dif x', \dif y'))=H_2(\dif x' |x)\delta_y(\dif y').
    \end{align*}
    We can thus realise $\bar{H}_1$ and $\bar{H}_2$ as operators on $\mathrm{L}^2_0(\Pi)$. 
    For the tensor product chain $H_1 \otimes H_2$, we can write it as $P_{1,2}=\bar{H}_1\bar{H}_2$.
    \begin{align*}
        \bar{H}_1\bar{H}_2((x,y),(\dif x', \dif y'))=H_1 \otimes H_2((x,y),(\dif x', \dif y'))=H_1(\dif y' |y)H_2(\dif x' |x).
    \end{align*}
    It is easy to check the operator $\bar{H}_1\bar{H}_2$ is reversible and positive, in particular $\bar{H}_1\bar{H}_2=\bar{H}_2\bar{H}_1$.

    For $f\in \mathrm{L}_0^2(\Pi)$, we denote $f_x(y):=f(x,y)$ and $f_y(x):=f(x,y)$ so that $f_x(y) \in \mathrm{L}^2(\Pi_Y)$ for $\Pi_X$-almost every $x$ and $f_y(x) \in \mathrm{L}^2(\Pi_X)$ for $\Pi_Y$-almost every $y$.
    
    We define the operators $\bar \Pi_X$ and $\bar \Pi_Y$ on the joint space $\mathrm{L}^2(\Pi)$ as following:
    \begin{align*}
        \bar \Pi_X f(x,y) = \int f(x',y') \Pi_X(\dif x')\delta_y(\dif y'),\quad
        \bar \Pi_Y f(x,y) = \int f(x',y') \Pi_Y(\dif y')\delta_x(\dif x').
    \end{align*}
    They can be viewed as two independent projection operators on the joint space $\mathrm{L}^2(\Pi)$.
    
    Since
    \begin{align*}
        \mathcal{E}_{\Pi}(\bar \Pi_Y,f)&=\frac{1}{2}\int (f(x,y)-f(x',y'))^2 \Pi_{Y}(\dif y')\delta_x(\dif x')\Pi(\dif x,\dif y)\\
        &=\frac{1}{2}\int (f_x(y)-f_x(y'))^2 \Pi_{Y}(\dif y')\Pi_Y(\dif y)\Pi_{X}(\dif x)\\
        &=\int_\mathcal X \mathrm{Var}_{\Pi_Y}(f_x)\Pi_X(\dif x),
    \end{align*}
    and
    \begin{align*}
        \mathcal{E}_{\Pi}(\bar{H}_1^2,f)&=\frac{1}{2}\int (f(x,y)-f(x',y'))^2 H_1^2(\dif y'| y)\delta_x(\dif x')\Pi(\dif x,\dif y)\\
        &=\frac{1}{2}\int (f_x(y)-f_x(y'))^2 H_{1}^2(\dif y'|y)\Pi_{Y}(\dif y)\Pi_X(\dif x)\\
        &=\int_\mathcal X \mathcal{E}_{\Pi_Y}((H_{1}^2,f_x)\Pi_X(\dif x),
    \end{align*}
    by \eqref{eq:tensor_prod_1}
    , we have
    \begin{align*}
        \mathcal{E}_{\Pi}(\bar \Pi_Y,f)
        &=\int_\mathcal{X} \mathrm{Var}_{\Pi_Y}(f_x) \Pi_X(\dif x)\\ & \leqslant s \cdot \int_\mathcal{X} \mathcal{E}_{\Pi_Y}(H_1^2,f_x)\Pi_X(\dif x)+ \beta_1(s) \cdot \|f\|_\mathrm{osc}^2 \\
        & = s \cdot \mathcal{E}_{\Pi}(\bar{H}_1^2,f)+ \beta_1(s) \cdot \|f\|_\mathrm{osc}^2.
    \end{align*}
    Similarly, using \eqref{eq:tensor_prod_2}, we have $\mathcal{E}_{\Pi}(\bar \Pi_X,f) \leqslant s \cdot \mathcal{E}_{\Pi}(\bar{H}_2^2,f)+ \beta_2(s) \cdot \|f\|_\mathrm{osc}^2$.

    Without loss of generality, after rearranging the WPI and rescaling the functions by the oscillation norm, we have
    \begin{align}
    \label{eq:beta_sup}
    \begin{split}
         \beta_1(s) \geqslant \sup_{\|f\|^2_\mathrm{osc}\leqslant1}\{ \mathcal{E}_\Pi(\bar \Pi_Y,f)-s \cdot \mathcal{E}_\Pi(H_1^2,f) \},\\
        \beta_2(s) \geqslant \sup_{\|f\|^2_\mathrm{osc}\leqslant1}\{ \mathcal{E}_\Pi(\bar \Pi_X,f)-s \cdot \mathcal{E}_\Pi(H_2^2,f) \}.
    \end{split}
    \end{align}
    
    By \Cref{thm: decomposition}, the decomposition of the variance of $f$ can be written as:
    \begin{align}
    \label{eq: variance_decomp}
        \|f\|_\Pi^2 &= \mathcal{E}_\Pi(\bar \Pi_Y,f)+\mathcal{E}_\Pi(\bar \Pi_X,\bar \Pi_Y f).
    \end{align}
    Now we claim $\mathcal{E}_\Pi(\bar \Pi_X,\bar \Pi_Y f) \leqslant \mathcal{E}_\Pi(\bar \Pi_X, f)$. 
    Fix a pair $(x,x') \in \mathcal{X} \times \mathcal{X}$, let $g_{x,x'}(y)=f(x,y)-f(x',y)$.
    Then by Jensen's inequality for the probability measure $\Pi_Y$,
    \begin{align*}
        \left(\int_\mathcal Y g_{x,x'}(y) \Pi_Y(\dif y)\right)^2 \leqslant \int_\mathcal Y g_{x,x'}(y)^2 \Pi_Y(\dif y).
    \end{align*}
    So for each $(x,x')$,
    \begin{align}
    \label{eq: dirichlet form ineq}
        &\left(\int_\mathcal Y f(x,y)\Pi_Y(\dif y)-\int_\mathcal Y f(x',y)\Pi_Y(\dif y)\right)^2 \leqslant \int_\mathcal Y \left(f(x,y)-f(x',y)\right) ^2 \Pi_Y(\dif y), \nonumber\\
        \Longrightarrow &\int_{\mathcal X,\mathcal Y}\int_{\mathcal X,\mathcal Y} \left(\int_\mathcal Y f(x,y)\Pi_Y(\dif y)-\int_\mathcal Y f(x',y)\Pi_Y(\dif y)\right)^2 \Pi_X(\dif x')\delta_y(\dif y')\Pi(\dif x, \dif y) \nonumber\\
        &\leqslant \int_{\mathcal X,\mathcal Y} \int_\mathcal Y \int_\mathcal X\left(f(x,y)-f(x',y')\right) ^2 \Pi_X(\dif x') \delta_y(\dif y')\Pi(\dif x, \dif y), \nonumber\\
        \Longrightarrow &\mathcal{E}_\Pi(\bar \Pi_X,\bar \Pi_Y f) \leqslant \mathcal{E}_\Pi(\bar \Pi_X, f).
    \end{align}
    Since $\Pi_Y$ and $\Pi_X$ are the invariant measures for $H_1$ and $H_2$, respectively, we have $\bar \Pi_Y \bar H_1=\bar \Pi_Y$. 
    Then, starting from the decompositions \eqref{eq: variance_decomp} and \Cref{thm: decomposition},
    \begin{align*}
        &\sup_{\|f\|_\mathrm{osc}^2\leqslant1} \left\{\|f\|_\Pi^2-s \cdot \mathcal{E}_\Pi(\bar{H}_2\bar{H}_1\bar{H}_1\bar{H}_2,f) \right\}\\
        &= \sup_{\|f\|_\mathrm{osc}^2\leqslant1} \left\{\mathcal{E}_\Pi(\bar{\Pi}_Y,f)+\mathcal{E}_\Pi(\bar{\Pi}_X,\bar{\Pi}_Y f)-s \cdot (\mathcal{E}_\Pi(\bar{H}_1^2,f)+\mathcal{E}_\Pi(\bar{H}_2^2,\bar{H}_1 f)) \right\}\quad \\
        &\leqslant \sup_{\|f\|_\mathrm{osc}^2\leqslant1} \left\{ \mathcal{E}_\Pi(\bar{\Pi}_Y,f)-s \cdot \mathcal{E}_\Pi(\bar{H}_1^2,f) \right\}+\sup_{\|f\|_\mathrm{osc}^2\leqslant1} \left\{ \mathcal{E}_\Pi(\bar{\Pi}_X,\bar{\Pi}_Y f)- s \cdot \mathcal{E}_\Pi(\bar{H}_2^2,\bar{H}_1 f) \right\}\\
        &\leqslant \beta_1(s)+\sup_{\|f\|_\mathrm{osc}^2\leqslant1} \left\{ \mathcal{E}_\Pi(\bar{\Pi}_X,\bar{\Pi}_Y \bar{H}_1 f)- s \cdot \mathcal{E}_\Pi(\bar{H}_2^2,\bar{H}_1 f) \right\} \quad \text{by \eqref{eq:beta_sup} and $\bar{\Pi}_Y \bar{H}_1=\bar{\Pi}_Y$}\\
        &\leqslant \beta_1(s)+\sup_{\|\bar{H}_1f\|_\mathrm{osc}^2\leqslant1} \left\{ \mathcal{E}_\Pi(\bar{\Pi}_X,\bar{\Pi}_Y \bar{H}_1 f)- s \cdot \mathcal{E}_\Pi(\bar{H}_2^2,\bar{H}_1 f) \right\} \quad \text{by $\|\bar{H}_1f\|_\mathrm{osc}^2\leqslant \|f\|_\mathrm{osc}^2$}\\
        &\leqslant \beta_1(s)+\sup_{\|f\|_\mathrm{osc}^2\leqslant1} \left\{ \mathcal{E}_\Pi(\bar{\Pi}_X,\bar{\Pi}_Y  f)- s \cdot \mathcal{E}_\Pi(\bar{H}_2^2,f) \right\}\\
        &\leqslant \beta_1(s)+\sup_{\|f\|_\mathrm{osc}^2\leqslant1} \left\{ \mathcal{E}_\Pi(\bar{\Pi}_X,f)- s \cdot \mathcal{E}_\Pi(\bar{H}_2^2,f) \right\} \quad \text{by \eqref{eq: dirichlet form ineq}} \\
        &\leqslant \beta_1(s)+ \beta_2(s)\quad \text{by \eqref{eq:beta_sup}}.
    \end{align*}
    
    It is clear that the function $\beta:=\beta_1+\beta_2$ satisfies the regularity conditions of the WPI definition for $\left({H_1 \otimes H_2}\right)^2$.
\end{proof}

\begin{corollary}
    Let $\{ H_i \}_{i=1}^n$ be $n$ independent Markov chains targeting  $\{ \Pi_i \}_{i=1}^n$, respectively. Suppose that for each $i = 1, \dots, n$, the chain $H_i^2$ satisfies the weak Poincar{\'{e}} inequality with function $\beta_i$, i.e., for all $f_i \in \mathrm{L}^2(\Pi_i)$,
    \begin{align*}
        \mathrm{Var}_{\Pi_i}(f_i) &\leqslant s \cdot \mathcal{E}_{\Pi_i}(H_i^2,f_i)+ \beta_i(s) \cdot \|f_i\|_\mathrm{osc}^2.
    \end{align*}
    
    Then, the tensor product chain $H = H_1 \otimes \cdots \otimes H_n$ satisfies the weak Poincar{\'{e}} inequality with function $\beta$, i.e., for all $f \in L_0^2(\Pi)$,
    \begin{align*}
       \|f\|_\Pi^2 &\leqslant s \cdot \mathcal{E}_{\Pi}(H^2,f)+ \beta(s) \cdot \|f\|_\mathrm{osc}^2,
    \end{align*}
    where $\beta(s)=\sum_{i=1}^n \beta_i(s)$.
\end{corollary}
\begin{proof}
    The proof follows directly by induction.
\end{proof}
The following corollary demonstrates that \Cref{thm: tensor product} is consistent with the standard Poincar{\'{e}} inequality, which characterizes the existence of a spectral gap. Specifically, it establishes that the minimum of the spectral gaps of two independent Markov chains provides a lower bound for the spectral gap of their tensor product chain. Since our primary interest lies in upper bounds on the convergence rate, this result implies that obtaining upper bounds based on spectral gaps for the individual chains directly yields an upper bound for the tensor product chain. The proof is in \Cref{prf:Thm:SPIproduct}. 

\begin{corollary}
\label{Thm:SPIproduct}
    Let $H_1$ and $H_2$ be two independent, reversible and positive Markov kernels targeting $\Pi_Y$ and $\Pi_X$, respectively, each satisfying Standard Poincar\'e inequalities with $c_1>0$ and $c_2>0$, respectively: for any $h \in \mathrm{L}^2(\Pi_Y)$ and $g \in \mathrm{L}^2(\Pi_X)$,
\begin{align*}
c_1\cdot \mathrm{Var}_{\Pi_Y}(h) \leqslant \mathcal{E}_{\Pi_Y}(H_1,h), \quad c_2\cdot \mathrm{Var}_{\Pi_X}(g)\leqslant \mathcal{E}_{\Pi_X}(H_2,g).
\end{align*}
Then the tensor product chain $H = H_1 \otimes H_2$ satisfies a standard Poincar{\'{e}} inequality:
\begin{align*}
    \min\{c_1,c_2\} \cdot \|f\|_\Pi^2 &\leqslant  \mathcal{E}_\Pi(H,f),\quad  \forall f \in \mathrm{L}^2_0(\Pi).
\end{align*}
\end{corollary}

    In the context of continuous-time Markov chains, \cite{barthe2005concentration} shows that the weak Poincar{\'{e}} inequality holds for the tensor product of $n$ independent copies of the same chain. Specifically, the corresponding function $\beta'$ for the tensor product satisfies $\beta'(s)=n\beta(s)$, where $\beta(s)$ is corresponding to the weak Poincar{\'{e}} function for a single chain. Our result can hence be seen as a generalisation of this result in the discrete-time case. 

\section{Toy Example: Normal Inverse Gamma Distribution}
\label{Sec: toy example}
We consider the problem of sampling from a Normal-inverse-gamma (NIG) distribution using the following deterministic-scan Gibbs sampler $P$:
\begin{align}
\label{NIG}
\begin{split}
    &\tau|\xi \sim \Gamma\left(1,\beta +\frac{\xi^2}{2}\right),\\
    &\xi|\tau \sim \mathcal{N}\left(0,\frac{1}{\tau}\right),
\end{split}
\end{align}
where $\beta>0$, and $(\xi,1/\tau)$ is NIG distributed.

We now consider the deterministic-scan MwG (MwG) sampler $P_{1,2}=H_1H_2$ with invariant distribution $\Pi(\tau,\xi)$.  
Specifically, let $H_1=H_{1|{\xi}}(\dif \tau'|\tau)\delta_\xi(\dif \xi')$ and $H_2=H_{2|{\tau}}(\dif \xi'|\xi)\delta_\tau(\dif \tau')$, 
where $H_{1|{\xi}}$ and $H_{2|{\tau}}$ denote the transition kernels corresponding to the Random Walk Metropolis (RWM) updates for the $\tau$- and $\xi$-components, respectively. Each RWM kernel targets the respective full conditional distributions $\Pi(\tau| \xi)$ and $\Pi(\xi| \tau)$, and is implemented with step size parameters $\sigma_{\xi}$ and $\sigma_{\tau}$. We denote the marginal distributions of $\tau$ and $\xi$ as $\Pi_T(\dif \tau)$ and $\Pi_\Xi(\dif \xi)$, respectively. 

The following lemma follows from the results of \cite{andrieu2024explicit}, which provides a way to obtain the WPIs for the RWM kernel $H_{1|{\xi}}(\dif \tau'|\tau)$ and $H_{2|{\tau}}(\dif \xi'|\xi)$.

\begin{lemma}
\label{eg: spectralgap_for_conditional}
    We write $\gamma_\xi$ for the spectral gap of $H_{1|{\xi}}(\dif \tau'|\tau)$ and $\gamma_\tau$ for the spectral gap of $H_{2|{\tau}}(\dif \xi'|\xi)$. Then, we have
    \begin{align*}
        \gamma_\xi \geqslant c \cdot \frac{(\beta_\xi\sigma_\xi)^6}{(\beta^2_\xi\sigma_\xi^2+1)^4}, \quad
        \gamma_\tau \geqslant c_0 \cdot \sigma_\tau^2\cdot \tau \cdot \exp(-2\cdot\sigma_\tau^2\cdot \tau),
    \end{align*}
    where $\beta_\xi=\beta+\frac{\xi^2}{2}$, $c=\pi^{-2} \times 2^{-11}$,  $c_0=1.972 \times 10^{-4}$.
\end{lemma}
\begin{proof}
    See \Cref{Prf:eg: spectralgap_for_conditional}.
\end{proof}

\begin{lemma}
    \label{eg:spectralgap_for_DG}
    There exists a constant $\gamma>0$, such that the following standard Poincar{\'{e}} inequality holds for the exact DG sampler $P^*P$: for any $f \in L_0^2(\Pi)$,
    \begin{align*}
        \gamma \|f\|_{\Pi}^2 \leqslant \mathcal{E}_{\Pi}(P^*P,f).
    \end{align*}
\end{lemma}
\begin{proof}
    See \Cref{Prf:eg:spectralgap_for_DG}.
\end{proof}

We consider two cases. In \Cref{eg: NIG_scale_sigma}, when the step sizes $\sigma_\xi^2$ and $\sigma_\tau^2$ are scaled appropriately, the chain $P_{1,2}$ converges exponentially. In contrast, \Cref{eg: NIG_fix_sigma} shows that fixing the step sizes at $\sigma_\xi^2=\sigma_\tau^2=\sigma_0^2$ yields only polynomial convergence. More generally, it is possible using our techniques to bound the convergence rate for arbitrary choices of the step sizes.

\begin{theorem}
\label{eg: NIG_scale_sigma}
    Under the choices $\sigma_\xi^2=3/\beta_\xi^2$ and $\sigma_\tau^2=1/2\tau$, the spectral gaps $\gamma_\xi$ and $\gamma_\tau$ are uniformly bounded over $\xi$ and $\tau$, respectively.
    Then, we have the convergence bound for MwG sampler $P_{1,2}$:  for any $f \in \mathrm{L}^2_0(\Pi)$  
    and any $n \in \mathbb{N}$, 
    \begin{align*}
    \|P_{1,2}^nf \|^2_\Pi \leqslant \|f\|^2_\Pi \cdot \exp\left(-\gamma_\tau\gamma_\xi\gamma\cdot n\right),
    \end{align*}
    where $\gamma_\xi= \frac{27}{256}\times \pi^{-2} \times 2^{-11}$ and $\gamma_\tau=\frac{1.972 \times 10^{-4}}{2e}$.
\end{theorem}
\begin{proof}
    See \Cref{Prf:eg: NIG_scale_sigma}.
\end{proof}

\begin{theorem}
\label{eg: NIG_fix_sigma}
    If we fix $\sigma_\xi^2$ and $\sigma_\tau^2$ as the same constants $\sigma_0^2$, then, we have the convergence bound for the MwG sampler $P_{1,2}$:  for any $f \in \mathrm{L}^2_0(\Pi)$ such that $\|f\|^2_\mathrm{osc} < \infty$ and any $n \in \mathbb{N}$,
    \begin{align*}
    \|P_{1,2}^nf \|^2_\Pi \leqslant \begin{cases}
        \tilde{C}_1 \cdot n^{-\frac{1}{14}} \cdot \|f\|^2_\mathrm{osc}, \quad\quad\quad\quad &\text{if}\quad  \frac{\beta}{\sigma^2_0}>1,\\
        \tilde{C}_2 \cdot n^{-\frac{\beta}{4\beta+10\sigma^2_0}}\cdot \|f\|^2_\mathrm{osc}, &\text{otherwise},
        \end{cases}
    \end{align*}
    where  $\tilde{C}_1$ and $\tilde{C}_2$ are constants which depend on $\beta$ and $\sigma^2_0$.
\end{theorem}
\begin{proof}
    See \Cref{Prf:eg: NIG_fix_sigma}.
\end{proof}

Furthermore, we consider the 2-MwG sampler $P_2=G_1H_2$, where $G_1$ denotes the operator on the joint space updating $\tau$ exactly from $\Pi(\tau| \xi)$, $H_2=H_{2|{\tau}}(\dif \xi'|\xi)\delta_\tau(\dif \tau')$, with $H_{2|{\tau}}$ denoting the transition kernel of Random Walk Metropolis (RWM) targetting $\Pi(\xi |\tau)$ with step size $\sigma_\tau^2$. The following result shows in some cases, our bound is optimal.

\begin{theorem}
\label{eg: NIG_fix_sigma_P2}
   If we fix $\sigma_\tau^2\equiv\sigma_0^2 $ to be constant, then we have the following convergence bound for $P_{2}$:  for any $f \in \mathrm{L}^2_0(\Pi)$ such that $\|f\|^2_\mathrm{osc} < \infty$ and any $n \in \mathbb{N}_+\backslash\{1\}$,
    \begin{align*}
    \|P_{2}^nf \|^2_\Pi \leqslant \begin{cases}
    \tilde{C'_1} \cdot (n-1)^{-\frac{1}{2}}\cdot \|f\|^2_\mathrm{osc},\quad\quad\quad\quad&\text{if}\quad \frac{\beta}{\sigma^2_0}>1,\\ \tilde{C'_2} \cdot (n-1)^{-\frac{\beta}{2\sigma^2_0}}\cdot \|f\|^2_\mathrm{osc},&\text{otherwise}.
    \end{cases}
    \end{align*}
    where  $\tilde{C'_1}$ and $\tilde{C'_2}$ are constants which depend on $\beta$ and $\sigma^2_0$.
\end{theorem}
\begin{proof}
    See \Cref{Prf:eg: NIG_fix_sigma_P2}.
\end{proof}
\begin{corollary}
\label{eg:sharp bound}
    When $\frac{\beta}{\sigma^2_0}>1$, we have $\|P_{2}^nf \|^2_\Pi \leqslant \tilde{C}\cdot n^{-\frac{1}{2}}\cdot \|f\|^2_\mathrm{osc}$ for any $n \in \mathbb{N}_+$ and any $f \in \mathrm{L}^2_0(\Pi)$, where the polynomial exponent is optimal in the sense that it cannot be improved to $n^{-\alpha}$ for any $\alpha>1/2$.
\end{corollary}
\begin{proof}
    See \Cref{prf:eg:sharp bound}.
\end{proof}

\section{Bayesian Hierarchical Model}
\label{Sec: bayes hierarchical model}
We consider a Bayesian Hierarchical Model: 
\begin{align*}
    \mathbf Y|\sigma, \beta \sim \mathcal N(\mathbf X \beta, \sigma^2 I),
\end{align*}
where $\mathbf Y\in \mathbb R^N$ is the response vector, $\mathbf X$ is a fixed $N \times p$ design matrix, $\beta\in \mathbb R^p$ are the regression coefficients, and $\sigma^2>0$. We suppose that $N>p$ and $\mathbf X$ has full column rank.  The priors for $\sigma^2$ and $\beta$ are given by
\begin{align*}
    \sigma^{-2} \sim \Gamma(a,b), \quad
    \beta \propto 1,
\end{align*}
where hyperparameters $a>1$ and $b>0$ are all assumed to be known.

We analyze a block Gibbs sampler $P$ that simulates a Markov chain $\{(\beta_n, \sigma^{-2}_n)\}_{n=1}^{\infty}$. The
full conditionals required for this Gibbs sampler are:
\begin{align*}
    \sigma^{-2}|\beta,\mathbf Y &\sim \Gamma\left(a+\frac{N}{2},b +\frac{(\mathbf Y-\mathbf X\beta)^\top (\mathbf Y-\mathbf X\beta)}{2}\right),\\
    \beta |\sigma^{-2},\mathbf Y &\sim N \left((\mathbf X^\top \mathbf X)^{-1}(\mathbf X^\top \mathbf Y), \sigma^{-2}(\mathbf X^\top \mathbf X)^{-1}\right).
\end{align*}
Now we consider the MwG sampler $P_2=G_1H_2$ with invariant distribution $\Pi(\beta, \sigma^{-2}|\mathbf Y)$. 
Let $H_2=H_{2|{\sigma^{-2}}}(\dif \beta'|\beta)\delta_{\sigma}(\dif \sigma')$, and $H_{2|\sigma^{-2}}$ be the operator associated with the Random-Walk Metropolis kernel for $\beta$ targeting $\Pi(\beta| \sigma^{-2})$ with step size $\sigma_0$. Furthermore, let $G_1$ be the operator corresponding to exact sampling from the conditional distribution $\Pi(\sigma^{-2}|\beta)$.

\begin{lemma}
\label{eg:bayes:spectralgap_DG}
    There exists a constant $\gamma>0$, such that a SPI holds for the DG sampler $P^*P$.
\end{lemma}
\begin{proof}
    See \Cref{Prf:eg:bayes:spectralgap_DG}.
\end{proof}

\begin{proposition}
\label{eg:rate for Bayes}
    Let $P$ and $P_2$ be as above. When $\sigma_0$ is fixed, we have the following convergence bound for the 2-MwG sampler $P_2$: for any $f \in \mathrm{L}^2_0(\Pi)$ such that $\|f\|^2_\mathrm{osc} < \infty$ and any $n \in \mathbb{N}_+\backslash\{1\}$,
    \begin{align*}
    \|P_{2}^nf \|^2_\Pi \leqslant \Tilde{C} \cdot \|f\|^2_\mathrm{osc} \cdot  (n-1)^{-\min\{a',b'/C_2\}},
    \end{align*}
    
    where $a'=a+\frac{N}{2}-\frac{p}{2}$, $b'=b+\frac{\mathbf Y^\top \mathbf Y-u^\top (\mathbf X^\top \mathbf X)u}{2}$, $u=(\mathbf X^\top \mathbf X)^{-1}\mathbf X^{\top}\mathbf Y$, $C_2=2\cdot \lambda_{\max}(\mathbf X^\top \mathbf X) \cdot p \cdot \sigma_0^2$, and $\lambda_{\max}(\mathbf X^\top \mathbf X)$ is the maximum eigenvalue of $\mathbf X^\top \mathbf X$, $\Tilde{C}$ a universal constant.
\end{proposition}
\begin{proof}
    See \Cref{Prf:eg:rate for Bayes}.
\end{proof}

\section{Data Augmentation for Diffusion Processes}
\label{Sec: diffusion}
We consider data augmentation schemes for discretely-observed diffusion processes. We consider the equations of the following form:
\begin{align}
\label{eq:diffusion}
    \dif X_t =b(X_t, \theta)\dif t + \dif B_t,  \quad t \in [0,T],
\end{align}
where $B_t$ is a standard Brownian motion on $\mathbb{R}$, $\theta \in \mathbb{R}^p$ is the parameter of interest, and the diffusion is observed at discrete time points $t=(t_0,t_1,\dots,t_N)$ with $t_0=0,t_N=T$. 

Given the entire trajectory $X_{[0,T]}=\{X_s, s \in [0,T]\}$, the likelihood of $\theta$ can be expressed using the Girsanov formula,
\begin{align*}
    \mathcal{G}\left(X_{[0, T]}, \theta\right)=\exp \left\{\int_0^T b\left(X_t, \theta\right) \mathrm{d} X_t-\frac{1}{2} \int_0^T b^2\left(X_t, \theta\right) \mathrm{d} t\right\}.
\end{align*}
 Given a prior distribution $p(\dif \theta)$,
the posterior distribution of $\theta$ has density with respect to the law of the Brownian motion given by $\pi(\dif \theta | X_{[0,T]}) \propto \mathcal{G}(X_{[0,T]},\theta)p(\theta)$.

Furthermore we write $X_{\mathrm{mis}}$ for the sample path of $X$ excluding the observations $Y =\{Y_{t_{i}},0 \leqslant i \leqslant N\}$. 
Hence, $\pi(\dif \theta | X_{\mathrm{mis}},Y) \propto \mathcal{G}(X_{[0,T]},\theta)p(\theta)$. 

Let $\mathbb{P}_\theta$ denote the law of $X$, and let $\mathbb{B}_\theta(Y_{t_{i-1}},Y_{t_{i}})$ denote the law of the resulting Brownian bridge between these observations. 
We denote the $(N+1)$-dimensional Lebesgue density of $Y$ under $\mathbb{P}_\theta$ by $g_\theta(Y_{t_0},\ldots, Y_{t_N})= g_\theta(Y_{t_0}) \prod _{i=1}^N g_\theta(Y_{t_{i-1}}, Y_{t_i})$, where $g_\theta(Y_{t_{i-1}}, Y_{t_i})$ denotes the transition density associated with \eqref{eq:diffusion} for moving from $Y_{t_{i-1}}$ to $Y_{t_i}$ over the time interval $(t_{i-1},t_i)$. 
We can write the conditional density of $X_{(t_{i-1},t_i)}$ with respect to the Brownian bridge measure as 
\begin{align*}
    \frac{\mathrm{d} \mathbb{P}_\theta}{\mathrm{d} \mathbb{B}_\theta(Y_{t_{i-1}},Y_{t_{i}})}\left(X_{(t_{i-1},t_i)}|Y\right)=\mathcal{G}\left(X_{(t_{i-1},t_i)}, \theta\right) \frac{f\left(Y_{t_{i-1}}, Y_{t_i}\right)}{g_\theta\left(Y_{t_{i-1}}, Y_{t_i}\right)}\propto \mathcal{G}\left(X_{(t_{i-1},t_i)}, \theta\right),
\end{align*}
where we have $f_\Delta(x, y)=\frac{1}{\sqrt{2 \pi}} e^{-\frac{1}{2 \Delta}(y-x)^2}$, and $\Delta$ is the time increment; in an abuse of notation we may omit the subscript $\Delta$.

In general, the update for $\theta$ follows a standard MH scheme and the update for $X_{(t_{i-1},t_i)}$ follows an Independent Metropolis--Hastings (IMH) scheme, see e.g. \cite{roberts2001inference}, where the new path $X_{(t_{i-1},t_i)}'$ is sampled according to a Brownian bridge on $[t_{i-1},t_i]$ constrained to be equal to $Y_{t_{i-1}}$ and $Y_{t_i}$ at the endpoints and it is accepted with probability 
\begin{align}
    \label{eq:acceptance for diffusion}
    \alpha_i(X,X')=\min\left(1,\frac{\mathcal{G}(X_{(t_{i-1},t_i)}',\theta)}{\mathcal{G}(X_{(t_{i-1},t_i)},\theta)}\right).
\end{align}

We update each of the $N$ pieces of missing data $X_{(t_{i-1},t_i)}$ in parallel, since they are independent given $\theta$.

This can be described as a coordinate-wise scheme on $\mathcal{X}=\mathbb{R}^p \times \prod_{i=1}^N\mathbb{R}^{(t_{i-1},t_i)}$ with invariant distribution $\pi(\dif \theta, \dif X_{\mathrm{mis}}|Y)$. The associated operator $P_{1,2}$ can be formally defined as $P_{1,2}=H_1H_2$, where $H_1=H_{1|X_{\mathrm{mis}}}(\dif \theta'|\theta)\delta_{X_{\mathrm{mis}}}(\dif X'_{\mathrm{mis}})$ , and $H_{1|X_{\mathrm{mis}}}$ is the operator associated with the standard Random-Walk Metropolis kernel for $\theta$ with the invariant distribution $\pi(\dif \theta|X_\mathrm{mis},Y)$; $H_2=H_{2|\theta}(\dif X'_{\mathrm{mis}}|X_{\mathrm{mis}})\delta_{\theta}(\dif \theta')$ and  $H_{2|\theta}(\dif X'_{\mathrm{mis}}|X_{\mathrm{mis}})= \prod_{i=1}^N H^\theta_i(\dif X_{(t_{i-1},t_i)}'|X_{(t_{i-1},t_i)})$ is the operator associated with the IMH kernel as in \eqref{eq:acceptance for diffusion} for $X_{\mathrm{mis}}$ with the invariant distribution $\pi(\dif X_\mathrm{mis}|\theta, Y)$. For notational simplicity, we omit the conditioning on $Y$. Similarly, the operator $P_2=G_1H_2$, where $G_1$ corresponds to the Markov kernel that samples $\theta$ from the conditional distribution $\pi(\dif \theta|X_\mathrm{mis})$ while keeping $X_\mathrm{mis}$ fixed, see \Cref{Alg:MwGDiffusion}.

\begin{algorithm}[H]
\caption{MwG Sampler for Diffusion}
\label{Alg:MwGDiffusion}
\begin{algorithmic}[1]
\Require  Current $(\theta, X_{\mathrm{mis}})$, observations $Y$.
\State \textbf{Update $\theta$:} Draw $\theta'$ from $H_{1|X_{\mathrm{mis}}}(\cdot |\theta)$. 
\State \textbf{Update $X_{\mathrm{mis}}$:} Draw $X_{(t_{i-1},t_i)}'$ from $H_i^\theta(\cdot |X_{(t_{i-1},t_i)})$, for $i={1,\ldots, N}$, implemented as follows.

\For {$i=1,\dots,N$}
    \State Propose the Brownian bridge $X'_{(t_{i-1},t_i)}$ from $Y_{t_{i-1}}$ to $Y_{t_i}$.
    \State Accept it with probability $\alpha_i$ as in \eqref{eq:acceptance for diffusion}.
\EndFor

\State Set $(\theta, X_{\mathrm{mis}})\leftarrow (\theta', X'_{\mathrm{mis}})$.
\end{algorithmic}
\end{algorithm}
We shall impose the following assumptions, inspired by \cite[Section 8]{ascolani2024scalability}, which are satisfied in many cases \cite{beskos2006retrospective}. Our following results extend the results of \cite{ascolani2024scalability} by removing their restrictive compact support condition on the prior.
\begin{assumption}
    \label{asp:diffusion}
    The function $b(x,\theta)$ is differentiable in $x$ for all $\theta$ and continuous in $\theta$ for all $x$. 
    There exists a function $M(\theta)$ such that $b^2(x,\theta) + b'(x,\theta)\geqslant M(\theta)$ for all $x$.
\end{assumption}

\begin{proposition}
    \label{thm:betaforIMHindiffusion}
    Let $P_{1,2}$ be as above, and $P$ be the standard Gibbs kernel. Under \Cref{asp:diffusion}, the operator $H_{2|\theta}$ satisfies the following WPI:
    \begin{align*}
        \|g\|_{\pi(\cdot|\theta)}^2 \leqslant s \cdot \mathcal{E}_{\pi(\cdot|\theta)}(H_{2|\theta},g) + \beta_2(s,\theta) \cdot \|g\|_\mathrm{osc}^2, \quad \forall g \in \mathrm{L}^2_0(\pi(\cdot|\theta)),
    \end{align*}
    where $\beta_2(s,\theta)=\mathbbm{1}_{\{s\leqslant \max_i\{\tilde{\mathcal{G}}(X_{(t_{i-1},t_i)},\theta)\}\}}$,   $\tilde{\mathcal{G}}(X_{(t_{i-1},t_i)},\theta)=\exp\{-\frac{1}{2}M(\theta)\Delta t_i+(A(Y_{t_{i}})-A(Y_{t_{i-1}}))\}$ and $A(u)=\int_0^u b(x,\theta)dx$.
\end{proposition}
\begin{proof}
    See \Cref{prf:thm:betaforIMHindiffusion}.
\end{proof}

The following two corollaries are established directly by combining \Cref{thm:betaforIMHindiffusion}
    with \Cref{thm:weakPtotildeP} and \Cref{thm:weakPtotildeP2}, respectively.

\begin{corollary}
\label{eg:diffusionWPI}
    In addition to \Cref{asp:diffusion},  assume that $P^*P$ satisfies \Cref{Asp:P}, and $H_{1|X_{\mathrm{mis}}}(\dif \theta'|\theta)$ satisfies \Cref{Asp:P1} with $\beta_1(s,X_{\mathrm{mis}})$.
    Then, there exists a function $K^*$ such that $P_{1,2}^*P_{1,2}$ satisfies a $K^*$-WPI: for all $f \in \mathrm{L}^2_0(\pi)$,
    \begin{align*}
        \frac{\mathcal{E}_\pi(P_{1,2}^*P_{1,2},f)}{\|f\|_\mathrm{osc}^2} \geqslant K^*\left(\frac{\|f\|_\pi^2}{\|f\|_\mathrm{osc}^2}\right),
    \end{align*}
    where 
     \begin{align}
        \begin{split}
            &K^*(v)=2K_1^*\circ K_2^* \circ \frac{1}{2}K_0^*\left( \frac{1}{4} \cdot v\right)=2K_1^*\left(K_2^*\left(\frac{1}{2}K_0^*\left( \frac{1}{4} \cdot v\right)\right)\right),\\
            &K_i^*(v)=\sup_{u\geqslant 0}\{uv - u \cdot \beta_i(1 / u)\}, \quad i=0,1,2.\\
            &\beta_1(s)=\int_{X_{\mathrm{mis}}} \beta_1(s,X_{\mathrm{mis}})\Pi_{X_{\mathrm{mis}}}(\dif X_{\mathrm{mis}}),\\
            &\beta_2(s)=\int_\theta \beta_2(s,\theta)\Pi_\theta(\dif \theta), \quad \beta_2(s,\theta)= \mathbbm{1}_{\{s\leqslant \max_i{\tilde{\mathcal{G}}(X_{(t_{i-1},t_i)},\theta)}\}},\\
            & \tilde{\mathcal{G}}(X_{(t_{i-1},t_i)},\theta)=\exp \left\{A\left(Y_{t_i}\right)-A(Y_{t_{i-1}})-\frac{1}{2} \Delta t_i M(\theta)\right\}.
        \end{split}
        \end{align}
\end{corollary}

\begin{corollary}
\label{eg:diffusion_rate}
    In addition to \Cref{asp:diffusion},  assume that $P^*P$ satisfies \Cref{Asp:P_SPI}, then for any $f \in \mathrm{L}^2_0(\pi)$ such that $\|f\|^2_\mathrm{osc} < \infty$ and any $n \in \mathbb{N}_+\backslash\{1\}$,
    \begin{align*}
        \|P_2^nf\|_\pi^2 \leqslant \|f\|^2_{\mathrm{osc}}\cdot F^{-1}(n-1),
    \end{align*}
    where $F:=\int_x^{\frac{1}{4}} \frac{\dif v}{K^*(v)}$ is as in \Cref{Thm:rate},
    \begin{align*}
        &K^*(v) =K_2^*\left(\frac{\gamma}{2}v\right), \quad K_2^*(v)=\sup_{u\geqslant 0}\{uv - u \cdot \beta_2(1 / u)\},\\
        &\beta_2(s)=\int_\theta \beta_2(s,\theta)\Pi_\theta(\dif \theta),\quad \beta_2(s,\theta)= \mathbbm{1}_{\{s\leqslant \max_i{\tilde{\mathcal{G}}(X_{(t_{i-1},t_i)},\theta)}\}},\\
        &\tilde{\mathcal{G}}(X_{(t_{i-1},t_i)},\theta)=\exp \left\{A\left(Y_{t_i}\right)-A(Y_{t_{i-1}})-\frac{1}{2} \Delta t_i M(\theta)\right\}.
    \end{align*}
\end{corollary}

For diffusions, it is generally difficult to establish SPIs for the standard Gibbs kernel. However, for a broad class of diffusion processes, we will show in \Cref{thm:SPI_for_diff} that a SPI can be obtained under explicit and verifiable boundedness conditions on the drift function $b(x,\theta)$. The proof is in \Cref{Prf:thm:SPI_for_diff}. By combining \Cref{Asp:SPI_for_diff} and \Cref{asp:diffusion}, the convergence bound for $P_2$ can be derived directly by \Cref{eg:diffusion_rate}. 
\begin{assumption}
\label{Asp:SPI_for_diff}
    The function $b(x,\theta)$ is differentiable in $x$ for all $\theta$ and continuous in $\theta$ for all $x$. Meanwhile, the function $ b^2(x,\theta)+b'(x,\theta)$ is uniformly bounded below by constant $M_L$ and uniformly bounded above by constant $M_U$ for some $a>0$, all $|x|\leqslant a$.
\end{assumption}

\begin{proposition}
\label{thm:SPI_for_diff}
      Let $P$ be the standard Gibbs kernel. $P^*P$ satisfies the SPI under \Cref{Asp:SPI_for_diff}.
\end{proposition}
\begin{remark}
In a multivariate setting $X_t\in\mathbb{R}^p$, consider drift $b(x,\theta)=\nabla_x B(x,\theta)$ for some $B(x,\theta)\in C^2(\mathbb{R}^{p+1})$.
Then, replacing $b^2(x,\theta)+b'(x,\theta)$ by $\|b(x,\theta)\|^2+\Delta_x B(x,\theta)$ in \Cref{Asp:SPI_for_diff}, \Cref{thm:SPI_for_diff} remains valid.
\end{remark}

\subsection{Ornstein--Uhlenbeck Process Example}
We consider applying this data augmentation scheme to an Ornstein--Uhlenbeck process:
    \begin{align}
    \label{eq:ou}
        \dif X_t = -\theta X_t\dif t + \dif B_t,  \quad t \in [0,T],
    \end{align}
with discrete observations $Y$ as in Section~\ref{Sec: diffusion}, and where the prior $p(\theta)$ is $\mathcal{N}(\mu_0,\tau_0^2)$. 
The full conditionals required for this Gibbs sampler are:
\begin{align*}
    &\theta | X_{\mathrm{mis}}, Y \sim \mathcal{N}\left(\frac{-\int_0^TX_t \dif X_t+\mu_0\tau_0^{-2}}{\int_0^T X_t^2 \dif t +\tau_0^{-2}},\frac{1}{\int_0^T X_t^2 \dif t +\tau_0^{-2}}\right),\\
    &\dif \mathbb{P}(X_{\mathrm{mis}}|\theta, Y) = \prod_{i=1}^N \mathcal{G}(X_{(t_{i-1},t_i)}, \theta)\frac{f(Y_{t_{i-1}},Y_{t_i})}{g_\theta(Y_{t_{i-1}},Y_{t_i})} \dif  \otimes_{i=1}^N \mathbb{B}(Y_{t_{i-1}},Y_{t_i}).
\end{align*}
We denote the operator of the standard Gibbs sampler as $P=G_1G_2$, where $G_1$ and $G_2$ denote the operators on the joint space corresponding to updating the corresponding components from the conditional distributions $\pi(\theta|X_{\mathrm{mis}},Y)$ and $\pi(X_{\mathrm{mis}}|\theta,Y)$, respectively.
Now we consider the MwG sampler $P_{2}=G_1H_2$ with invariant distribution $\pi(\theta, X_{\mathrm{mis}})$. Let $H_2=H_{2|\theta}(\dif X'_{\mathrm{mis}}|X_{\mathrm{mis}})\delta_{\theta}(\dif \theta')= \prod_{i=1}^N H^\theta_i(\dif X_{(t_{i-1},t_i)}'|X_{(t_{i-1},t_i)})\delta_{\theta}(\dif \theta')$, where $H^\theta_i(\dif X_{(t_{i-1},t_i)}'|X_{(t_{i-1},t_i)})$ is the IMH updating scheme, which is reversible and positive.

Since the drift $-\theta x$ does not satisfy the conditions in \Cref{Asp:SPI_for_diff}, we instead provide \Cref{eg: ou_SPI}. The proofs of the following two results are in \Cref{Prf:eg: ou_SPI}.
\begin{lemma}
\label{eg: ou_SPI}
     $P^*P$ satisfies the SPI with the constant $\gamma>0$.
\end{lemma}

\begin{corollary}
\label{eg: rate for ou}
    Let $P$ and $P_2$ be as above, then we have the convergence bound for the MwG $P_2$: for any $f \in \mathrm{L}^2_0(\pi)$ such that $\|f\|^2_\mathrm{osc} < \infty$ and any $n \in \mathbb{N}_+\backslash\{1\}$,
    \begin{align*}
    \|P_{2}^nf \|^2_\pi \leqslant \Tilde{C} \cdot \|f\|^2_\mathrm{osc} \cdot \exp\left(-\frac{a}{\delta} \log^2\left(\frac{n-1}{\gamma/2}\right)\right),
    \end{align*}
    where $\Tilde{C}$ is a universal constant, $a=\frac{2}{\eta^2 \tau_0^2}$, $\eta=\max_{i}\{\Delta t_{i}-Y_{t_i}^{2}+Y_{t_{i-1}}^{2}\}$, $\delta$ is any constant strictly greater than $1$.
\end{corollary}

\appendix
\section{Notation}
We summarize the notations and definitions that appear in the auxiliary results related to measures and operators.
\begin{itemize}
    \item We assume $\mathrm{L}^2(\mu)$ denote the Hilbert space of (equivalence classes of) real-valued measurable functions that are square integrable with respect to $\mu$, equipped with the inner product $\langle f, g\rangle_\mu=\int_{\mathsf{E}} f(x) g(x) \mathrm{d} \mu(x)$, for $f, g \in \mathrm{L}^2(\mu)$. Then let $\|f\|^2_\mu = \langle f, f\rangle_{\mu}$, then write $\mathrm{L}^2_0(\mu)$ as the set of functions $f \in \mathrm{L}^2(\mu)$ satisfying $\mu(f) = 0$.
    \item For a set $A \in \mathcal{F}$, we denote the corresponding indicator function by $\mathbbm{1}_A: \mathsf{E} \mapsto \{0,1\}$. We denote the Dirac measure at $x$ by $\delta_x(A)$, defined as $\delta_x(A)=\mathbbm{1}_A(x)$, for any $A \in \mathcal{F}$.
    \item Let $T$ be a Markov transition kernel on $\mathsf{E}$, define a linear operator $T: \mathrm{L}^2(\mu)\rightarrow \mathrm{L}^2(\mu)$ by $(Tf)(x)=\int_\mathsf{E} T(x,\dif  x')f(x')$, for any $x \in \mathsf{E}$.
    \item For such an operator $T$ , we write $T^*$ for its adjoint operator $T^* : \mathrm{L}^2(\mu)\rightarrow \mathrm{L}^2(\mu)$, which satisfies $\langle f, T g\rangle_\mu = \langle T^*f, g\rangle_\mu$ for any $f,g \in \mathrm{L}^2(\mu)$.
    \item We will use the notation $\mathrm{K}^*$ to denote the optimised weak Poincar{\'{e}} inequality function appearing in the general theoretical results, and $K^*$ to denote our specific instance corresponding to particular operators under consideration.
    \end{itemize}

\section{Technical Proofs}
\subsection{Proofs for \Cref{Sec: Comparison for Gibbs,Sec: tensor product}}
\begin{proof}[Proof of \Cref{lem:P^2toP}]
\label{Prf:lem:P^2toP}
Observe that $0 \leqslant\|f-Tf\|_\mu^2=\|f\|_\mu^2-\langle Tf, f\rangle_\mu-\langle f, Tf\rangle_\mu+\|Tf\|_\mu^2$. For any function $f \in \mathrm{L}^2(\mu)$, this gives $\|f\|_\mu^2-\langle Tf, f\rangle_\mu \geqslant \langle f, Tf\rangle_\mu-\|Tf\|_\mu^2$. Since $f$ is real-valued function, $\langle f, Tf\rangle_\mu=\langle Tf, f\rangle_\mu$.

Thus, $ \mathcal{E}_\mu(T^*T, f)=\|f\|_\mu^2-\langle Tf, f\rangle_\mu+\langle Tf, f\rangle_\mu-\|Tf\|_\mu^2 \leqslant 2\mathcal{E}_\mu(T, f)$.
\end{proof}

\begin{proof}[Proof of \Cref{lem:P_1reversible}]
    \label{Prf:lem:P_1reversible} 
    First, it is easy to check the detailed
    balance condition for $G_{1}$ to show its reversibility. Then, for any
    $f \in \mathrm{L}^{2}(\Pi)$, we have
    \begin{align*}
        \langle G_{1}f,f\rangle_{\Pi} & =\int f(x,y)f(x',y')\Pi_{Y|X}(\dif y'|x)\delta_{x}(dx')\Pi(\dif x,\dif y)                         \\
                                      & =\int_\mathcal{X}\left[\int_\mathcal{Y}f(x,y)f(x,y')\Pi_{Y|X}(\dif y'|x)\Pi_{Y|X}(\dif y|x)\right]\Pi_{X}(\dif x) \\
                                      & =\int_\mathcal{X}\left[\int_\mathcal{Y}f(x,y)\Pi_{Y|X}(\dif y|x)\right]^{2}\Pi_{X}(\dif x)\geqslant 0,
    \end{align*}
    which show $G_{1}$ is positive.

    $H_{1|x}$ is $\Pi_{x}$-reversible and positive under \Cref{Asp:P1}, and hence it is straightforward
    to show $H_{1}$ is $\Pi$-reversible. Note that
    \begin{align*}
        \langle H_{1}f,f\rangle_{\Pi} & =\int f(x,y)f(x',y')H_{1|x'}(\dif y'|y)\delta_{x}(dx')\Pi(\dif x,\dif y)                         \\
                                      & =\int_\mathcal{X}\left[\int_\mathcal{Y}f(x,y)f(x,y')H_{1|x}(\dif y'|y)\Pi_{Y|X}(\dif y|x)\right]\Pi_{X}(\dif x) \\
                                      & =\int_\mathcal{X}\left[\langle H_{1|x}f_{x},f_{x}\rangle_{\Pi(\cdot|x)}\right]\Pi_{X}(\dif x)\geqslant 0,
    \end{align*}
    where $f_{x}=f(x,\cdot)$, showing $H_{1}$ is positive as well.
    
    Hence, it is easy to check $G_2$ and $H_2$ are positive.
\end{proof}

\begin{proof}[Proof of \Cref{lem:P1_to_tildeP1}]
\label{prf:lem:P1_to_tildeP1}
    For any $f \in \mathrm{L}^2(\Pi)$, the Dirichlet form of $H_1$ satisfies
    \begin{align*}
        \mathcal{E}_{\Pi}(H_1,f)&=\frac{1}{2}\int (f(x,y)-f(x',y'))^2 H_{1|x}(\dif y'|y)\delta_x(\dif x')\Pi(\dif x,\dif y)\\
        &=\frac{1}{2}\int (f_x(y)-f_x(y'))^2 H_{1|x}(\dif y'|y)\Pi_{Y|X}(\dif y|x)\Pi_X(\dif x)\\
        &=\int_\mathcal{X} \mathcal{E}_{\Pi(\cdot|x)}(H_{1|x},f_x)\Pi_X(\dif x).
    \end{align*}
    Similarly,
    For any $f \in \mathrm{L}^2(\Pi)$, we have 
    \begin{align*}
        \mathcal{E}_{\Pi}(G_1,f)&=\frac{1}{2}\int (f(x,y)-f(x',y'))^2 \Pi_{Y|X}(\dif y'|x)\delta_x(\dif x')\Pi(\dif x,\dif y)\\
        &=\frac{1}{2}\int (f_x(y)-f_x(y'))^2 \Pi_{Y|X}(\dif y'|x)\Pi_{Y|X}(\dif y|x)\Pi_X(\dif x)\\
        &=\int_\mathcal{X} \mathrm{Var}_{\Pi(\cdot|x)}(f_x)\Pi_X(\dif x).
    \end{align*}
    Hence, for any $s > 0$, we have 
    \begin{align*}
        \begin{split}
            \mathcal{E}_{\Pi}(G_1,f)&=\int_\mathcal{X} \mathrm{Var}_{\Pi(\cdot|x)}(f_x)\Pi_X(\dif x)\\
            & \leqslant \int_\mathcal{X} \left[ s \cdot \mathcal{E}_{\Pi(\cdot|x)}(H_{1|x},f_x)+\beta_1(s,x) \cdot \|f_x\|^2_\mathrm{osc} \right] \Pi_X(\dif x)\\
            &=s \cdot \mathcal{E}_\Pi(H_1,f)+\|f\|^2_{\mathrm{osc}} \cdot \int_\mathcal{X}\beta_1(s,x)\Pi_X(\dif x).
        \end{split}
    \end{align*}
    Now, we check that the function $\beta_1$ satisfies the regularity conditions of the WPI definition.
    Since $s \mapsto \beta_1(s, x)$ is monotone decreasing for each $x\in \mathcal{X}$, it follows that $\beta_1$ is also monotone decreasing. Furthermore, since $\beta_1(\cdot, x) \leqslant 1/4$ pointwise for all $x \in \mathcal{X}$, and $\int_\mathcal{X} \Pi_X(\dif x) = 1$, and each $\beta_1(s, x) \rightarrow 0$ as $s \rightarrow \infty$, the dominated convergence theorem yields that
    \begin{align*}
        \beta_1(s) =  \int_\mathcal{X}\beta_1(s,x)\Pi_X(\dif x) \rightarrow 0, \quad \text{as} \quad s \rightarrow \infty.
    \end{align*}
    So the mapping $\beta_1: (0, \infty) \rightarrow [0, \infty)$ is decreasing, and satisfies $\lim_{s\rightarrow 0} \beta_1(s) = 0$.
    We can similarly obtain the comparison WPI between $G_2$ and $H_2$ with $\beta_2(s) =  \int_\mathcal{Y}\beta_2(s,y)\Pi_Y(\dif y)$.
\end{proof}
\begin{proof}[Proof of \Cref{thm:K*:P1tildeP2totildeP1tildeP2}]
\label{prf:thm:K*:P1tildeP2totildeP1tildeP2}
    The proof follows the same argument as that of \Cref{thm:K*:PtoP1}, with $G_2$ replaced by $H_2$, which does not affect the comparison between $G_1$ and $H_1$.
\end{proof}
\begin{proof}[Proof of \Cref{thm:K*:P*toP2*}]
\label{prf:thm:K*:P*toP2*}
    The proof is almost identical to that of \Cref{thm:K*:PtoP1}, except that the roles of the two components are interchanged, that is, the first component is now taken as
    $G_{2}$ and $H_2$, while the second becomes $G_{1}$ and $H_1$.
\end{proof}
\begin{proof}[Proof of \Cref{Thm:PtoP*}]
    \label{Prf:Thm:PtoP*} By \Cref{Thm:K*:PtoP*}, we have
    \begin{align*}
        \mathcal{E}_{\Pi}(T^{*}T,f)\geqslant \mathcal{E}_{\Pi}(TT^{*},Tf).
    \end{align*}
    We apply the WPI to $Tf$:
    \begin{align*}
        \begin{split}\|Tf\|_{\Pi}^{2}&\leqslant s \cdot \mathcal{E}_{\Pi}\left(TT^{*}, Tf\right)+\beta(s) \cdot \|Tf\|^{2}_{\mathrm{osc}}\\&\leqslant s \cdot \mathcal{E}_{\Pi}\left(T^{*}T, f\right)+\beta(s) \cdot \|f\|^{2}_{\mathrm{osc}}.\end{split}
    \end{align*}
    The inequality holds because
    $\|Tf\|^{2}_{\mathrm{osc}}\leqslant \|f\|^{2}_{\mathrm{osc}}$. Then,
    \begin{align*}
        \begin{split}\|f\|_{\Pi}^{2}&\leqslant \|f\|_{\Pi}^{2}-\|Tf\|_{\Pi}^{2}+ s \cdot \mathcal{E}_{\Pi}\left(T^{*}T, f\right)+\beta(s) \cdot \|f\|^{2}_{\mathrm{osc}}\\&= (s+1) \cdot \mathcal{E}_{\Pi}\left(T^{*}T, f\right)+\beta(s) \cdot \|f\|^{2}_{\mathrm{osc}}\\&= s' \cdot \mathcal{E}_{\Pi}\left(T^{*}T, f\right)+\beta(s'-1) \cdot \|f\|^{2}_{\mathrm{osc}},\end{split}
    \end{align*}
    where $s'=s+1$. Hence, for any $f \in \mathrm{L}^{2}_{0}(\Pi)$,
    \begin{align*}
        \|f\|_{\Pi}^{2} & \leqslant s \cdot \mathcal{E}_{\Pi}\left(T^{*}T, f\right)+\Tilde{\beta}(s) \cdot \|f\|^{2}_{\mathrm{osc}},
    \end{align*}
    where $\Tilde{\beta}(s)=\beta(s-1)$ when $s>1$. On the other hand, $\Tilde{\beta}
    (s)\leqslant \frac{1}{4}$ because $\|f\|_{\Pi}^{2}\leqslant \|f\|^{2}_{\mathrm{osc}}
    /4$. Hence we take $\Tilde{\beta}(s)=
    \begin{cases}
        \beta(s-1),  & s>1            \\
        \frac{1}{4}, & 0<s\leqslant 1
    \end{cases}$, then $\Tilde{\beta}(s)$ is a non-increasing function from $(0,\infty
    )$ to $[0,\infty)$ and decreasing when $s>1$ with $\lim_{s \rightarrow
    \infty}\Tilde{\beta}(s) = 0$.
\end{proof}

\begin{proof}[Proof of \Cref{thm:strongPtoP^*}]
    \label{Prf:thm:strongPtoP^*} It is easy to observe that, for any
    $f \in \mathrm L_{0}^{2}(\Pi)$,
       $ \frac{\|Pf\|_{\Pi}^{2}}{\|f\|_{\Pi}^{2}}\leqslant 1-\gamma$,
    that is equivalent to $\left(\sup_{f \neq 0}\frac{\|Pf\|_{\Pi}}{\|f\|_{\Pi}}\right)^{2}\leqslant 1-\gamma$.
    Since $\|P\|_{\Pi}=\|P^{*}\|_{\Pi}$, we have $\left(\sup_{f \neq 0}\frac{\|P^{*}f\|_{\Pi}}{\|f\|_{\Pi}}\right)^{2}\leqslant 1-\gamma$.
    Hence for any $f \in \mathrm L_{0}^{2}(\Pi)$, we have $\gamma \cdot\|f\|_{\Pi}^{2}\leqslant \mathcal{E}_{\Pi}(PP^{*}, f)$.
\end{proof}
\begin{proof}[Proof of \Cref{thm:strongPtotildeP}]
    \label{Prf:thm:strongPtotildeP}

    The following $\mathrm K^*$-WPIs can be obtained directly by \Cref{Asp:P_SPI}, \Cref{lem:P1_to_tildeP1},
    \begin{align*}
        &\frac{\mathcal{E}_{\Pi}\left(P^{*}P, f\right)}{\|f\|^{2}_{\text{osc}}}  \geqslant \frac{\gamma \cdot \|f\|_{\Pi}^{2}}{\|f\|^{2}_{\text{osc}}},                                 \\
        &\frac{\mathcal{E}_{\Pi}(H_{i},f)}{\|f\|^{2}_{\text{osc}}}               \geqslant K_{i}^{*}\left(\frac{\mathcal{E}_{\Pi}\left(G_{i}, f\right)}{\|f\|^{2}_{\text{osc}}}\right), \quad i=1,2,
    \end{align*}
    where $K_{1}^{*}$ and $K_{2}^{*}$ correspond to $\beta_{1}(s)=\int_\mathcal{X}\beta_{1}(s,x)\Pi_{X}(\text{d}x)$, $\beta_{2}(s
    )=\int_\mathcal{Y}\beta_{2}(s,y)\Pi_{Y}(\text{d}y)$, respectively.

    Then, we can obtain the following $\mathrm K^*$-WPIs by \Cref{thm:strongPtoP^*} and
    \Cref{thm:K*:P*toP2*}:
    \begin{align*}
        \frac{\mathcal{E}_{\Pi}\left(PP^{*}, f\right)}{\|f\|^{2}_{\text{osc}}}         & \geqslant \frac{\gamma \cdot \|f\|_{\Pi}^{2}}{\|f\|^{2}_{\text{osc}}}=K_{0}^{*}\left(\frac{\|f\|_{\Pi}^{2}}{\|f\|^{2}_{\text{osc}}}\right), \\
        \frac{\mathcal{E}_{\Pi}\left(P_{2}P_{2}^{*}, f\right)}{\|f\|^{2}_{\text{osc}}} & \geqslant 2 \cdot K_{2}^{*}\left(\frac{1}{2}\cdot\frac{\mathcal{E}_{\Pi}\left(PP^{*}, f\right)}{\|f\|^{2}_{\text{osc}}}\right),
    \end{align*}
    where $K_{0}^{*}(v)=\gamma \cdot v$. Then, by chaining rule \Cref{Thm:chaining}, we can obtain the following $\mathrm K^*$-
    WPI:
    \begin{align*}
        \frac{\mathcal{E}_{\Pi}\left(P_{2}P_{2}^{*}, f\right)}{\|f\|^{2}_{\text{osc}}} & \geqslant 2\cdot K_{2}^{*}\left(\frac{1}{2}\cdot K^{*}_{0}\left(\frac{\|f\|_{\Pi}^{2}}{\|f\|^{2}_{\text{osc}}}\right)\right)=2\cdot K_{2}^{*}\left(\frac{1}{2}\gamma \cdot \frac{\|f\|_{\Pi}^{2}}{\|f\|^{2}_{\text{osc}}}\right).
    \end{align*}
    Next, we can derive the following $\mathrm K^*$-WPI by \Cref{Thm:K*:PtoP*}:
    \begin{align}
        \label{eq:SimplifyConstant}\frac{\mathcal{E}_{\Pi}\left(P_{2}^{*}P_{2}, f\right)}{\|f\|^{2}_{\text{osc}}} & \geqslant 2\cdot K_{2}^{*}\left(\frac{1}{2}\cdot K^{*}_{0}\left(\frac{1}{2}\cdot \frac{\|f\|_{\Pi}^{2}}{\|f\|^{2}_{\text{osc}}}\right)\right)=2\cdot K_{2}^{*}\left(\frac{1}{4}\gamma \cdot \frac{\|f\|_{\Pi}^{2}}{\|f\|^{2}_{\text{osc}}}\right).
    \end{align}
    Finally, by \Cref{thm:K*:P1tildeP2totildeP1tildeP2} and \Cref{Thm:chaining},
    we can obtain the following $\mathrm K^*$-WPIs:
    \begin{align*}
        \frac{\mathcal{E}_{\Pi}\left(P_{1,2}^{*}P_{1,2}, f\right)}{\|f\|^{2}_{\text{osc}}} & \geqslant 2\cdot K_{1}^{*}\left(\frac{1}{2}\cdot \frac{\mathcal{E}_{\Pi}\left(P_{2}^{*}P_{2}, f\right)}{\|f\|^{2}_{\text{osc}}}\right)\geqslant 2\cdot K_{1}^{*}\left(K_{2}^{*}\left(\frac{\gamma}{4}\cdot \frac{\|f\|_{\Pi}^{2}}{\|f\|^{2}_{\text{osc}}}\right)\right).
    \end{align*}
\end{proof}
\begin{proof}[Proof of \Cref{thm:ApproxPtoP_X}]
    \label{Prf:thm:ApproxPtoP_X} Let $g \in \mathrm{L}^{2}_{0}(\Pi_{X})$ and $f_{g}
    \in \mathrm{L}^{2}_{0}(\Pi)$ be such that $f_{g}(x, y) = g(x)$. It is
    obvious that $\|f_{g}\|^{2}_{\mathrm{osc}}=\|g\|^{2}_{\mathrm{osc}}$. By the
    proof of \Cref{thm:PtoP_X}, we have $\|f\|^{2}_{\Pi}=\|g\|^{2}_{\Pi_X}$.

    Then we show
    $\mathcal{E}_{\Pi}(P_{2}^{*}P_{2},f_{g})=\mathcal{E}_{\Pi_{X}}(\bar{P}_{X}^{*}
    \bar{P}_{X},g)$. Since $G_{1}$ is projection, we have $P_{2}^{*}=H_{2}G_{1}$
    and $P_{2}^{*}P_{2}=H_{2}G_{1}H_{2}$. It is easy to show $\mathcal{E}_{\Pi}(P_{2}^{*}P_{2},f_{g})=\mathcal{E}_{\Pi_{X}}(\bar{P}_{X}^{*}\bar{P}_{X},g)$.
    This implies for any $g \in \mathrm{L}^{2}_{0}(\Pi_{X})$, the WPI $\|g\|_{\Pi_X}^{2} \leqslant s \cdot \mathcal{E}_{\Pi_{X}}(\bar{P}_{X}^{*}\bar{P}_{X},g) +\beta(s) \cdot \|g\|^{2}_{\mathrm{osc}}$
    holds with the same $\beta$.
\end{proof}

\begin{proof}[Proof of \Cref{thm:ApproxP_XtoP}]
    \label{Prf:thm:ApproxP_XtoP} The proof is nearly the same as the proof of \Cref{thm:P_XtoP}. Let $f\in \mathrm{L}^{2}_{0}(\Pi)$ and $g_{f}(x)=(P_{2}f)(x,y)$, then it is easy to verify that $g_{f}\in \mathrm{L}^{2}_{0}(\Pi_{X})$, and $\|g_{f}\|_{\Pi_X}
    ^{2}=\|P_{2}f\|_{\Pi}^{2}$. So the augment in \Cref{thm:P_XtoP} still works with replacing $P$ with $P_2$ and $P_X$ with $\bar{P}_X$.
\end{proof}
\begin{proof}[Proof of \Cref{lem:P_xreversible}]
    \label{Prf:lem:P_xreversible} First, it is easy to check the detailed
    balance condition for $P_{X}$ to show its reversibility. Then, for any
    $g \in \mathrm L^{2}(\Pi_{X})$, we have
    \begin{align*}
        \langle P_{X}g,g\rangle_{\Pi_X} & =\int g(x)g(x')\Pi_{X|Y}(\text{d}x'|y)\Pi_{Y|X}(\text{d}y|x)\Pi_{X}(\text{d}x)        \\
                                        & =\int_\mathcal{Y}\left[\int_\mathcal{X}g(x)\Pi_{X|Y}(\text{d}x|y)\right]^{2}\Pi_{Y}(\text{d}y)\geqslant 0,
    \end{align*}
    which show $P_{X}$ is positive.
    Since $H_{2|y}$ is $\Pi_{y}$-reversible, positive under \Cref{Asp:P2}, it is straightforward
    to show $\bar{P}_{X}$ is $\Pi_{X}$-reversible. Note that $\bar{P}_{X}$ is positive as well since
    \begin{align*}
        \langle \bar{P}_{X}g,g\rangle_{\Pi_X} & =\int g(x)g(x')H_{2|y}(\text{d}x'|x)\Pi_{Y|X}(\text{d}y|x)\Pi_{X}(\text{d}x)                         \\
                                              & =\int_\mathcal{Y}\left[\int_\mathcal{X}g(x)g(x')H_{2|y}(\text{d}x'|x)\Pi_{X|Y}(\text{d}x|y)\right]\Pi_{Y}(\text{d}y) \\
                                              & =\int_\mathcal{Y}\left[\langle H_{2|y}g,g\rangle_{\Pi(\cdot|y)}\right]\Pi_{Y}(\text{d}y)\geqslant 0.
    \end{align*}
\end{proof}
\begin{proof}[Proof of \Cref{thm:XDGtoXHDG}]
\label{Prf:thm:XDGtoXHDG}
    Since $\int_\mathcal{X} g^2(x) \Pi_X(\dif x)=\int_\mathcal{Y}\int_\mathcal{X} g^2(x) \Pi(\dif x|y)\Pi_Y(\dif y)<\infty$, for $g \in \mathrm{L}^2(\Pi_X)$, observe that $\int_\mathcal{X} g^2(x) \Pi(\dif x|y)<\infty$ and $\mathcal{E}_{\Pi(\cdot|y)}(H_{2|y},g)<\infty$ for almost every $y \in \mathcal{Y}$. 
    By \eqref{equ:DirichletForm},
    \begin{align*}
        \mathcal{E}_{\Pi_X}(\bar{P}_X,g)&=\frac{1}{2}\int_\mathcal{X}\int_\mathcal{X}(g(x)-g(x'))^2\bar{P}_X(x,\dif x')\Pi_X(\dif x)\\
        &=\frac{1}{2}\int_\mathcal{X}\int_\mathcal{X} (g(x)-g(x'))^2\int_\mathcal{Y}H_{2|y}(\dif x'|x)\Pi(\dif y|x)\Pi_X(\dif x).
    \end{align*}
    Changing the order of integration yields
    \begin{align*}
        &\frac{1}{2}\int_\mathcal{X}\int_\mathcal{X}(g(x)-g(x'))^2 \int_\mathcal{Y} H_{2|y}(\dif x'|x)\Pi(\dif y|x)\Pi_X(\dif x)\\
        = &\frac{1}{2}\int_\mathcal{Y}\int_\mathcal{X}\int_\mathcal{X}(g(x)-g(x'))^2  H_{2|y}(\dif x'|x)\Pi(\dif x|y)\Pi_Y(\dif y)\\
        =&\int_\mathcal{Y} \mathcal{E}_{\Pi(\cdot|y)}(H_{2|y},g)\Pi_Y(\dif y).
    \end{align*}
    Therefore, for any $g \in \mathrm{L}^2(\Pi_X)$, the Dirichlet form of $\bar{P}_X$ satisfies
    \begin{align*}
        \mathcal{E}_{\Pi_X}(\bar{P}_X,g)=\int_\mathcal{Y} \mathcal{E}_{\Pi(\cdot|y)}(H_{2|y},g)\Pi_Y(\dif y),
    \end{align*}
    where for each $y \in \mathcal Y$, $\mathcal{E}_{\Pi(\cdot|y)}(H_{2|y},g)=\frac{1}{2}\int_\mathcal{X}\int_\mathcal{X} (g(x)-g(x'))^2 H_{2|y}(x,\dif x')\Pi(\dif x|y)<\infty$
    is the Dirichlet form of $H_{2|y}$.
    
    As remarked earlier, $X$-marginal DG sampling with transition kernel $P_X$ is a special case of the $X$-marginal of 2-MwG sampler wherein $H_{2|y}(\dif x'|x) =\Pi_{X|Y}(\dif x'|y)$. This immediately yields the following representation of $\mathcal{E}_{\Pi_X}(P_X,g)$.
    Similarly,
    for any $g \in \mathrm{L}^2(\Pi_X)$, we have 
    \begin{align*}
        \mathcal{E}_{\Pi_X}(P_X,g)&=\frac{1}{2}\int_\mathcal{Y}\int_\mathcal{X}\int_\mathcal{X}(g(x)-g(x'))^2  \Pi_{X|Y}(\dif x'|y)\Pi_{X|Y}(\dif x|y)\Pi_Y(\dif y)\\
        &=\int_\mathcal{Y} \mathrm{Var}_{\Pi(\cdot|y)}(g)\Pi_Y(\dif y).
    \end{align*}
    We now establish the comparison result. Given the two expression of $\mathcal{E}_{\Pi_X}(\bar{P}_X,g)$ and $\mathcal{E}_{\Pi_X}(P_X,g)$ and \Cref{Asp:P2}, for any $s > 0$, we have
    \begin{align*}
        \begin{split}
            \mathcal{E}_{\Pi_X}(P_X,g)&=\int_\mathcal{Y} \mathrm{Var}_{\Pi(\cdot|y)}(g)\Pi_Y(\dif y)\\
            & \leqslant \int_\mathcal{Y} \left[s \cdot \mathcal{E}_{\Pi(\cdot|y)}(H_{2|y},g)+\beta(s,y) \cdot \|g\|^2_{\mathrm{osc}}\right] \Pi_Y(\dif y)\\
            &=s \cdot \mathcal{E}_{\Pi_X}(\bar{P}_X,g)+\|g\|^2_{\mathrm{osc}} \cdot \int_\mathcal{Y}\beta(s,y)\Pi_Y(\dif y) .
        \end{split}
    \end{align*}
    Then, we check that the function $\beta$ satisfies the regularity conditions of the WPI definition. Since $s \mapsto \beta(s, y)$ is monotone decreasing for each $y \in \mathcal{Y}$, it follows that $\beta$ is also monotone decreasing. Furthermore, since $\beta(\cdot, y) \leqslant 1/4$ pointwise for all $y \in \mathcal{Y}$, and $\int_\mathcal{Y} \Pi_Y(\dif y) = 1$, and each $\beta(s, y) \rightarrow 0$ as $s \rightarrow \infty$, the dominated convergence theorem yields that $\beta(s) =  \int_\mathcal{Y}\beta(s,y)\Pi_Y(\dif y) \rightarrow 0,$ as $s \rightarrow \infty$.
    Hence, the mapping $\beta: (0, \infty) \rightarrow [0, \infty)$ is decreasing, and satisfies $\lim_{s\rightarrow 0} \beta(s) = 0$.
\end{proof}

\begin{proof}[Proof of \Cref{Thm:SPIproduct}]
\label{prf:Thm:SPIproduct}  
Let $\gamma_1$ and $\gamma_2$ be the spectral gaps of $H_1$ and $H_2$, respectively: for any $h \in \mathrm{L}^2(\Pi_Y)$ and $g \in \mathrm{L}^2(\Pi_X)$,
\begin{align*}
    \gamma_1= \inf_{h: \mathrm{Var}_{\Pi_Y}(h) \neq 0} \frac{\mathcal{E}_{\Pi_Y}(H_1,h)}{\mathrm{Var}_{\Pi_Y}(h)}>0,\quad
    \gamma_2= \inf_{g: \mathrm{Var}_{\Pi_X}(g) \neq 0} \frac{\mathcal{E}_{\Pi_X}(H_2,g)}{\mathrm{Var}_{\Pi_X}(g)}>0.
\end{align*}
It is easy to see $c_1\leqslant \gamma_1$ and $c_2\leqslant \gamma_2$.
Then, let $C_{H_1}$ and {$C_{H_2}$} be the spectral gaps of $H_1^2$ and $H_2^2$ respectively. i.e. 
    \begin{align*}
        C_{H_1}= \inf_{h: \mathrm{Var}_{\Pi_Y}(h) \neq 0} \frac{\mathcal{E}_{\Pi_Y}(H_1^2,h)}{\mathrm{Var}_{\Pi_Y}(h)},\quad
        C_{H_2}= \inf_{g: \mathrm{Var}_{\Pi_X}(g) \neq 0} \frac{\mathcal{E}_{\Pi_X}(H_2^2,g)}{\mathrm{Var}_{\Pi_X}(g)}.
    \end{align*} 
    Since $H_1$ is reversible, we have
    \begin{align*}
        C_{H_1} = 1-\|H_1^2\|_{\Pi_Y}=1-\|H_1\|_{\Pi_Y}^2=1-(1-\gamma_1)^2.
    \end{align*}
    Since $\gamma_1\in[0,1]$, this establishes a one-to-one correspondence between $\gamma_1$ and $C_{H_1}$, and similarly between $\gamma_2$ and $C_{H_2}$. Then, we have the following inequality:
    \begin{align*}
        \left(1-(1-\gamma_1)^2\right) \cdot \mathrm{Var}_{\Pi_Y}(h) & \leqslant   \mathcal{E}_{\Pi_Y}(H_1^2,h).
    \end{align*}
    As in \Cref{rmk:strongtoweak} it then follows that:
    \begin{align*}
        \mathrm{Var}_{\Pi_Y}(h) & \leqslant s \cdot \mathcal{E}_{\Pi_Y}(H_1^2,h)+ \beta_1(s) \cdot \|h\|_\mathrm{osc}^2,
    \end{align*}
    where $\beta_1=\mathbbm{1}_{{\{1-(1-\gamma_1)^2 \leqslant 1/s \}}}$. 
    Similarly $\mathrm{Var}_{\Pi_X}(g) \leqslant s \cdot \mathcal{E}_{\Pi_X}(H_2^2,g)+ \beta_2(s) \cdot \|g\|_\mathrm{osc}^2$, where $\beta_2=\mathbbm{1}_{{\{1-(1-\gamma_2)^2 \leqslant 1/s \}}}$. 
    From \Cref{thm: tensor product}, we have that
    \begin{align*}
        \|f\|_\Pi^2 &\leqslant s \cdot \mathcal{E}_{\Pi}(H^2,f)+ \beta(s) \cdot \|f\|_\mathrm{osc}^2,
    \end{align*}
    where $\beta(s) \leqslant \beta_1(s)+\beta_2(s)= \mathbbm{1}_{{\{1-(1-\gamma_1)^2 \leqslant 1/s \}}} + \mathbbm{1}_{{\{1-(1-\gamma_2)^2 \leqslant 1/s \}}}$.
    
    Since $\beta(s)\leqslant 1$ for all $s\geqslant 0$, we have
    \begin{align*}
        \beta(s)\leqslant \mathbbm{1}_{{\{\min\{1-(1-\gamma_1)^2, 1-(1-\gamma_2)^2\} \leqslant 1/s \}}}.
    \end{align*}
    Hence, we can obtain the following inequality:
    \begin{align*}
        \min\{1-(1-\gamma_1)^2, 1-(1-\gamma_2)^2\} \cdot \|f\|_{\Pi}^2 \leqslant \mathcal{E}_{\Pi}(H^2,f).
    \end{align*}
    If we denote the spectral gap of $H^2$ as $C_P$, and the spectral gap of $H$ as $\gamma$, we have 
    \begin{align*}
        \min\{1-(1-\gamma_1)^2, 1-(1-\gamma_2)^2\} &\leqslant C_p = 1-(1-\gamma)^2.
    \end{align*}
    Since $1-(1-\gamma)^2$ is a strictly increasing function on $[0,1]$, $\min\{\gamma_1,\gamma_2\}\leqslant \gamma$, and so
    \begin{align*}
        \min\{\gamma_1,\gamma_2\} \cdot \|f\|_{\Pi}^2 &\leqslant \mathcal{E}_{\Pi}(H,f).
    \end{align*}
    Hence, as $\min\{c_1,c_2\}\leqslant \min\{\gamma_1,\gamma_2\}$, we have
    \begin{align*}
        \min\{c_1,c_2\} \cdot \|f\|_{\Pi}^2 &\leqslant \mathcal{E}_{\Pi}(H,f).
    \end{align*}
    
\end{proof}

\subsection{Proofs for
\Cref{Sec: toy example,Sec: bayes hierarchical model,Sec: diffusion}}

\begin{proof}[Proof of \Cref{eg: spectralgap_for_conditional}]
    \label{Prf:eg: spectralgap_for_conditional} 
    We first show the lower bound for
    $\gamma_{\tau}$. By  \cite[Theorem 1]{andrieu2024explicit}, for targeting the conditional
    distribution $\mathcal{N}\left(0,\frac{1}{\tau}\right)$ by RWM with the proposal
    increments $\mathcal{N}(0,\sigma_{\tau}^{2})$, we have that
    $\gamma_{\tau} \geqslant c_0 \cdot \sigma_{\tau}^{2}\cdot \tau \cdot \exp
    (-2\cdot\sigma_{\tau}^{2}\cdot \tau)$, where $c_0=1.972 \times 10^{-4}$.

    For the lower bound of $\gamma_{\xi}$, since $\tau|\xi \sim \Gamma\left(1,\beta
    +\frac{\xi^{2}}{2}\right)$ is not $m$-strongly convex, we make use of Theorem 18 and Lemma 38 of \cite{andrieu2024explicit}. Although Assumption 9 in \cite{andrieu2024explicit}, which requires the probability distribution to have a positive density on $\mathbb{R}$, does not hold in our case (the exponential distribution), the arguments used in the proofs of Theorem 18, Lemma 38 in \cite{andrieu2024explicit} remain valid for this distribution. Hence,
    \begin{align*}
        \gamma_{\xi} \geqslant 2^{-5}\cdot \epsilon^{2} \cdot \min \left\{1,4 \cdot \delta^{2} \cdot{I}_{\pi}\left(\frac{1}{4}\right)^{2}\right\},
    \end{align*}
    where $\epsilon=\frac{1}{2}\alpha_{0}$, $\delta=\alpha_{0} \cdot \sigma_{\xi}$,
    ${I}_{\pi}\left(p\right)=\beta_{\xi} \min\{p,1-p\}$ and $\beta_{\xi}=\beta
    +\frac{\xi^{2}}{2}$. $\alpha_{0}$ is minimal
    acceptance rate of $H_{1,\xi}$.

    It is direct to compute
    \begin{align*}
        \alpha_{0} & \geqslant \exp{\left(\frac{1}{2}\sigma_\xi^2\beta_\xi^2\right)}\cdot(1-\Phi(\beta_{\xi}\sigma_{\xi})) \geqslant\frac{1}{\sqrt{2\pi}}\cdot \frac{\beta_{\xi}\sigma_{\xi}}{\beta_{\xi}^{2}\sigma_{\xi}^{2}+1},
    \end{align*}
    since $1-\Phi(t) \geqslant \frac{1}{\sqrt{2\pi}}\cdot\frac{t}{t^{2}+1}\cdot\exp
    (-t^{2}/2)$ when $t\geqslant 0$. Hence, we have
    \begin{align*}
        \gamma_{\xi} \geqslant 2^{-9}\cdot \alpha_{0}^{4} \cdot \sigma_{\xi}^{2}\beta_{\xi}^{2} \geqslant c \cdot \frac{(\beta_{\xi}\sigma_{\xi})^{6}}{(\beta^{2}_{\xi}\sigma_{\xi}^{2}+1)^{4}},
    \end{align*}
    where $c=\pi^{-2}\times 2^{-11}$.
\end{proof}

\begin{proof}[Proof of \Cref{eg:spectralgap_for_DG}]
    \label{Prf:eg:spectralgap_for_DG} We consider the marginal chain of $\tau$.
    Let $\kappa=T(\tau)=1/\tau$, then for the Markov chain ${\tau_n}$,
    $\kappa_{n}=T(\tau_{n})$ is a Markov chain because $T$ is a bijection. Hence,
    they share the same convergence rate in total variation distance. We denote the transition kernel of the
    Markov chain ${\kappa_n}$ as $P_{K}(\kappa, \cdot)$, and the corresponding operator $P_{K}$.

    We write $\kappa_{n+1}$ as the auto regressive chain of $\kappa_{n}$, then
    we have
    \begin{align*}
        \kappa_{n+1} & = \frac{N^{2}}{2E}\cdot \kappa_{n} + \frac{\beta}{E},
    \end{align*}
    where $N$ is a standard normal random variable and $E$ is an exponential
    random variable with parameter $1$.

    Let $V(\kappa)=\kappa^{a}$, where $0<a<1$ is a constant. Since for any $x, y>0$, it is obvious $(x+y)^a\leqslant(x^a+y^a)$, and $\frac{N^{2}}{2E}$ and $\frac{\beta}{E}$ are both positive random variables, then we have
    \begin{align*}
        \mathbb{E}(V(\kappa_{n+1})|\kappa_{n}) & = \mathbb{E}\left(\frac{N^{2}}{2E}\cdot \kappa_{n} + \frac{\beta}{E}\right)^{a} \leqslant \mathbb{E}\left(\frac{N^{2}}{2E}\right)^{a} \cdot \kappa_{n}^{a} + \mathbb{E}\left(\frac{\beta}{E}\right)^{a}.
    \end{align*}
    Since $\mathbb{E}\left(\frac{N^{2}}{2E}\right)^{a}=\mathbb{E}((2E)^{-a}\mathbb{E}(N^{2a}|E))= \frac{\Gamma(a+1/2)}{\sqrt{\pi}}\mathbb{E}(E^{-a})=\frac{\Gamma(1-a)\Gamma(a+\frac{1}{2})}{\sqrt{\pi}}$, and $\mathbb{E}\left(\frac{\beta}{E}\right)^{a}=\beta^{a}\Gamma(1-a)$, we obtain 
    \begin{align*}
         \mathbb{E}(V(\kappa_{n+1})|\kappa_{n}) \leqslant \frac{\Gamma(1-a)\Gamma(a+\frac{1}{2})}{\sqrt{\pi}}V(\kappa_{n}) + \beta^{a}\Gamma(1-a),
    \end{align*}
    where $\Gamma(\cdot)$ is the Gamma function.
    
    Picking $a=\frac{1}{4}$,
    $\frac{\Gamma(1-a)\Gamma(a+\frac{1}{2})}{\sqrt{\pi}}$ attains it minimum
    value $\frac{\Gamma(\frac{3}{4})^{2}}{\sqrt{\pi}}<1$. We denote $\eta=\frac{\Gamma(\frac{3}{4})^{2}}{\sqrt{\pi}}$,
    $b=\beta^{1/4}\Gamma(\frac{3}{4})$, then we have for $V(\kappa)=\kappa^{1/4}:
    [0,\infty)\rightarrow [0,\infty)$, $P_K V \leqslant \eta V + b$.

    By Lemma 3.1 \cite{jones2004sufficient}, set $W (\kappa) = 1 + V (\kappa)$, $W
    : [0,\infty)\rightarrow [1,\infty)$, then for any $a>0$
    \begin{align*}
        P_{K} W \leqslant \rho W + L\mathbbm{1}_{C},
    \end{align*}
    where $\rho=(a + \eta)/(a + 1)$, $L=b+1-\eta$,
    $C=\left\{\kappa \in [0,\infty): W(\kappa)\leqslant \frac{(a+1)L}{a(1-\rho)}\right\}$.

    It is easy to see that $C$ is small, i.e., there exists a constant $\alpha>0$,
    and a probability measure $\mu$ such that for any $\kappa\in C$, we have
    \begin{align*}
        P_{K}(\kappa,\cdot) & \geqslant \alpha \cdot \mu(\cdot).
    \end{align*}
    Then, by \cite[Theorem 1]{taghvaei2021lyapunov}, we have
    \begin{align*}
        \gamma \|f\|_{\Pi_K}^{2} \leqslant \mathcal{E}_{\Pi_K}(P_K,f),
    \end{align*}
    where $\gamma=\frac{1+\frac{2L}{\alpha}}{1-\rho}$, $\Pi_K=T_{\#}\Pi_{T}$ is the pushforward measure, and $\Pi_{T}$ is the marginal distribution of $\tau$.

    We denote $P_{T}$ as the operator coresponding to the Markov chain $P_T(\tau,\cdot)$ for $\tau_n$. Since $\Pi_{T}$ is the invariant distribution of $\tau_n$, $\Pi_K$ is the invariant distribution of $\kappa_n$, we have $\|P_{T}
    \|_{\Pi_{T}}=\|P_{K}\|_{\Pi_K}=1-\gamma$. By \cite[Theorem 3.2]{liu1994covariance}, we have
    \begin{align*}
        \|P\|_{\Pi}^{2}=1-\gamma=\|P_{T}\|_{\Pi_T}.
    \end{align*}
    Hence,
    \begin{align*}
        \inf_{f\in \mathrm{L}^2_0(\Pi)}\frac{\mathcal{E}_{\Pi}(P^{*}P,f)}{\|f\|_{\Pi}^{2}} & =1-\sup_{f\in \mathrm{L}^2_0(\Pi)}\frac{\|Pf\|_{\Pi}^{2}}{\|f\|_{\Pi}^{2}}=\gamma ,\\
        \gamma\|f\|_{\Pi}^{2}                                                     & \leqslant \mathcal{E}_{\Pi}(P^{*}P,f).
    \end{align*}
\end{proof}

\begin{proof}[Proof of \Cref{eg: NIG_scale_sigma}]
    \label{Prf:eg: NIG_scale_sigma} By \Cref{rmk:requirementforH}, $H_{1|\xi}$ and $H_{2|\tau}$ are positive and reversible. It is easily shown that  when $\sigma_{\xi}\sim O
    (1/\beta_{\xi})$, the lower bound of $\gamma_{\xi}$ does not depend on $\xi$.
    When $\sigma_{\xi}=\sqrt{3}/\beta_{\xi}$, the lower bound of $\gamma_\xi$ attains the maximal value.
    Hence, $\beta_{1}(s)= \mathbbm{1}_{\{s\leqslant \gamma_\xi^{-1}\}}$, where $\gamma_{\xi}=\frac{27}{256}\times c$.

    Similarly, when $\sigma_{\tau}^{2}= 1/2\tau$, the lower bound of
    $\gamma_{\tau}$ is maximized and does not depend on $\tau$. Hence, $\beta_{2}(s)= \mathbbm{1}_{\{s\leqslant \gamma_\tau^{-1}\}}$, where $\gamma_{\tau}=\frac{c_0}{2e}$.
    We thus have $K_{1}^{*}(v)=\gamma_{\xi} \cdot v$ and $K_{2}^{*}(v)=\gamma_{\tau} \cdot v$.
    
    Then, we can obtain the following $\mathrm K^*$-WPIs by \Cref{thm:strongPtoP^*} and
    \Cref{thm:K*:P*toP2*}:
    \begin{align*}
        \frac{\mathcal{E}_{\Pi}\left(PP^{*}, f\right)}{\|f\|^{2}_{\text{osc}}}         &\geqslant \frac{\gamma \cdot \|f\|_{\Pi}^{2}}{\|f\|^{2}_{\text{osc}}}, \\
        \frac{\mathcal{E}_{\Pi}\left(P_{2}P_{2}^{*}, f\right)}{\|f\|^{2}_{\text{osc}}} \geqslant 2 \cdot K_{2}^{*}\left(\frac{1}{2}\cdot\frac{\mathcal{E}_{\Pi}\left(PP^{*}, f\right)}{\|f\|^{2}_{\text{osc}}}\right)&=\gamma_{\tau} \cdot \frac{\mathcal{E}_{\Pi}\left(PP^{*}, f\right)}{\|f\|^{2}_{\text{osc}}}\geqslant \gamma\cdot \gamma_\tau \cdot \frac{\|f\|_{\Pi}^{2}}{\|f\|^{2}_{\text{osc}}}.
    \end{align*}
    In this case, since 2-MwG $P_2P_2^*$ actually satisfies the SPI, we can apply \Cref{thm:strongPtoP^*} to obtain
    \begin{align*}
        \frac{\mathcal{E}_{\Pi}\left(P_{2}^*P_{2}, f\right)}{\|f\|^{2}_{\text{osc}}} \geqslant \gamma\cdot \gamma_\tau \cdot \frac{\|f\|_{\Pi}^{2}}{\|f\|^{2}_{\text{osc}}}.
    \end{align*}
    Finally, by \Cref{thm:K*:P1tildeP2totildeP1tildeP2} and \Cref{Thm:chaining},
    we can obtain the following $\mathrm K^*$-WPIs:
    \begin{align*}
        \frac{\mathcal{E}_{\Pi}\left(P_{1,2}^{*}P_{1,2}, f\right)}{\|f\|^{2}_{\text{osc}}} & \geqslant 2\cdot K_{1}^{*}\left(\frac{1}{2}\cdot \frac{\mathcal{E}_{\Pi}\left(P_{2}^{*}P_{2}, f\right)}{\|f\|^{2}_{\text{osc}}}\right)=\gamma_\xi \cdot \frac{\mathcal{E}_{\Pi}\left(P_{2}^{*}P_{2}, f\right)}{\|f\|^{2}_{\text{osc}}}\geqslant \gamma\cdot \gamma_\tau \cdot \gamma_\xi \cdot\frac{\|f\|_{\Pi}^{2}}{\|f\|^{2}_{\text{osc}}}.
    \end{align*}
    This yields $K^*(v)=\gamma_{\tau}\gamma_{\xi}\gamma \cdot v$. 
    By \Cref{rmk:strongtoweak}, we find 
    \begin{align*}
        \|P_{1,2}^nf \|^2_\Pi \leqslant \|f\|^2_\Pi \cdot \left(1-\gamma_\tau\cdot\gamma_\xi\cdot\gamma\right)^n\leqslant \|f\|^2_\Pi \cdot\exp\left(-\gamma_\tau\gamma_\xi\gamma\cdot n\right).
    \end{align*}
\end{proof}

\begin{proof}[Proof of \Cref{eg: NIG_fix_sigma}]
    \label{Prf:eg: NIG_fix_sigma}
    By \Cref{rmk:requirementforH}, $H_{1|\xi}$ and $H_{2|\tau}$ are positive and reversible. 
    Let $\sigma_\xi^2$=$\sigma_\tau^2=\sigma_0^2$. 
    From \Cref{eg: spectralgap_for_conditional}, we
    have
    \begin{align*}
        \gamma_{\xi} & \geqslant c \cdot \frac{(\beta_{\xi}\sigma_{0})^{6}}{(\beta^{2}_{\xi}\sigma_{0}^{2}+1)^{4}}\geqslant c' \frac{1}{\beta^{2}_{\xi}\sigma_{0}^{2}},
    \end{align*}
    where
    $c'=c \times \left(\frac{\beta^{2}\sigma_{0}^{2}}{\beta^{2}\sigma_{0}^{2}+1}\right)^{4}$
    since $\beta_{\xi}\geqslant \beta$. We denote $\gamma_{\xi}'=c' \frac{1}{\beta^{2}_{\xi}\sigma_{0}^{2}}$,
    then we can obtain $\beta_{1}(s,\xi)= \mathbbm{1}_{\{s\leqslant \gamma_\xi'^{-1}\}}$.
    
    Since the marginal distribution of $\xi$ is $t_{1}(0,\sqrt{2\beta})$, the Student's $t$-distribution with one degree of freedom, location $0$, and scale $\sqrt{2\beta}$, then,
    \begin{align*}
        \begin{split}\beta_{1}(s)&=\int \beta_{1}(s,\xi) \Pi_{\Xi}(\dif \xi)\\&=\int \mathbbm{1}_{\{s\leqslant \gamma_\xi'^{-1}\}}\cdot \frac{1}{\pi\cdot \sqrt{2\beta}}\cdot \left(1+\frac{\xi^{2}}{2\cdot \beta}\right)^{-1}\dif \xi\\&=\frac{2}{\pi}\left(\frac{\pi}{2}-\text{arctan}\left(\left(\frac{\left(c' s/\sigma_{0}^{2}\right)^{1/2}-\beta}{\beta}\right)^{1/2}\right)\right)
    \leqslant \frac{2}{\pi}\cdot \left(\frac{\beta}{\sqrt{\frac{c'\cdot s}{\sigma_{0}^{2}}}-\beta}\right)^{1/2}.\end{split}
    \end{align*}
    When $\sqrt{\frac{c'\cdot s}{\sigma_{0}^{2}}}>\beta$, it is obvious that we can
    find a constant $C_{1}(c',\beta,\sigma_{0}^{2})$ such that $\beta_{1}(s)\leqslant C_{1} \cdot s^{-\frac{1}{4}}$.

    Similarly, we denote $\gamma_{\tau}'=c_0 \cdot \sigma_{0}^{2}\cdot \tau
    \cdot \exp(-2\cdot\sigma_{0}^{2}\cdot \tau)$, then $\beta_{2}(s,\tau) = \mathbbm{1}_{\{s\leqslant \gamma_\tau'^{-1}\}}$.

    Since the marginal distribution of $\tau$ is
    $\Gamma\left(\frac{1}{2},\beta\right)$,
    \begin{align*}
        \begin{split}\beta_{2}(s)&=\int \beta(s,\tau) \Pi_{Y}(\dif \tau)\\&=\int \mathbbm{1}_{\{s\leqslant \gamma_\tau'^{-1}\}}\frac{\beta^{1/2}}{\sqrt{\pi}}\tau^{-1/2}\exp(-\beta \tau)\dif \tau\\&=\frac{1}{\sqrt{\pi}}\left(\Gamma_{L}\left(\frac{1}{2},-\frac{\beta}{2\sigma_{0}^{2}}w_0\left(-\frac{2}{c_0\cdot s}\right)\right)+\Gamma_{U}\left(\frac{1}{2},-\frac{\beta}{2\sigma_{0}^{2}}w_{-1}\left(-\frac{2}{c_0\cdot s}\right)\right)\right),\end{split}
    \end{align*}
    where $s\geqslant 2e/c_0$,  $w_0$ and $w_{-1}$ are two real branches of the Lambert W function that give real solutions of $w(x)e^{w(x)} = x$, $\Gamma_{L}$ is the lower incomplete
    gamma function and $\Gamma_{U}$ is the upper incomplete gamma function.

    When
    $-\frac{\beta}{2\sigma_{0}^{2}}w_0\left(-\frac{2}{c_0\cdot s}\right)\leqslant1$, we  have $\Gamma_{L}(t,x) \leqslant 1/t \cdot x^{t}$ by \cite[(5)]{jameson2016incomplete}. 
    For $-1<w_0(-x)<0$, let $w=w_0(-x)$. Hence, $1\leqslant e^{w+1}$ and $-w\leqslant -e \cdot we^w$. Since $x=we^w$, we have $-w_0(-x)\leqslant e\cdot x$. Then
    \begin{align*}
        \Gamma_{L}\left(\frac{1}{2},-\frac{\beta}{2\sigma_{0}^{2}}w_0\left(-\frac{2}{c_0\cdot s}\right)\right)\leqslant  2\cdot\left(-\frac{\beta}{2\sigma_{0}^{2}}w_0\left(-\frac{2}{c_0\cdot s}\right)\right)^{1/2} 
        \leqslant \frac{2\cdot \sqrt{e \cdot \beta}}{\sqrt{c_0 \cdot \sigma_{0}^{2}}}\cdot s^{-\frac{1}{2}}.
    \end{align*}

    When $-\frac{\beta}{2\sigma_{0}^{2}}w_{-1}\left(-\frac{2}{c_0\cdot s}\right)>1$,
    \begin{align*}
        \Gamma_U\left(\frac{1}{2},-\frac{\beta}{2\sigma_{0}^{2}}w_{-1}\left(-\frac{2}{c_0\cdot s}\right)\right)\leqslant & \exp{\left(\frac{\beta}{2\sigma_{0}^{2}} w_{-1}\left(-\frac{2}{c_0\cdot s}\right)\right)}\cdot \left(-\frac{\beta}{2\sigma_{0}^{2}} w_{-1}\left(-\frac{2}{c_0\cdot s}\right)\right)^{-\frac{1}{2}} \\
        \leqslant & \left(\frac{2}{c_0}\right)^{\frac{\beta}{2\sigma_{0}^{2}}}\cdot s^{-\frac{\beta}{2\sigma_{0}^{2}}} ,
    \end{align*}

    since $\Gamma_{U}(t,x) \leqslant \exp(-x) \cdot x^{t-1}\leqslant \exp(-x)$ if
    $x>1$ and $t=1/2$, see \cite[(6)]{jameson2016incomplete}, and $w_{-1}(x) \leqslant \ln(-x)$ as $-1/e\leqslant x<0$, see (9) in \cite{loczi2020explicit}.

    We denote
    $C_{\sigma_0}=\frac{1}{\sqrt{\pi}}\frac{2\cdot \sqrt{e\cdot\beta}}{\sqrt{c_0
    \cdot \sigma_{0}^{2}}}$
    and
    $C_{\sigma_0}'=\frac{1}{\sqrt{\pi}}\left(\frac{2}{c_0}\right)^{\frac{\beta}{2\sigma_{0}^{2}}}$.
    When $\beta>\sigma_{0}^{2}$, $s^{-\frac{1}{2}}$ dominates, hence $\beta_{2}(s
    ) \leqslant C_{\sigma_0}\cdot s^{-\frac{1}{2}}$. Otherwise, $\beta_{2}(s) \leqslant
    C_{\sigma_0}'\cdot s^{-\frac{\beta}{2\sigma_{0}^{2}}}$.

    We combine the result of fixing $\sigma_{0}^{2}$ in every step. In that case, taking $C_{2}=\max\{C_{\sigma_0},C_{\sigma_0}'\}$, we have
    \begin{align}
        \beta_{1}(s) & \leqslant C_{1} \cdot s^{-\frac{1}{4}},\nonumber\\
        \beta_{2}(s) & \leqslant C_{2} \cdot \max\{s^{-\frac{1}{2}}, s^{-\frac{\beta}{2\sigma_{0}^{2}}}\}=C_{2} s^{\max(-\frac{1}{2},-\frac{\beta}{2\sigma_{0}^{2}})}.\label{eq:beta2}
    \end{align}
    By \cite[Lemma 15]{andrieu2022comparison},
    \begin{align*}
        K_{1}^{*}(v)  =C'_{1} \cdot v^{1+4},   \quad         
        K_{2}^{*}(v)  =C'_{2} \cdot v^{1+1/\min\{\frac{1}{2},\frac{\beta}{2\sigma_{0}^{2}}\}},
    \end{align*}
    where $C'_{1}=\frac{C_1/4}{(C_1(1+1/4))^{1+4}}$, $C_2'=\frac{C_2C_0}{(C_2(1+C_0))^{1+C_0^{-1}}}$, $C_0=\max(-\frac{1}{2},-\frac{\beta}{2\sigma_{0}^{2}})$.
    By \Cref{thm:strongPtotildeP}, we have
    \begin{align*}
        K^{*}(v)=2\cdot K_{1}^{*}(K_{2}^{*}(\frac{\gamma}{4}\cdot v))=C'\cdot v^{5\cdot (1+1/\min\{\frac{1}{2},\frac{\beta}{2\sigma_{0}^{2}}\})},
    \end{align*}
    where
    $C'=2C_{1}'C_{2}'^{5}\left(\gamma/4\right)^{5\cdot (1+1/\min\{\frac{1}{2},\frac{\beta}{2\sigma_{0}^{2}}\})}$.

    Hence, by Lemma 15 in \cite{andrieu2022comparison},
    \begin{align*}
        F^{-1}(n)\leqslant \begin{cases}\tilde{C}_{1} \cdot n^{-\frac{1}{14}}\quad\quad\quad\quad&\text{if}\quad \frac{\beta}{\sigma^2_0}>1,\\ \tilde{C}_{2} \cdot n^{-\frac{\beta}{4\beta+10\sigma^2_0}}&\text{otherwise}.\end{cases}
    \end{align*}
    where $\tilde{C}_{1}=(\frac{1}{14C'})^{\frac{1}{14}}$, $\tilde{C}_{2}=\left(\frac{\beta}{(4\beta+10\sigma^2_{0})C'}
    \right)^{\frac{\beta}{4\beta+10\sigma^2_{0}}}$.
\end{proof}

\begin{proof}[Proof of \Cref{eg: NIG_fix_sigma_P2}]
\label{Prf:eg: NIG_fix_sigma_P2}
    As shown in \Cref{Prf:eg: NIG_fix_sigma} \Cref{eq:beta2},
    \begin{align*}
        \beta_{2}(s) & \leqslant C_{2} \cdot \max\{s^{-\frac{1}{2}}, s^{-\frac{\beta}{2\sigma_{0}^{2}}}\}=C_{2} s^{\max(-\frac{1}{2},-\frac{\beta}{2\sigma_{0}^{2}})},\\
        K_{2}^{*}(v)  &=C'_{2} \cdot v^{1+1/\min\{\frac{1}{2},\frac{\beta}{2\sigma_{0}^{2}}\}},
    \end{align*}
    where $C_2$ and $C_2'$ are universal constants as in \Cref{eq:beta2}.
    By \Cref{eg:spectralgap_for_DG} and \Cref{thm:weakPtotildeP2}, we have 
    \begin{align*}
        K^*(v) =K_2^*\circ \frac{1}{2}K_0^*(v)=K_2^*(\gamma \cdot v/2)=C' \cdot v^{1+1/\min\{\frac{1}{2},\frac{\beta}{2\sigma_{0}^{2}}\}},
    \end{align*}
    where $C'=C'_{2} \cdot \left(\frac{\gamma}{2}\right)^{1+1/\min\{\frac{1}{2},\frac{\beta}{2\sigma_{0}^{2}}\}}$.
    
    Hence, by Lemma 15 in \cite{andrieu2022comparison} and \Cref{thm:weakPtotildeP2},
    \begin{align*}
        F^{-1}(n)\leqslant \begin{cases}\tilde C'_1 \cdot (n-1)^{-\frac{1}{2}}\quad\quad\quad\quad&\text{if}\quad \frac{\beta}{\sigma^2_0}>1,\\ \tilde C'_2 \cdot (n-1)^{-\frac{\beta}{2\sigma^2_0}}&\text{otherwise}.\end{cases}
    \end{align*}
    where $\tilde C'_1=(\frac{1}{2C'})^{1/2},  \tilde C'_2=(\frac{\beta}{2\sigma_0^2C'})^{\beta/2\sigma_0^2}$.
\end{proof}

\begin{lemma}
\label{thm: rate_MH}
    Let $P_{\MH}(x,\dif x')$ denote the Gaussian RWM kernel targeting $\pi=t_1(0, \sqrt{2\beta})$ with step size $\sigma_0^2$. Then for any $K^*$-WPI that $P_{\MH}(x,\dif x')$ satisfies, there exists a constant $C$ such that for $0\leqslant v\leqslant 1/4$,
    \begin{align*}
        K^*(v) \leqslant C v^{3}, 
    \end{align*}
    Hence, the corresponding convergence bound is $\|P_{\MH}^nf\|^2\leqslant F^{-1}(n)\|f\|_{\mathrm{osc}}^2$, where $F$ is as in \Cref{Thm:rate} and $F^{-1}(n) \geqslant \tilde{C} n^{-1/2}$, where $\tilde{C}$ is a universal constant.
\end{lemma} 
\begin{proof}
    Consider for any $r>1$, the function $f_r: \mathbb R \to \mathbb R$,
    \begin{align*}
        f_r(x)= \begin{cases}0, \quad&  x\leqslant r,\\ \frac{x-r}{r}, &r<x<2r,\\ 1, &x\geqslant 2r. \end{cases}
    \end{align*} 
    Next, we estimate $\mathrm{Var}_\pi (f_r)$ and $\mathcal{E}_\pi(P_{\MH},f_r)$ as a function of $r$. Denote $m=\mathbb{E_\pi}(f_r)$,
    \begin{align*}
        \mathrm{Var}_\pi (f_r)&= \int (f_r-\mathbb{E_\pi}(f_r))^2\pi(\dif x)\geqslant m^2 \pi((-\infty,r])+(1-m)^2 \pi([2r,\infty))\\
        &\geqslant\frac{\pi([2r,\infty))\pi((-\infty,r])}{\pi([2r,\infty))+\pi((-\infty,r])} \geqslant \frac{1}{2}\pi([2r,\infty))\geqslant\frac{1}{2\pi} \arctan \left(\frac{\sqrt{2\beta}}{2r}\right)\geqslant \frac{\sqrt{2\beta}}{8\pi r},
    \end{align*}
    where the inequality in the first line holds by setting the integrand to 0 on $(r, 2r)$; the second inequality follows by minimizing the resulting quadratic form in $m$; the third inequality uses $\pi((-\infty,r])>1/2$ when $r>1$, and  the last inequality holds for $r>1$. Denote $c=\frac{\sqrt{2\beta}}{8\pi}$.
    
    Then,
    \begin{align*}
        \mathcal{E}(P_{\MH},f_r)&=\frac{1}{2}\int (f_r(x)-f_r(x'))^2P_{\MH}(x,\dif x')\pi(\dif x)\\
        &\leqslant \frac{1}{2}\int (f_r(x)-f_r(x+z))^2 Q_\sigma(\dif z )\pi(\dif x),
    \end{align*}
    where $Q_\sigma(\dif z)$ is the Gaussian distribution with variance $\sigma^2$.
    
        Since $f_r(x)$ is $1/r$-Lipchitz, $(f_r(x)-f_r(x+z))^2\leqslant \frac{z^2}{r^2}$. Also, $(f_r(x)-f_r(x+z))^2\leqslant 1$
        Moreover, one can show that $f_r(x)-f_r(x+z) \neq 0$ if and only if either $z>0$, $x \in (r-z,2r)$  or $z<0$, $x \in (r,2r-z)$. Therefore,
    \begin{align*}
        \mathcal{E}(P_{\MH},f_r)&\leqslant \frac{1}{2} \int_{|z|\leqslant r/2}\left(\int_{r-|z|}^{2r+|z|} \frac{z^2}{r^2} \pi(\dif x)\right)Q_\sigma(\dif z )+ \frac{1}{2} \int_{|z|\leqslant r/2} \left(\int  \pi(\dif x)\right)Q_\sigma(\dif z )\\
        &\leqslant \frac{1}{2} \int_{|z|\leqslant r/2}\left(\frac{z^2}{r^2} \cdot \frac{\sqrt{2\beta}}{\pi}\frac{r+2|z|}{(r-|z|)^2}\right)Q_\sigma(\dif z ) + \frac{1}{2} \int_{|z|\leqslant r/2} Q_\sigma(\dif z )\\
        &\leqslant \frac{1}{2} \frac{\sqrt{2\beta}}{\pi}\int_{|z|\leqslant r/2}\left(\frac{z^2}{r^2} \cdot \frac{8}{r}\right) Q_\sigma(\dif z ) + \frac{1}{2} \int_{|z|\leqslant r/2} Q_\sigma(\dif z )  \quad \text{ as $|z|\leqslant r/2$}\\
        &\leqslant \frac{1}{2} \frac{\sqrt{2\beta}}{\pi}\frac{6}{r^3}\int z^2 Q_\sigma(\dif z ) + \frac{1}{2} \int_{|z|\leqslant r/2} Q_\sigma(\dif z )\\
        &\leqslant \frac{4\sqrt{2\beta}\sigma^2}{\pi}\frac{1}{r^3}+ \frac{1}{2} \exp(-r^2/8)\leqslant C' \frac{1}{r^3},
    \end{align*}
    where the second inequality is because the density $\pi=t_1(0, \sqrt{2\beta})$ can be bounded $\frac{\sqrt{2\beta}}{\pi}\frac{1}{2\beta +x^2}\leqslant \frac{\sqrt{2\beta}}{\pi}\frac{1}{x^2} \leqslant \frac{\sqrt{2\beta}}{\pi}\frac{1}{(r-|z|)^2}$.
    
    Therefore, we show that, for $f_r(x)$ with $\|f_r\|^2_\mathrm{osc}=1$, as $K^*$ is increasing, 
    \begin{align}
        K^*(v) \leqslant K^*\left(\frac{\mathrm{Var}_\pi(f_r)}{\|f_r\|^2_\mathrm{osc}}\right) \leqslant \frac{\mathcal{E}(P_{\MH}, f_r)}{\|f_r\|^2_\mathrm{osc}} \leqslant C v^3,
    \end{align}
    where let $v=c/r$, so $C=C'/c^3$.
    
    As the convergence bound $n^{-1/2}$ corresponds to the $K^*$ function $C v^3$, $F^{-1}(n) \geqslant \tilde{C}n^{-1/2}$ based on Lemma 15 in \cite{andrieu2022comparison} and Theorem 13 in \cite{andrieu2026weak}, where $F$ is as in \Cref{Thm:rate}.
\end{proof}

\begin{proof}[Proof of \Cref{eg:sharp bound}]
    \label{prf:eg:sharp bound}
    By \Cref{eg: NIG_fix_sigma_P2}, we have $\|P_{2}^nf \|^2_\Pi \leqslant C'\cdot (n-1)^{-\frac{1}{2}}\cdot \|f\|^2_\mathrm{osc} \leqslant \|P_{2}^nf \|^2_\Pi \leqslant 2C'\cdot n^{-\frac{1}{2}}\cdot \|f\|^2_\mathrm{osc}$. Then we want to show the sharpness of $n^{-\frac{1}{2}}$.

    Let $P_\Xi(\xi, \dif \xi')=\int_\tau \Pi(\dif \tau|\xi)H_{2|\tau}(\dif \xi' |\xi)\dif \tau$ denote the $\xi$-marginal chain of $P_2$. Then, the acceptance rate associated with $P_\Xi(\xi, \dif \xi')$ is 
    \begin{align*}
        \alpha_G(\xi,\xi')=\int \min\left\{1, \frac{\Pi(\xi'|\tau)}{\Pi(\xi|\tau)}\right\} \Pi_T( \dif \tau |\xi)=\int \min\left\{1, e^{-\frac{\tau}{2}(\xi'^2-\xi^2)}\right\} \Pi_T( \dif \tau |\xi).
    \end{align*}
    When $\xi'^2\geqslant\xi^2$, 
    \begin{align*}
        \alpha_G(\xi,\xi')=\int e^{-\frac{\tau}{2}(\xi'^2-\xi^2)} \Pi_T( \dif \tau|\xi)=\int e^{-\frac{\tau}{2}(\xi'^2-\xi^2)}\cdot\left(\beta+\frac{\xi^2}{2}\right)e^{-(\beta+\frac{\xi^2}{2})\tau} \dif \tau=\frac{\beta+\xi^2/2}{\beta+\xi'^2/2};
    \end{align*}
    when $\xi'^2<\xi^2$, $\alpha_G(\xi,\xi')=\int \Pi_T( \dif \tau|\xi)=1$.
    Hence, $\alpha_G(\xi,\xi')=\min\left\{1, \frac{\beta+\xi^2/2}{\beta+\xi'^2/2}\right\}.$

    As the marginal distribution $\Pi_\Xi$ of $\xi$ is $t_1(0, \sqrt{2\beta})$, we consider $P_{\MH}(\xi,\dif \xi')$ as targeting $\Pi_\Xi$ via a Gaussian proposal RWM with step size $\sigma_0^2$. The acceptance rate $\alpha_{\MH}(\xi,\xi')=\min\left\{1, \frac{\Pi_\Xi(\xi')}{\Pi_\Xi(\xi)}\right\}=\min\left\{1, \frac{\beta+\xi^2/2}{\beta+\xi'^2/2}\right\}=\alpha_G(\xi,\xi')$. One the other hand, the proposal of $P_\Xi(\xi, \dif \xi')$ is $q_G(\xi,\xi')= \int \frac{1}{\sqrt{2\pi}\sigma_0}\exp(-\frac{(\xi-\xi')^2}{2\sigma_0^2})\dif \Pi_T( \dif \tau|\xi)=q_\MH(\xi,\xi')$, where $\sigma_0$ is the fixed step size. Thus $P_\Xi=P_\MH$.
    
    Therefore, as in \Cref{thm: rate_MH} $n^{-\frac{1}{2}}$ is sharp for $P_{\MH}(\xi,\dif \xi')$, it is sharp for $P_\Xi(\xi, \dif \xi')$.
    As shown in \Cref{thm:ApproxPtoP_X} and \Cref{thm:ApproxP_XtoP}, $n^{-\frac{1}{2}}$ is also sharp for $P_2$. 
\end{proof}

\begin{lemma}
\label{thm:Px_L1_to_P_SPI} 
We denote by $T_{X}$ the $X$-marginal chain of $T$, where $T((x,y),(\dif x',\dif y'))=\Pi_{X|Y}(\dif x'|y')\Pi_{Y|X}(\dif y'|x)$, and $T_{X}(x,\dif x')=\int_{\mathcal{Y}} \Pi_{X|Y}(\dif x'|y)\Pi_{Y|X}(\dif y|x)$. If $T_{X}$ is $\mathrm{L^{1}}$-geometrically convergent, then
    $T^{*}T$ satisfies a SPI with some $\gamma>0$ as in \eqref{eq:SPI}.
    
\end{lemma}
\begin{proof}
\label{prf:thm:Px_L1_to_P_SPI}
\Cref{lem:P_xreversible} shows the $X$-marginal chain $T_{X}$ is $\Pi_X$-reversible. 
    By \cite[Theorem 2.1]{roberts1997geometric}, since  $T_{X}$ is reversible and $\mathrm{L^{1}}$-geometrically convergent, then $T_{X}$ has a $\mathrm{L^{2}}(\Pi_X)$ spectral gap $\gamma>0$. By \cite[Theorem 3.2]{liu1994covariance},
    \begin{align*}
        \|T\|_{\Pi}^{2}=\|T_{X}\|_{\Pi_X}=1-\gamma,
    \end{align*}
    where the last equality follows from the reversibility of $T_X$.
    Hence,
    \begin{align*}
        \inf_{f\in \mathrm{L}^2_0(\Pi)}\frac{\mathcal{E}_{\Pi}(T^{*}T,f)}{\|f\|^{2}_{\Pi}} & =1-\sup_{f\in \mathrm{L}^2_0(\Pi)}\frac{\|Tf\|_{\Pi}^{2}}{\|f\|_{\Pi}^{2}}=\gamma,\\
        \Longrightarrow \quad \gamma\|f\|_{\Pi}^{2}                                                      & \leqslant \mathcal{E}_{\Pi}(T^{*}T,f).
    \end{align*}
\end{proof}

\begin{lemma}
    \label{thm:driftforBayes} 
    In the setting of \Cref{Sec: bayes hierarchical model}, we write $P_{\mathcal B}$ for the marginal chain of $\beta$, so $P_{\mathcal B}(\beta, \dif \beta')=\int_{\sigma^{-2}} \pi_{\beta|\sigma^{-2}}(\dif \beta'|\sigma^{-2})\pi_{\sigma^{-2}|\beta}(\dif \sigma^{-2}|\beta)$. $P_{\mathcal B}$ is $\Pi(\beta)$-reversible by \Cref{lem:P_xreversible}.
    
    Then, there exists a constant $\rho \in [0, 1)$ and a finite constant
    $L$ such that, for every $\beta' \in \mathbb{R}^{p}$,
    \begin{align*}
        \mathbb{E}(V(\beta)|\beta') \leqslant \rho \cdot V(\beta)+L,
    \end{align*}
    where the drift function is defined as $V(\beta)= \alpha\|\mathbf Y-\mathbf X\beta\|^{2}$,
    which is unbounded off petite sets. The constant $\alpha$ can be any real
    number that satisfies $\alpha>0$.
\end{lemma}
\begin{proof}
Denoting $\lambda=\sigma^{-2}$, we have
    \begin{align*}
        \mathbb{E}(V(\beta)|\beta') = \mathbb{E}(\mathbb{E}(V(\beta)|\lambda)|\beta')=\mathbb{E}(\mathbb{E}(\alpha\|\mathbf Y-\mathbf X\beta\|^{2}|\lambda)|\beta').
    \end{align*}
    Since
    \begin{align*}
        \mathbb{E}(\|\mathbf Y-\mathbf X\beta\|^{2}|\lambda) & =\text{tr}(\mathbf X\text{Var}(\beta|\lambda)\mathbf X^\top)+\|\mathbf X\mathbb{E}(\beta|\lambda)-\mathbf Y\|^{2}        \\
                                             & = \frac{1}{\lambda}\cdot\text{tr}(\mathbf X(\mathbf X^\top\mathbf X)^{-1}\mathbf X^\top)+\|\mathbf Y-\mathbf X(\mathbf X^\top\mathbf X)^{-1}(\mathbf X^\top\mathbf Y)\|^{2} \\
                                             & =\frac{1}{\lambda}\cdot p+\|\mathbf Y-\mathbf X(\mathbf X^\top\mathbf X)^{-1}(\mathbf X^\top\mathbf Y)\|^{2},
    \end{align*}
    we can obtain
    \begin{align*}
        \mathbb{E}(\mathbb{E}(\|\mathbf Y-\mathbf X\beta\|^{2}|\lambda)|\beta') & =p \cdot \mathbb{E}\left(\frac{1}{\lambda}|\beta'\right)+\|\mathbf Y-\mathbf X(\mathbf X^\top\mathbf X)^{-1}(\mathbf X^\top\mathbf Y)\|^{2}             \\
      & =p \cdot \frac{b+\frac{\|\mathbf Y-\mathbf X\beta'\|^{2}}{2}}{a+\frac{N}{2}-1}+\|\mathbf Y-\mathbf X(\mathbf X^\top\mathbf X)^{-1}(\mathbf X^\top\mathbf Y)\|^{2}       \\
        & =\frac{p}{2a+N-2} \|\mathbf Y-\mathbf X\beta'\|^{2} +\frac{p \cdot b}{a+\frac{N}{2}-1}+\|\mathbf Y-\mathbf X(\mathbf X^\top\mathbf X)^{-1}(\mathbf X^\top\mathbf Y)\|^{2}.
    \end{align*}
    Thus, we obtain
    \begin{align*}
        \mathbb{E}(V(\beta)|\beta') \leqslant \rho \cdot V(\beta')+L,
    \end{align*}
    where $\rho =\frac{ p}{2a+N-2}<1$, $L=\frac{\alpha \cdot p \cdot b}{a+\frac{N}{2}-1}
    +\alpha \|\mathbf Y-\mathbf X(\mathbf X^\top\mathbf X)^{-1}(\mathbf X^\top\mathbf Y)\|^{2}$.
\end{proof}

\begin{proof}[Proof of \Cref{eg:bayes:spectralgap_DG}]
    \label{Prf:eg:bayes:spectralgap_DG} By \Cref{thm:driftforBayes}, the drift
    function is given by $V(\beta)=\alpha\|\mathbf Y-\mathbf X\beta\|^{2}$. Consider the set $C_{m}=\left\{\beta \in \mathbb{R}^{p}: \alpha\|\mathbf Y-\mathbf X\beta\|^{2} \leqslant m\right\}$,
    where $m$ is a positive constant. When $\mathbf X$ has full column rank, $C_{m}$ is a compact set
    and $V(\beta)$ is unbounded off $C_{m}$. By
    Lemma 15.2.8. and Theorem 15.0.2. in \cite{meyn2012markov}, we have the marginal
    chain of $\beta$ $P_{\mathcal B}$ is $\mathrm{L}^1$-geometrically ergodic with some rate $\rho$, where
    $\rho<1$.

    By \Cref{thm:Px_L1_to_P_SPI}, there exists a constant $\gamma>0$ such that $\gamma \|f\|_\Pi^{2} \leqslant \mathcal{E}_\Pi(P^{*}P,f)$.
    
\end{proof}

\begin{proof}[Proof of \Cref{eg:rate for Bayes}]
    \label{Prf:eg:rate for Bayes} 
    By \Cref{rmk:requirementforH}, $H_{2|{\sigma^{-2}}}$ is positive and reversible. Denoting $\lambda=\sigma^{-2}$, we have
    \begin{align*}
        \pi(\beta, \lambda|\mathbf Y) & \propto \pi(\mathbf Y|\beta, \lambda)\pi(\beta)\pi(\lambda)                                                                                                                               \\
                              & \propto\lambda^{\frac{N}{2}}\exp\left(-\frac{\lambda}{2}(\mathbf Y-\mathbf X\beta)^\top(\mathbf Y-\mathbf X\beta)\right)\cdot \lambda^{a-1}\exp(-b\lambda)                                                         \\
                              & \propto \lambda^{a-1+\frac{N}{2}}\exp(-b\lambda) \exp\left(-\frac{1}{2}(\beta-u)^\top Q(\beta-u)\right)  \exp\left(-\frac{1}{2}\lambda \mathbf Y^\top\mathbf Y+\frac{1}{2}u^\top Qu\right),
    \end{align*}
    where $Q:=\lambda \mathbf X^\top\mathbf X$, $u:=(\mathbf X^\top\mathbf X)^{-1}\mathbf X^\top\mathbf Y$. Then, we can derive the marginal
    distribution of $\lambda$ by integrating $\beta$:
    \begin{align*}
        \pi(\lambda|\mathbf Y) & =\int \pi(\beta, \lambda|\mathbf Y) \pi(\dif \beta)                                                                                              \\
                       & \propto \lambda^{a-1+\frac{N}{2}}\exp(-b\lambda) \cdot |Q^{-1}|^{\frac{1}{2}}\cdot \exp\left(-\frac{1}{2}(\lambda \mathbf Y^\top\mathbf Y-u^\top Qu)\right) \\
                       & \propto \lambda^{a-1+\frac{N}{2}-\frac{p}{2}}\cdot \exp\left(-\left(b+\frac{\mathbf Y^\top\mathbf Y-u^\top(\mathbf X^\top\mathbf X)u}{2}\right)\lambda\right).
    \end{align*}
    Hence, $\lambda \sim \Gamma(a',b')$, where $a'=a+\frac{N}{2}-\frac{p}{2}$ and
    $b'=b+\frac{\mathbf Y^\top\mathbf Y-u^\top(\mathbf X^\top\mathbf X)u}{2}$. Furthermore
    \begin{align*}
        \mathbf Y^\top\mathbf Y-u^\top(\mathbf X^\top\mathbf X)u=\mathbf Y^\top(I-\mathbf X(\mathbf X^\top\mathbf X)^{-1}\mathbf X^\top)\mathbf Y \geqslant 0,
    \end{align*}
    since $\mathbf X(\mathbf X^\top\mathbf X)^{-1}\mathbf X^\top$ is a projection. So the parameters are proper.

    For the conditional distribution $\pi(\beta |\lambda)$, for each $\lambda$,
    it is possible to bound the spectral gap $\gamma_{\lambda}$ of $H_2$ by conductance as in \cite{andrieu2024explicit}.

    $\pi(\beta |\lambda)$ has the density
    $\pi(\beta |\lambda) \propto \exp (-U (\beta))$, where $U$ is $m$-strongly convex
    and $L$-smooth with $m=\lambda \cdot \lambda_{\min}(\mathbf X^\top\mathbf X)$ and $L=\lambda \cdot
    \lambda_{\max}(\mathbf X^\top\mathbf X)$. Let $\lambda_{\max}(\mathbf X^\top\mathbf X)$ and 
    $\lambda_{\min}(\mathbf X^\top\mathbf X)$ denote the maximal and minimal eigenvalues of
    $\mathbf X^\top\mathbf X$, respectively.

    By \cite[Theorem 1]{andrieu2024explicit}, then the spectral gap
    $\gamma_{\lambda}$ of $H_{2}$ satisfies
    \begin{align*}
        \gamma_{\lambda} & \geqslant C \cdot L \cdot p \cdot \sigma_{0}^{2} \cdot \exp \left(-2 \cdot L \cdot p \cdot \sigma_{0}^{2}\right) \cdot \frac{m}{L}\cdot \frac{1}{p} \\
                         & = C \cdot m \cdot \sigma_{0}^{2} \cdot\exp \left(-2 \cdot L \cdot p \cdot \sigma_{0}^{2}\right)                                                     \\
                         & = C \cdot \lambda \cdot m' \cdot \sigma_{0}^{2} \cdot \exp \left(-2 \cdot \lambda \cdot L' \cdot p \cdot \sigma_{0}^{2}\right),
    \end{align*}
    where $C=1.972 \times 10^{-4}$, $m'=\lambda_{\min}(\mathbf X^\top\mathbf X)$,
    $L'=\lambda_{\max}(\mathbf X ^\top\mathbf X)$.

    Since $\beta(s,\lambda)=1_{\{1/s\geqslant\gamma_\lambda\}}$, we are interested in $\lambda$ such that $1/s\geqslant\gamma_{\lambda}$ holds, that is
    \begin{align*}
        s \leqslant C_{1} \cdot \lambda^{-1}\cdot \exp (C_{2} \cdot \lambda),
    \end{align*}
    where $C_{1}=(C \cdot m' \cdot \sigma_{0}^{2})^{-1}$, $C_{2}=2\cdot L' \cdot
    p \cdot \sigma_{0}^{2}$.
    Since the marginal distribution of $\lambda$ is $\Gamma(a',b')$,
    \begin{align*}
        \begin{split}\beta(s)&=\int \beta(s,\lambda) \Pi(\dif \lambda)\\&=\int 1_{\{1/s\geqslant\gamma_\lambda\}}\Pi(\dif \lambda)\\&=\frac{1}{\Gamma(a')}\left(\Gamma_\mathrm{L}\left(a',-b' C_{2}^{-1}\cdot w_{0}\left(-\frac{C_{1}C_{2}}{ s}\right)\right)+\Gamma_\mathrm{U}\left(a',-b' C_{2}^{-1}\cdot w_{-1}\left(-\frac{C_{1}C_{2}}{ s}\right)\right)\right),\end{split}
    \end{align*}
    where $s\geqslant e C_1C_2$, $w_0$ and $w_{-1}$ are two real branches of the Lambert W function that give real solutions of $w(x)e^{w(x)} = x$, 
    $\Gamma$ is the gamma function, $\Gamma_\mathrm{L}$ is the lower incomplete gamma function
    and $\Gamma_\mathrm{U}$ is the upper incomplete gamma function.

    When $-b'\cdot C_{2}^{-1}\cdot w_{0}\left(-\frac{C_{1}C_{2}}{ s}\right)<1$,
    \begin{align*}
        \Gamma_\mathrm{L}\left(a',-b'\cdot C_{2}^{-1}\cdot w_{0}\left(-\frac{C_{1}C_{2}}{ s}\right)\right)\leqslant  \frac{1}{a'}\cdot\left(-b'\cdot C_{2}^{-1}\cdot w_{0}\left(-\frac{C_{1}C_{2}}{ s}\right)\right)^{a'}
        \leqslant                                           \frac{(b' \cdot C_{1} \cdot e)^{a'}}{a'}\cdot s^{-a'},
    \end{align*}
    since $\Gamma_\mathrm{L}(t,x) \leqslant 1/t \cdot x^{t}$ if $x<1$ by \cite[(5)]{jameson2016incomplete}; 
    and $-w_0(-x) \leqslant e\cdot x$ when $-1<w_0(-x)<0$, since $1\leqslant e^{w_0(-x)+1}$ and $-w_0(-x)\leqslant -e \cdot w_0(-x)e^w_0(-x)$.
  When $-b'\cdot C_{2}^{-1}\cdot w_{-1}\left(-\frac{C_{1}C_{2}}{
    s}\right)>1$,
    \begin{align*}
        \Gamma_\mathrm{U}&\left(a',-b'\cdot C_{2}^{-1}\cdot w_{-1}\left(-\frac{C_{1}C_{2}}{ s}\right)\right)\\
        \leqslant &C_0 \exp{\left(b'\cdot C_2^{-1} \cdot w_{-1}\left(-\frac{C_{1}C_{2}}{ s}\right)\right)}\cdot \left(-b'\cdot C_{2}^{-1}\cdot w_{-1}\left(-\frac{C_{1}C_{2}}{ s}\right)\right)^{a'-1} \\
        \leqslant  & C'_0\left( C_{1}C_{2}\right)^{\frac{b'}{C_{2}}}\cdot s^{-\frac{b'}{C_{2}}},
    \end{align*}
    since $\Gamma_\mathrm{U}(t,x) \leqslant B\exp(-x) \cdot x^{t-1}$ if $t>1$, $B>1$ and $x>\frac{B(t-1)}{B-1}$, see \cite[(1.5)]{borwein2009uniform}
    and $w_{-1}(x) \leqslant \ln(-x)$ as $-1/e\leqslant x<0$, see \cite[(9)]{loczi2020explicit}. Here $C_0$ and $C'_0$ are constants depending on $C_1$, $C_2$ and $a'$.

    When $\sigma_{0}^{2}<\frac{b'}{2a'L'p}$, $s^{-a'}$ dominates, hence $\beta(s)
    \leqslant C_{\sigma_0}\cdot s^{-a'}$. Otherwise, $\beta(s) \leqslant C_{\sigma_0}
    \cdot s^{-\frac{b'}{C_{2}}}$, where $C_{\sigma_0}=\max\{(b' C_{1} \cdot e)^{a'}/a',C'_0\left( C_{1}C_{2}\right)^{b'/C_{2}}\}$. 

    By \cite[Lemma 15]{andrieu2022comparison},
    $K_{2}^{*}(v)=C(C_{\sigma_0},C_{2},a',b') \cdot v^{1+1/\min\{a',b'/C_2\}}$, where $C(C_{\sigma_0},C_{2},a',b')=\frac{C_{\sigma_0}c_1}{C_{\sigma_0}(1+c_1)^{1+c_1^{-1}}}$, $c_1= \min\{a',b'/C_2\}$.

    It is easy to see that $K_{0}^{*}(v)=\gamma v$, by \Cref{thm:weakPtotildeP2},
    \begin{align*}
        K^{*}(v) =K_{2}^{*}\circ \frac{1}{2}K_{0}^{*}(v)=C' v^{1+1/\min\{a',b'/C_2\}},
    \end{align*}
    where
    $C'=(\frac{1}{2}\gamma)^{1+1/\min\{a',b'/C_2\}}\cdot C(C_{\sigma_0},C_{2},a',
    b')$.

    Then, by \cite[Lemma 15]{andrieu2022comparison} and \Cref{thm:weakPtotildeP2},
    we have
    \begin{align*}
        F^{-1}(n) \leqslant \Tilde{C}\cdot (n-1)^{-\min\{a',b'/C_2\}},
    \end{align*}
    where $\Tilde{C}=(\frac{\min\{a',b'/C_{2}\}}{C'})^{\min\{a',b'/C_2\}}$.
\end{proof}

\begin{proof}[Proof of \Cref{thm:betaforIMHindiffusion}]
\label{prf:thm:betaforIMHindiffusion}
    As in \cite[(4)]{beskos2006retrospective}, under the condition that $b$ is everywhere differentiable, we can eliminate the It\^o integral after applying It\^o's lemma to $A\left(B_t\right)$ for $A(u):=\int_0^u b(x,\theta) \mathrm{d} x, u \in \mathbb{R}$. Then since $b^2(x,\theta) + b'(x,\theta)\geqslant M(\theta)$ for all $x$,
    \begin{align*}
        \mathcal{G}(X_{(t_{i-1},t_i)},\theta)&=\exp \left\{A\left(Y_{t_i}\right)-A(Y_{t_{i-1}})-\frac{1}{2} \int_{t_{i-1}}^{t_i}\left(b^2\left(X_t, \theta\right)+b^{\prime}\left(X_t,\theta\right)\right) \mathrm{d} t\right\} \\
        & \leqslant\exp \left\{A\left(Y_{t_i}\right)-A(Y_{t_{i-1}})-\frac{1}{2} \Delta t_i M(\theta)\right\}.
    \end{align*}
    
    Let $\tilde{\mathcal{G}}(X_{(t_{i-1},t_i)},\theta)=\exp \left\{A\left(Y_{t_i}\right)-A(Y_{t_{i-1}})-\frac{1}{2} \Delta t_i M(\theta)\right\}$.
    Since we update $X_{(t_{i-1},t_i)}$ in parallel by IMH, we operate on the tensor product space $\mathcal{X}=\prod_{i=1}^N\mathbb{R}^{(t_{i-1},t_i)}$ with the product measure $\prod_{i=1}^N \pi(\dif X_{(t_{i-1},t_i)}|\theta,Y)$, and the tensor product chain $H_{2|\theta}(\dif X'_{\mathrm{mis}}|X_{\mathrm{mis}})= \prod_{i=1}^N H^\theta_i(\dif X_{(t_{i-1},t_i)}'|X_{(t_{i-1},t_i)})$, where $H^\theta_i$ is the operator associated with the IMH kernel as in \eqref{eq:acceptance for diffusion} for $X_{(t_{i-1},t_i)}$ with the invariant distribution $\pi(\dif X_{(t_{i-1},t_i)}|\theta,Y)$. 

    Furthermore, we denote $\pi(\dif X_{(t_{i-1},t_i)}|\theta,Y)$ by $\pi_i(\cdot|\theta)$. By \cite[Proposition 28, Remark 29]{andrieu2022comparison}, for all $i=1, \ldots, N$, the chain $H^\theta_i$ satisfies a WPI with function $\beta_{2,i}(s,\theta)$: for all $g_i \in \mathrm{L}^2_0(\pi_i(\cdot|\theta))$, 
    \begin{align*}
        \mathrm{Var}_{\pi_i(\cdot|\theta)}(g_i) &\leqslant s \cdot \mathcal{E}_{\pi_i(\cdot|\theta)}(H^{\theta}_i,g_i)+ \beta_{2,i}(s,\theta) \cdot \|g_i\|_\mathrm{osc}^2,
    \end{align*}
    where $\beta_{2,i}(s,\theta)=\mathbbm{1}_{\{s\leqslant \tilde{\mathcal{G}}(X_{(t_{i-1},t_i)},\theta)\}}$.

    By \Cref{rmk:strongtoweak}, we have for all $i=1,\ldots, N$, $H^{\theta}_i$ satisfies a standard Poincar\'e inequality: \begin{align*}
        \tilde{\mathcal{G}}(X_{(t_{i-1},t_i)},\theta) ^{-1} \cdot \mathrm{Var}_{\pi_i(\cdot|\theta)}(g_i) &\leqslant \mathcal{E}_{\pi_i(\cdot|\theta)}(H^{\theta}_i,g_i),
    \end{align*}
    Then, by \Cref{Thm:SPIproduct}, we can conclude that $H_{2|\theta}$ satisfies a SPI:  for all $g \in \mathrm{L}^2_0(\pi(\cdot|\theta))$,
    \begin{align*}
        \min_i\{        \tilde{\mathcal{G}}(X_{(t_{i-1},t_i)},\theta) ^{-1}\} \cdot \|g\|_{\pi(\cdot|\theta)}^2\leqslant \mathcal{E}_{\pi(\cdot|\theta)}(H_{2|\theta},g).
    \end{align*}
    Then, applying \Cref{rmk:strongtoweak}, we can conclude that $H_{2|\theta}$ satisfies a  WPI with function $\beta_2(s,\theta)=\mathbbm{1}_{\{s\leqslant \max_i{\tilde{\mathcal{G}}(X_{(t_{i-1},t_i)},\theta)}\}}$: for all $g \in \mathrm{L}^2_0(\pi(\cdot|\theta))$,
    \begin{align*}
       \|g\|_{\pi(\cdot|\theta)}^2 &\leqslant s \cdot \mathcal{E}_{\pi(\cdot|\theta)}(H_{2|\theta},g)+ \beta_2(s, \theta) \cdot \|g\|_\mathrm{osc}^2.
    \end{align*}
\end{proof}

\begin{proof}[Proof of \Cref{thm:SPI_for_diff}]
\label{Prf:thm:SPI_for_diff}
Firstly, we claim there exist $\epsilon^i>0$ and $K>0$, such that $ \mathbb{P}_\theta(C^i_K)\leqslant\epsilon^i$, where $C_K^i=\{X_{(t_{i-1},t_i)}:\max_{t_{i-1\leqslant s \leqslant t_i}} X_s \geqslant K\}$, for $i=1,\ldots,N$.

Since $\mathcal{G}(X_{(t_{i-1},t_i)},\theta)=\exp \left\{A\left(Y_{t_i}\right)-A(Y_{t_{i-1}})-\frac{1}{2} \int_{t_{i-1}}^{t_i}\left(b^2\left(X_t, \theta\right)+b^{\prime}\left(X_t,\theta\right)\right) \mathrm{d} t\right\}$, denote $\mathcal{\tilde{G}}(X_{(t_{i-1},t_i)},\theta)=\exp \left\{-\frac{1}{2} \int_{t_{i-1}}^{t_i}\left(b^2\left(X_t, \theta\right)+b^{\prime}\left(X_t,\theta\right)\right) \mathrm{d} t\right\}$. Hence,
\begin{align*}
    \mathbb{P}_\theta(C_K^i)=\frac{\mathbb{E}_{{\mathbb{B}_\theta}}(\mathcal{G}\mathbbm{1}_{C_K^i})}{\mathbb{E}_{\mathbb{B}_\theta}(\mathcal{G})}=\frac{\mathbb{E}_{{\mathbb{B}_\theta}}(\mathcal{\tilde{G}}\mathbbm{1}_{C_K^i})}{\mathbb{E}_{\mathbb{B}_\theta}(\mathcal{\tilde{G}})}\leqslant \frac{\exp{\{-\frac{1}{2}M_L \Delta t_i\}} \mathbb{B}_\theta(C_K^i)}{\mathbb{E}_{\mathbb{B}_\theta}(\mathcal{\tilde{G}})},
\end{align*}
where $\Delta t_i=t_i-t_{i-1}$. For the denominator,
$ \mathbb{E}_{\mathbb{B}_\theta}(\mathcal{\tilde{G}}) \geqslant\mathbb{E}_{\mathbb{B}_\theta}(\mathcal{\tilde{G}} \mathbbm{1}_{\{\max_{t_{i-1\leqslant s \leqslant t_i}}|X_s|\leqslant a\}}) \geqslant \exp{\{-\frac{1}{2}M_U \Delta t_i\}}\mathbb{B}_\theta(\max_{t_{i-1\leqslant s \leqslant t_i}}|X_s|\leqslant a)$. Also there exists a constant $c>0$ such that $\mathbb{B}_\theta(\max_{t_{i-1\leqslant s \leqslant t_i}}|X_s|\leqslant a)\geqslant c$. Therefore,
\begin{align*}
    \mathbb{P}_\theta(C_K^i)
    &\leqslant \frac{\exp{\{-\frac{1}{2}M_L \Delta t_i\}} \mathbb{B}_\theta(C_K^i)}{\exp{\{-\frac{1}{2}M_U \Delta t_i\}}\mathbb{B}_\theta(\max_{t_{i-1\leqslant s \leqslant t_i}}|x_s|\leqslant a)}\\
    &=\exp\left\{\frac{1}{2}(M_U-M_L) \Delta t_i\right\}\frac{\mathbb{B}_\theta(C_K^i)}{\mathbb{B}_\theta(\max_{t_{i-1\leqslant s \leqslant t_i}}|X_s|\leqslant a)}\\
    &\leqslant \exp\left\{\frac{1}{2}(M_U-M_L) \Delta t_i\right\}\frac{\mathbb{B}_\theta(C_K^i)}{c}=: \epsilon^i.
\end{align*}
Now, choose $K$ large enough such that $\epsilon^i <1$. Thus the claim is proven.

Hence, it is easy to see there exist $\epsilon>0$ and $K>0$, such that $\mathbb{P}_\theta(C_K)\leqslant\epsilon,$ where $C_K=\{X:\max_{0\leqslant s \leqslant T} X_s \geqslant K\}$, since given $Y$, the path segments $\{X_{(t_{i-1},t_i)}\}_{i=1}^N$ are conditionally independent.

Next, we claim there exist $\eta >0$
and a probability measure $\nu(\cdot)$ such that for the $\theta$ marginal transition kernel $P_\Theta(\theta, \cdot)$ and all $\theta \in \mathbb{R}$, we have $P_\Theta(\theta, \cdot) \geqslant \eta\nu(\cdot)$. Since
\begin{align*}
    P_\Theta(\theta, A)=\int_{\mathcal{X}} \pi(A|X) \mathbb{P}_\theta(\dif X) \geqslant \int_{C_K^\complement } \pi(A|X) \mathbb{P}_\theta(\dif X),
\end{align*}
let $\nu(A)= \frac{1}{Z}\int_{C_K^\complement} \pi(A|X) \mathbb{P}_\theta(\dif X)$ be a probability measure, where $Z=\mathbb{P}_\theta(C_K^\complement)\geqslant 1-\epsilon$. Hence, there exist $\eta>0$, such that for any measuable set $A$, $P_\Theta(\theta, A)\geqslant Z \nu(A)\geqslant \eta \nu(A)$.

Finally, by \cite[Theorem 8]{roberts2004general} and  \Cref{thm:minorisation condition for diffusion}, 
    the marginal chain of $\theta$, $P_{\Theta}$, is uniformly ergodic. Then by \Cref{thm:Px_L1_to_P_SPI},
    there exists a constant $\gamma>0$, such that $\gamma \|f\|_\pi^{2}\leqslant \mathcal{E}_\pi(P^{*}P,f)$.
\end{proof}

\begin{lemma}\label{lem:OU_probs}
    In the setting of the discretely-observed Ornstein--Uhlenbeck process in \eqref{eq:ou}, there exists a cumulative distribution function $F(\cdot)$ such that
    \begin{align*}
    \mathbb{P}_\theta\left(\int_{0}^{T}X_{t}^{2} \dif t\leqslant u \right) \geqslant F(u),
    \end{align*}
    where $F(\cdot)$ does not depend on $\theta$.
\end{lemma}
\begin{proof}
    To begin with, suppose $X_{t}$ is an Ornstein--Uhlenbeck bridge with $X_{0}=a$, $X_{T}=b$.
    The conditioned process $X_{t}$ satisfies the stochastic differential equation
    \begin{align*}
        \dif X_{t}=-\theta X_{t} \dif t-2\theta \cdot \frac{X_{t}-b e^{\theta(T-t)}}{e^{2\theta(T-t)}-1}\dif t+\dif W_{t}.
    \end{align*}
    We denote
    $\alpha(\theta,t)=-\theta \left(1+\frac{2}{e^{2\theta(T-t)}-1}\right)$, $\beta
    (\theta,t)=\frac{2b\theta e^{\theta(T-t)}}{e^{2\theta(T-t)}-1}$, then the
    stochastic differential equation can be written as
    \begin{align*}
        \dif X_{t}=\dif W_{t}+\left\{\alpha(\theta,t)X_{t} +\beta(\theta,t)\right\}\dif t.
    \end{align*}
    It is easy to check $\alpha(\theta,t) \leqslant-\frac{1}{T-t}$ and  
        $-\frac{|b|}{T-t}\leqslant  \beta(\theta,t)\leqslant \frac{|b|}{T-t}$.
    Now, we compute $\mathbb{E}X_{t}$ and $\mathbb{E}X_{t}^{2}$ aiming for finding
    a uniform upper bound of $\mathbb{E}X_{t}^{2}$. Since $X_{t}$ can be written
    as
    \begin{align*}
        X_{t} = \int_{0}^{t}\alpha(\theta,s)X_{s}\dif s+\int_{0}^{t}\beta(\theta,s)\dif s +W_{t},
    \end{align*}
    where $W_{t}$ is a Brownian motion. We take the expectation then take the derivative
    of $t$,
    \begin{align*}
        \begin{cases}\frac{\dif \mathbb{E}X_t}{\dif t}=\alpha(\theta,t)\mathbb{E}X_{t}+\beta(\theta,t),\\ \mathbb{E}X_{0}=a.\end{cases}
    \end{align*}
    Hence, we can obtain the solution $\mathbb{E}X_{t}$:
    \begin{align*}
    \mathbb{E}X_{t}=\exp\left(\int_{0}^{t}\alpha(\theta,s)\dif s\right)\int_{0}^{t} a+\beta(\theta,s)\exp\left(-\int_{0}^{s}\alpha(\theta,r)\dif r\right)\dif s.
    \end{align*}
    Since $\alpha(\theta,t)$ and $\beta(\theta,t)$ both can be bounded from above,
    \begin{align*}
        \mathbb{E}X_{t} & =a\exp\left(\int_{0}^{t}\alpha(\theta,s)\dif s\right) +\int_{0}^{t} \beta(\theta,s)\exp\left(\int_{s}^{t}\alpha(\theta,r)\dif r\right)\dif s       \\
                        & \leqslant |a|\exp\left(\int_{0}^{t}-\frac{1}{T-s}\dif s\right)+\int_{0}^{t} \beta(\theta,s)\exp\left(\int_{s}^{t}-\frac{1}{T-r}\dif r\right)\dif s \\
                        & \leqslant |a|\cdot\frac{T-t}{T}+\int_{0}^{t} \beta(\theta,s)\frac{T-t}{T-s}\dif s                                                                  \\
                        & \leqslant |a|\cdot\frac{T-t}{T}+\int_{0}^{t} \frac{|b|}{T-s}\cdot\frac{T-t}{T-s}\dif s                                                             \\
                        & =|a|\cdot\frac{T-t}{T}+|b|\cdot\frac{t}{T}.
    \end{align*}
    Similarly, starting from $\dif X_{t}^{2}$, we have
    \begin{align*}
        \dif X_{t}^{2} =2X_{t}\dif X_{t}+\dif \langle X_{t}\rangle =\left(2\alpha(\theta,t)X_{t}^{2}+2\beta(\theta,t)X_{t}+1\right)\dif t+2X_{t}\dif W_{t}.
    \end{align*}
    Hence,
    \begin{align*}
        X_{t}^{2}=\int_{0}^{t}\left(2\alpha(\theta,s)X_{s}^{2}+2\beta(\theta,s)X_{s}+1\right)\dif s+2\int_{0}^{t}X_{s}\dif W_{s}.
    \end{align*}
    When we take the expectation, $\mathbb{E}\int_{0}^{t}X_{s}\dif W_{s}=0$ as $\int_{0}
    ^{t}X_{s}\dif W_{s}$ is a martingale. Then taking the derivative of $t$, we have
    \begin{align*}
        \begin{cases}\frac{\dif \mathbb{E}X_t^2}{\dif t}=2\alpha(\theta,t)\mathbb{E}X_{t}^{2}+2\beta(\theta,t)\mathbb EX_{t}+1,\\ \mathbb{E}X_{0}^{2}=a^{2}.\end{cases}
    \end{align*}
    Denote $\gamma(\theta,t)=2\beta(\theta,t)\mathbb{E}X_{t}+1$, the solution of
    $\mathbb{E}X_{t}^{2}$ is
    \begin{align*}
        \mathbb{E}X_{t}^{2} & =\exp\left(\int_{0}^{t}2\alpha(\theta,s)\dif s\right)\int_{0}^{t} a^{2}+\gamma(\theta,s)\exp\left(-\int_{0}^{s}2\alpha(\theta,r)\dif r\right)\dif s                                \\
                            & =a^{2}\exp\left(\int_{0}^{t}2\alpha(\theta,s)\dif s\right)+\int_{0}^{t} \gamma(\theta,s)\exp\left(\int_{s}^{t}2\alpha(\theta,r)\dif r\right)\dif s                                 \\
                            & \leqslant a^{2}\exp\left(\int_{0}^{t}-\frac{2}{T-s}\dif s\right)+\int_{0}^{t} \gamma(\theta,s)\exp\left(\int_{s}^{t}-\frac{2}{T-r}\dif r\right)\dif s                              \\
                            & \leqslant a^{2} \cdot\frac{(T-t)^{2}}{T^{2}}+\int_{0}^{t} \gamma(\theta,s)\left(\frac{T-t}{T-s}\right)^{2}\dif s                                                                   \\
                            & \leqslant a^{2} \cdot\frac{(T-t)^{2}}{T^{2}}+\int_{0}^{t} \left(2\frac{|b|}{T-s}\left(\frac{|a|(T-s)}{T}+\frac{|b|s}{T}\right)+1\right)\cdot\left(\frac{T-t}{T-s}\right)^{2}\dif s \\
                            & =\left(\frac{|a|(T-t)}{T}+\frac{|b|t}{T}\right)^{2}+t\left(1-\frac{t}{T}\right).
    \end{align*}
    Then,
    \begin{align*}
        \int_{0}^{T}\mathbb{E}X_{t}^{2}\dif t\leqslant\frac{1}{6}T\left(2|ab|+2a^{2}+2b^{2}+T\right).
    \end{align*}
    Hence, by Markov's inequality,
    \begin{align*}
        \mathbb{P}\left(\int_{0}^{T}X_{t}^{2} \dif t\geqslant u \right) \leqslant \frac{\mathbb{E}\int_{0}^{T}X_{t}^{2}\dif t}{u}=\frac{\int_{0}^{T}\mathbb{E}X_{t}^{2}\dif t}{u}\leqslant \frac{\frac{1}{6}T\left(2|ab|+2a^{2}+2b^{2}+T\right)}{u}.
    \end{align*}
    When $u\geqslant\frac{1}{6}T\left(2|ab|+2a^{2}+2b^{2}+T\right)$,
    \begin{align*}
        \mathbb{P}\left(\int_{0}^{T}X_{t}^{2} \dif t\leqslant u \right) \geqslant 1- \frac{\frac{1}{6}T\left(2|ab|+2a^{2}+2b^{2}+T\right)}{u}.
    \end{align*}

    Since $X_{t}$ from $0$ to $T$ can be divided into $N$ conditionally independent pieces
    starting from $X_{t_{i-1}}=Y_{t_{i-1}}$ and ending at $X_{t_{i}}=Y_{t_i}$, 
    given the parameter $\theta$ and obeserved points $Y_{t_{i}}$, $\Delta t_i=t_i-t_{i-1}$, $i=1,\ldots,N$,
    \begin{align*}
        \mathbb{E}\int_{0}^{T} X_{t}^{2}\dif t & =\int_{0}^{T}\mathbb{E}X_{t}^{2} \dif t=\sum^{N}_{i=1}\int_{t_{i-1}}^{t_{i}}\mathbb{E}X_{t}^{2} \dif t =\sum^{N}_{i=1}\int_{0}^{\Delta t_{i}}\mathbb{E}X_{t}^{2} \dif t
        \\& \leqslant \sum^{N}_{i=1}\frac{1}{6}\Delta t_{i}\left(2|Y_{t_{i-1}}Y_{t_{i}}|+2Y_{t_{i-1}}^{2}+2Y_{t_{i}}^{2}+\Delta t_{i}\right).
    \end{align*}
    Then,
    \begin{align*}
        \mathbb{P}\left(\int_{0}^{T}X_{t}^{2} \dif t\leqslant u \right) \geqslant 1- \frac{\sum^{N}_{i=1}\frac{1}{6}\Delta t_{i}\left(2|Y_{t_{i-1}}Y_{t_{i}}|+2Y_{t_{i-1}}^{2}+2Y_{t_{i}}^{2}+\Delta t_{i}\right)}{u}.
    \end{align*}
\end{proof}

\begin{lemma}
    \label{thm:minorisation condition for diffusion} In the same setting as Lemma~\ref{lem:OU_probs}, there exists $\epsilon >0$
    and a probability measure $\nu(\cdot)$ such that for the $\theta$-marginal transition kernel $P(\theta, \cdot)$, we have
    \begin{align*}
        P(\theta, \cdot) \geqslant \epsilon\nu(\cdot) \quad \theta \in \mathbb{R}.
    \end{align*}
\end{lemma}
\begin{proof}
    Let $p\cdot u_{0}=\sum^{N}_{i=1}\frac{1}{6}\Delta t_{i}\left(2|Y_{t_{i-1}}Y_{t_{i}}
    |+2Y_{t_{i-1}}^{2}+2Y_{t_{i}}^{2}+\Delta t_{i}\right)$, where $p\in(0,1)$, then $\mathbb{P}\left(\int_{0}^{T}X_{t}^{2} \dif t\leqslant u_{0} \right) \geqslant 1-p$.
    So for any $\theta$, by Lemma~\ref{lem:OU_probs},
    \begin{align*}
        \mathbb{P}\left(\frac{1}{u_{0}+\tau_{0}^{-2}}\leqslant\frac{1}{\int_{0}^{T} X_{t}^{2} \dif t +\tau_{0}^{-2}}\leqslant\frac{1}{\tau_{0}^{-2}}\right)\geqslant1-p.
    \end{align*}
    Let $\mathcal{F}=\left\{\sigma^2:  \frac{1}{u_{0}+\tau_{0}^{-2}}\leqslant \sigma
    ^{2}\leqslant\frac{1}{\tau_{0}^{-2}}\right\}$, 
    and define the function $\psi(\theta)=\inf_{\sigma^2\in\mathcal{F}} \phi(\theta;\mu\sigma^2,\sigma^2)$, where $\mu
    =-\int_{0}^{T}X_{t} \dif X_{t}+\mu_{0}\tau_{0}^{-2}$, $ \sigma^{2}=\frac{1}{\int_{0}^{T}
    X_{t}^{2} \dif t +\tau_{0}^{-2}}$ and $\phi(\cdot;\mu\sigma^2,\sigma^2)$ is the normal density with mean $\mu\sigma^2$ and variance $\sigma^2$. Since $\phi $ is positive and continuous in all arguments, by compactness of $\mathcal{F}$,
    $\psi >0 $ for all $\theta $, so that $Z=\int \psi(\theta) \dif \theta >0$.
    Hence we can define a probability measure, for measurable $A$, by $\tilde{\nu}(A)=\frac{1}{Z}\int_A \psi(\theta) \dif \theta$.
    
    Since $\theta | X_{\mathrm{mis}}, Y \sim \mathcal N\left(\mu\sigma^2,\sigma^2\right)$, and $\mathbb{P}(\sigma^2 \in \mathcal{F}) \geqslant 1-p$, then for
    all $\theta\in \mathbb{R}$, and all measurable $A$,
    \begin{align*}
        P(\theta,A)\geqslant (1-p) \cdot Z \cdot \tilde{\nu}(A).
    \end{align*}
    We consider
    \begin{align*}
        \epsilon=\inf_{\sigma^2}\int\min\left\{f_{\{\mu\sigma_1^2,\sigma_1^2\}}(x),f_{\{\mu\sigma_2^2,\sigma_2^2\}}(x)\right\} \dif x,
    \end{align*}
    by Proposition 5 in \cite{roberts2001small}, $\epsilon=Z$. Now,
    we consider the total variation distance between
    $\mathcal{N}(\mu\sigma_{1}^{2},\sigma_{1}^{2})$ and
    $\mathcal{N}(\mu\sigma_{2}^{2},\sigma_{2}^{2})$ for any
    $\sigma_{1}^{2}\leqslant \sigma_{2}^{2}$. Letting $\Phi$ be the cumulative distribution function of the standard normal distribution, we have 
    \begin{align*}
        \operatorname{TV}\left(\mathcal{N}(\mu \sigma^{2}_{1}, \sigma_{1}^{2}), \mathcal{N}(\mu\sigma^{2}_{2}, \sigma_{2}^{2})\right)&=\Phi\left(\frac{c_{1}-\mu\sigma^{2}_{1}}{\sigma_{1}}\right)+1-\Phi\left(\frac{c_{2}-\mu\sigma^{2}_{1}}{\sigma_{1}}\right)\\
        &+\Phi\left(\frac{c_{2}-\mu\sigma^{2}_{2}}{\sigma_{2}}\right)-\Phi\left(\frac{c_{1}-\mu\sigma^{2}_{2}}{\sigma_{2}}\right),
    \end{align*}
    where $c_{1}<c_{2}$, and
    \begin{align*}
        (c_{1},c_{2})=\pm\frac{\sigma_{1}\sigma_{2}\sqrt{(\mu\sigma_1^2-\mu\sigma_2^2)^2+(\sigma_2^2-\sigma_1^2)\ln\left(\frac{\sigma_2^2}{\sigma_1^2}\right)}}{\sigma_{2}^{2}-\sigma_{1}^{2}}.
    \end{align*}
    Denote $f(\sigma_{a}^{2},\sigma_{b}^{2})=\sup_{\sigma^2}\operatorname{TV}\left(\mathcal{N}\left(\mu_{1}, \sigma_{1}^{2}\right), \mathcal{N}\left(\mu_{2}, \sigma_{2}^{2}\right)\right)$, 
    where $\sigma_{a}^{2}=\frac{1}{u_{0}+\tau_{0}^{-2}}$ and
    $\sigma_{b}^{2}=\frac{1}{\tau_{0}^{-2}}$. Hence,
    \begin{align*}
        \epsilon & = \inf_{\sigma^2}\int\min\left\{f_{\{\mu\sigma_1^2,\sigma_1^2\}}(x),f_{\{\mu\sigma_2^2,\sigma_2^2\}}(x)\right\}                                                \\
                 & = \inf_{\sigma^2}\left(1-\operatorname{TV}\left(\mathcal{N}\left(\mu_{1}, \sigma_{1}^{2}\right), \mathcal{N}\left(\mu_{2}, \sigma_{2}^{2}\right)\right)\right) = 1-f(\sigma_{a}^{2},\sigma_{b}^{2}).
    \end{align*}
    We can choose $u_{0}$ such that $0<1-p<1$, and
    $0<1-f(\sigma_{a}^{2},\sigma_{b}^{2})<1$ holds naturally. Hence, $P(\theta,A)\geqslant\left(1-p\right)\left(1-f(\sigma_{a}^{2},\sigma_{b}^{2})\right)\tilde{\nu}(A)$.
\end{proof}

\begin{proof}[Proof of \Cref{eg: ou_SPI}]
    \label{Prf:eg: ou_SPI} By \cite[Theorem 8]{roberts2004general} and  \Cref{thm:minorisation condition for diffusion}, 
    the marginal chain of $\theta$, $P(\theta,\cdot)$, is uniformly ergodic. Then by \Cref{thm:Px_L1_to_P_SPI},
    there exists a constant $\gamma>0$, such that $\gamma \|f\|_\pi^{2}\leqslant \mathcal{E}_\pi(P^{*}P,f)$.
\end{proof}

\begin{lemma}
    \label{thm:K*forexplog} Assume $\beta(s)=c \exp(-(a\log(s)+b)^{2})$, for
    $s\geqslant \exp{(-b/a)}$, and $\beta(s)=c$, for
    $0\leqslant s < \exp{(-b/a)}$, where $c>0$. Then there exist $C>0$, $0 < v_{0}
    < 1$, such that for $v \in [0,v_{0}]$,
    \begin{align*}
        K^{*}(v) \geqslant C v\cdot \exp\left(\frac{-\sqrt{-\delta\log{v}}+b}{a}\right).
    \end{align*}
    In addition, there exists $\tilde{C}> 0$ such that for any $\delta>1$ and
    all $n \in \mathbb{N}_{+}$,
    \begin{align*}
        F^{-1}(n) \leqslant \Tilde{C}\exp\left(-\frac{a^{2}}{\delta}\log^{2}\left(n\right)\right).
    \end{align*}
\end{lemma}
\begin{proof}
    Since $\beta(s) \rightarrow 0$ as $s\geqslant 0$ and $s \rightarrow \infty$ monotonically,
    when $0<v<1$, it holds that $v-\beta(1/u)\leqslant 0$ for $u\geqslant 1$, so
    the maximization in $K^{*}$ can be restricted to $u < 1$:
    \begin{align*}
        K^{*}(v)=\sup_{u\in[0,1)}\left\{uv-uc_{1} \exp(-(-a\log{u}+b)^{2})\right\}.
    \end{align*}
    For some $\delta>1$ and $v \in (0,1)$, take $u_{v} :=\exp(\frac{-\sqrt{-\delta\log{v}}+b}{a}
    )$, that is such that $v=\exp(-\frac{1}{\delta}(-a\log{u}+b)^{2})<1$. As a result
    we have a lower bound on $K^{*}$ as
    \begin{align*}
        K^{*}(v) & \geqslant v\cdot \exp\left(\frac{-\sqrt{-\delta\log{v}}+b}{a}\right)- c\exp\left(\frac{-\sqrt{-\delta\log{v}}+b}{a}\right) v^{\delta}.
    \end{align*}
    Provided that $\delta>1$, the second term will decay faster as
    $v\rightarrow 0$ and the first statement follows.
    
    For $0 \leqslant w \leqslant v_{0}$, let $M:=\int_{v_0}^{1} \frac{\dif v}{K^{*}(v)}$.
    Thus, we can bound
    \begin{align*}
        F(w) & =\int_{w}^{1} \frac{\dif v}{K^{*}(v)}\\
             & =\int_{w}^{v_0}\frac{\dif v}{K^{*}(v)}+\int_{v_0}^{1} \frac{\dif v}{K^{*}(v)}\\
             & \leqslant C^{-1}\int_{w}^{v_0}v^{-1}\cdot \exp\left(\frac{\sqrt{-\delta\log{v}}-b}{a}\right) \dif v+M \\
             & =C^{-1}\left.\left[-\frac{2a}{\delta}\exp\left(\frac{\sqrt{-\delta\log(v)}-b}{a}\right)\left(\sqrt{-\delta\log{v}}-a\right)\right]\right|_{w}^{v_0}+M'\\
             & \leqslant C^{-1}\cdot\frac{2a^{2}}{\delta}\exp{\left(-\frac{b}{a}\right)}\cdot \frac{\sqrt{-\delta\log{w}}}{a}\exp\left(\frac{\sqrt{-\delta\log{w}}}{a}\right)+M''.
    \end{align*}
    This lead for $n\geqslant M''$ to
    \begin{align*}
        F^{-1}(n) \leqslant \exp\left(-\frac{a^{2}}{\delta}W_{0}^{2}\left(\frac{n-M''}{C'}\right)\right),
    \end{align*}
    where $W(\cdot)$ is the $w_{0}$ branch of Lambert W function and $C'=C^{-1}\cdot\frac{2a^{2}}{\delta}
    \exp{\left(-\frac{b}{a}\right)}$. For $n > M''$,
    \begin{align*}
        F^{-1}(n) \leqslant \exp\left(-\frac{a^{2}}{\delta}\log^{2}\left(\frac{n-M''}{C'}\right)\right).
    \end{align*}
    For $n > M''$ , by concavity of $[0, M] \ni x \rightarrow \log(n-M'')$, we
    have $\log(n)-M''\cdot \frac{1}{n-M''}\leqslant \log(n-M'')$ and
    $\frac{M''}{n-M''}\leqslant \frac{M''}{\lfloor M''\rfloor+1-M''}$, so there exists
    $C'' > 0$ such that for $n/C'\geqslant1$, for any $\delta>1$,
    \begin{align*}
        F^{-1}(n) \leqslant C''\exp\left(-\frac{a^{2}}{\delta}\log^{2}\left(\frac{ n}{C'}\right)\right).
    \end{align*}
    Then, since $\exp{(-\log^2(x))}$ decays faster than any polynomial, we conclude that there exists $\Tilde{C}> 0$ such that for $n \in \mathbb{N}_{+}$, for any
    $\delta'>\delta>1$,
    \begin{align*}
        F^{-1}(n) \leqslant \Tilde{C}\exp\left(-\frac{a^{2}}{\delta'}\log^{2}\left(n\right)\right).
    \end{align*}
\end{proof}
\begin{proof}[Proof of \Cref{eg: rate for ou}]
    \label{Prf:eg: rate for ou} Since
    \begin{align*}
         Y_{t_{i+1}}|Y_{t_i}=y_{i} &\sim \mathcal N(y_{i}m_{i},v_{i}),  \quad m_{i}=\exp(-\theta\Delta t_{i}), \quad v_{i}=\frac{1-\exp(-2\theta \Delta t_{i})}{2\theta},\\
         Y_{t_0}&\sim \mathcal N\left(0,\frac{1}{2\theta}\right), \quad \Delta t_{i} =t_{i+1}-t_{i}, i=0,\ldots,N-1,
    \end{align*}
    the distribution of $\pi(\theta|Y)$ can be obtained by Bayes' rule. We have
    \begin{align*}
        \log\pi(\theta|Y) & \propto \sum_{i=0}^{N-1}\log p(Y_{t_{i+1}}|Y_{t_i}, \theta) +\log p(Y_{t_0}|\theta)+ \log p(\theta)                                                                             \\
                          & \propto -\frac{1}{2}\sum_{i=0}^{N-1}\left[\log(v_{i})+\frac{(Y_{t_{i+1}}-m_{i})^{2}}{v_{i}}\right]-\frac{1}{2}\theta Y_{t_0}^2+\frac{1}{2}\log \theta -\frac{1}{2}\frac{(\theta-\mu_{0})^{2}}{\tau_{0}^{2}}.
    \end{align*}
    Let $l(\theta)=\log(v_{i})+\Delta t_{i} \theta-\log\Delta t_i$, since $l'(\theta)>0$ as $\theta >0$ and $l'(\theta)<0$ as $\theta <0$, $l(\theta)\geqslant l(0)=0$.
    Therefore,
    \begin{align*}
        -\frac{1}{2}\sum_{i=0}^{N-1}\left[\log(v_{i})+\frac{(Y_{t_{i+1}}-m_{i})^{2}}{v_{i}}\right]\leqslant -\frac{1}{2}\sum_{i=0}^{N-1}\log(v_{i})&\leqslant \frac{1}{2}\sum_{i=0}^{N-1}\Delta t_{i} \theta-\frac{1}{2}\sum_{i=0}^{N-1}\log \Delta t_i\\
        &=\frac{1}{2}T\theta-\frac{1}{2}\sum_{i=0}^{N-1}\log \Delta t_i,
    \end{align*}
    then, there exists a constant $M$ depend on $Y_{t_i}$, $i=0,\ldots,N $,  such that
    \begin{align*}
        \log\pi(\theta|Y) & \leqslant \frac{1}{2}T\theta -\frac{1}{2}\theta Y_{t_0}^2+\frac{1}{2}\log \theta-\frac{1}{2}\frac{(\theta-\mu_{0})^{2}}{\tau_{0}^{2}}+M \\
        &\leqslant \frac{1}{2}(T+1)\theta-\frac{1}{2}\frac{(\theta-\mu_{0})^{2}}{\tau_{0}^{2}}+M \leqslant -\frac{1}{2}\frac{(\theta-m)^{2}}{\tau_{0}^{2}}+M',
    \end{align*}
    where $m=\mu_{0}+\frac{\tau_{0}^{2} (T+1)}{2}$,
    $M' = M - \frac{\mu_{0}^{2}}{2\tau_{0}^{2}}+ \frac{m^{2}}{2\tau_{0}^{2}}$. Therefore,
    there exists a finite positive number $K$, such that $\pi(\theta|Y) \leqslant
    K \cdot q(\theta)$, where $q(\theta) \sim \mathcal{N}(m,\tau_{0}^{2})$, $K=\sqrt{2\pi}
    \tau_{0}\exp(M')$.

    In this scheme, since the function $b^2(x,\theta)+b'(x,\theta)=-\theta + \theta^{2} x^{2}$ is bounded below
    by $-\theta$ for every $\theta$, by \Cref{thm:betaforIMHindiffusion}, we have
    $\beta_{2}(\theta,s)=\mathbbm{1}_{{\{s \leqslant \max_i \{\tilde{\mathcal{G}}_i(X)\}\}}}$,
    where
    $\tilde{\mathcal{G}}_{i}(X)=\exp\left\{\frac{1}{2}\theta \left(\Delta t_{i}-Y
    _{t_i}^{2}+Y_{t_{i-1}}^{2}\right)\right\}$. Hence,
    \begin{align*}
        \int \beta_{2}(\theta,s)\pi(\dif \theta) & \leqslant K\int \mathbbm{1}_{{\{s \leqslant \max_i\{ \tilde{\mathcal{G}}_i(X)\}\}}}q(\dif \theta) =K \cdot\left(1-\Phi\left(\frac{2\log(s)/\eta -m}{\tau_{0}}\right)\right),
    \end{align*}
    where $\eta=\max_{i}\{\Delta t_{i}-Y_{t_i}^{2}+Y_{t_{i-1}}^{2}\}$. Since $1-\Phi
    (x)\leqslant\frac{1}{2}\exp{(-x^2/2)}$ for $x \geqslant 0$, as
    $\frac{2\log(s)/\eta -m}{\tau_{0}}\geqslant 0$, we have
    \begin{align*}
        \int \beta_{2}(\theta,s)\pi(\dif \theta)\leqslant \frac{K}{2}\exp{\left(-\frac{1}{2}\left(\frac{2}{\eta\tau_{0}}\log{x}-\frac{m}{\tau_{0}}\right)^2\right)}.
    \end{align*}
    We take $\beta_{2}(s)=\frac{1}{4}$ when $0<s<1$,
    $\beta_{2}(s)=\max\{\frac{1}{4},\frac{K'}{2}\}\exp{(-\frac{2}{\eta^{2} \tau_{0}^{2} }\log^2s)}
    \geqslant \int \beta_{2}(\theta,s)\pi(\dif \theta)$
    when $s\geqslant 1$. By \Cref{thm:weakPtotildeP2} and \Cref{thm:K*forexplog}, we have $K^{*}(v) =K_{2}^{*}\left(\frac{\gamma \cdot v}{2}\right)$.
    
    Then by \Cref{thm:K*forexplog} and Lemma 14 in
    \cite{ascolani2024scalability}, we have
    \begin{align*}
        F^{-1}(n) \leqslant \Tilde{C}\exp\left(-\frac{a}{\delta}\log^{2}\left(\frac{n-1}{\gamma/2}\right)\right),
    \end{align*}
    where $\Tilde{C}$ is a constant, $a=\frac{2}{\eta^{2} \tau_{0}^{2}}$, $\delta$
    is any constant greater than $1$.
\end{proof} 

\bibliographystyle{abbrv}
\bibliography{ref}

\end{document}